\documentclass[a4paper]{amsart}

\usepackage[english]{babel}
\usepackage[utf8x]{inputenc}
\usepackage[T1]{fontenc}
\usepackage{amsmath}
\usepackage{amssymb}
\usepackage{amsthm}
\usepackage{scrextend}
\usepackage{listings}
\usepackage{xcolor}
\usepackage{booktabs}
\allowdisplaybreaks

\usepackage[a4paper,top=3cm,bottom=2cm,left=3cm,right=3cm,marginparwidth=1.75cm]{geometry}

\usepackage{enumerate}
\usepackage{amsmath}
\usepackage{graphicx}
\usepackage[colorinlistoftodos]{todonotes}
\usepackage[colorlinks=true, allcolors=blue]{hyperref}
\usepackage{mathtools}

\newtheorem{thm}{Theorem}
\newtheorem{defi}[thm]{Definition}
\newtheorem{prop}[thm]{Proposition}
\newtheorem{coroll}[thm]{Corollary}
\newtheorem{hypo}[thm]{Hypothesis}
\newtheorem{rmk}[thm]{Remark}
\newtheorem{lemma}[thm]{Lemma}

\newtheorem{notat}[thm]{Notation}

\newcommand{\xtilde}[1]{\tilde{X}_{#1}}

\newcommand{\vtilde}[1]{\tilde{V}_{#1}}
\newcommand{\RR}{\mathbb{R}}
\newcommand{\pD}{\partial D}
\newcommand{\PP}{\mathbb{P}}

\newcommand{\EE}{\mathbb{E}}
\newcommand{\CM}{\mathcal{M}}

\newcommand{\supp}{\text{supp}}
\newcommand{\loc}{\text{loc}}
\newcommand{\CD}{\mathcal{D}}
\newcommand{\CB}{\mathcal{B}}
\newcommand{\CP}{\mathcal{P}}
\newcommand{\CG}{\mathcal{G}}
\newcommand{\CT}{\mathcal{T}}
\newcommand{\BpD}{B_{\pD}}

\title[Convergence to equilibrium for a collisionless gas]{A coupling approach for the convergence to equilibrium for a collisionless gas}
  \author{Armand Bernou and Nicolas Fournier}
\address{Sorbonne Universit\'e, CNRS, Laboratoire de Probabilit\'e, Statistique et Mod\'elisation, F-75005 Paris, France.}
\email{armand.bernou@sorbonne-universite.fr, nicolas.fournier@sorbonne-universite.fr}
\thanks{A.B. was supported by grants from R\'egion Ile-de-France.}



\begin{document}
\begin{abstract}
We use a probabilistic approach to study the rate of convergence to equilibrium for a collisionless (Knudsen) gas in dimension equal to or larger than 2. The use of a coupling between two stochastic processes allows us to extend and refine, in total variation distance, the polynomial rate of convergence given in \cite{Aoki201187} and \cite{Kuo2013}. This is, to our knowledge, the first quantitative result in collisionless kinetic theory in dimension equal to or larger than 2 that does not require any symmetry of the domain, nor a monokinetic regime. Our study is also more general in terms of reflection at the boundary: we allow for rather general diffusive reflections and for a  specular reflection component.
\end{abstract}

\maketitle

\section{Introduction}

We consider a Knudsen (collisionless) gas enclosed in a vessel and investigate the rate of convergence to equilibrium. 
We study a $C^2$ bounded domain (open, connected) $D$ in $\mathbb{R}^n$, with $n \geq 2$. The boundary of this domain, $\partial D$, is considered at rest, and when a gas particle collides with the boundary, a reflection which is either diffuse or specular occurs.
For a point $x$ in $\partial D$, $n_x$ denotes the unit inward normal at $x$. 

The distribution function of the gas, $f(t,x,v)$, represents the density of particles with position $x \in \bar{D}$ and velocity $v \in \RR^n$ at time $t \geq 0$. We assume that it satisfies the free-transport equation with both a boundary condition and an initial condition:
\begin{equation}
\label{Problem1}
\left\{ 
	\begin{aligned}
   & \partial_t f + v \cdot \nabla_x f  = 0, \hspace{8cm} (x,v) \in D \times \mathbb{R}^n, \\
&f(t,x,v)(v \cdot n_x)  =  - \alpha(x) c_0 M(v) (v \cdot n_x) \int_{\{v' \cdot n_x < 0\}} f(t,x,v') (v' \cdot n_x) dv' \\
& \quad + (1-\alpha(x)) f(t,x,v-2(v \cdot n_x)n_x) (v \cdot n_x), \hspace{3.5cm} x \in \pD, v \cdot n_x > 0, \\
&f(0,x,v)  = f_0(x,v), \hspace{7.75cm} (x,v) \in D \times \mathbb{R}^n,
	\end{aligned}
\right.
\end{equation}
where the constant $c_0 > 0$ is given by 
\begin{align}
\label{EqCM}
c_0 = \int_{\{u \cdot n_x > 0\}} M(u) (u \cdot n_x) du,
\end{align} 
for any choice of $x \in \pD$. The independence of $c_0$ with respect to $x$ is a consequence of the radial symmetry assumption made below on the density $M$.

This dynamic does not take into account collisions between particles that may occur inside $D$. This is legitimate for the study of Knudsen gases, which are dilute enough. This model represents particles moving in $D$ following the free transport dynamic until they collide with the boundary. When a particle reaches the boundary at some point $x \in \pD$, it is specularly reflected with probability $1-\alpha(x)$, and diffusively reflected with probability $\alpha(x)$. In the latter case, its new velocity is chosen using $M$. See Definition \ref{DefiProcess} for the precise probabilistic interpretation of the model.

Here are our main assumptions. 
\begin{hypo}
\begin{itemize}
\item $D$ is a $C^2$ open connected bounded set in $\RR^n$, with $n \geq 2$.
\item  $\alpha: \pD \to [0,1]$ is uniformly bounded from below, i.e. there exists $\alpha_0 > 0$ such that:
\begin{align}
\label{hypoalpha}
\alpha(x) \geq \alpha_0 , \quad  \forall x \in \pD.
\end{align}
\item 
\label{HypoM}
$M: \RR^n \to \RR_+$ is radially symmetric with $\int_{\RR^n} M(v) dv = 1$, $\int_{\RR^n} \|v\| M(v) dv < \infty$, and there exist $\delta_1 > 0$ and some continuous, radially symmetric, $\bar{M}: \RR^n \to \RR_+$ such that $0 <\bar{M}(v) \leq M(v) $ for all $v \in \RR^n$ such that $0 < \|v\| \leq \delta_1$. 
\end{itemize}
\end{hypo}

The paradigmatic example (and most physically relevant one) of such $M$ is the Maxwellian distribution with parameter (temperature) $\theta$, that fits into this framework:
\begin{align}
\label{DefMaxwellian}
M(v) = \frac{1}{(2\pi \theta)^{\frac{n}{2}}} e^{-\frac{\|v\|^2}{2 \theta}}, \qquad v \in \RR^n.
\end{align}

Observe that informally, (\ref{Problem1}) preserves mass. Indeed, for a strong solution to (\ref{Problem1}), Green's formula gives:
\begin{align*}
\frac{d}{dt} \int_{D \times \RR^n} f(t,x,v)  dv dx= - \int_{D \times \RR^n} \nabla_x (v f(t,x,v))  dvdx = \int_{\pD \times \RR^n} f(t,x,v) (v \cdot n_x) dvdx = 0,
\end{align*}
where the last equality is a consequence of the boundary condition in (\ref{Problem1}).

\subsection{Main result}

The stationary problem corresponding to (\ref{Problem1}) leads to an equilibrium in the phase space. Its distribution is given by (assuming the initial data to be of total mass 1)
\begin{align*}
\mu_{\infty}(x,v) = \frac{M(v)}{|D|}, \quad \forall (x,v) \in D \times \RR^n,
\end{align*}
where $|D|$ denotes the Lebesgue measure of $D$ in $\RR^n$.
Note that (unsurprisingly) the equilibrium distribution is space-homogeneous in $D$.

It is known that there is convergence towards this equilibrium distribution in $L^1$ distance, see for instance Arkeryd and Nouri \cite[Theorem 1.1]{Arkeryd1997} for a proof in the case where $\alpha \equiv 1$ and with slight restrictions on $D$. The goal of this paper is to characterize the rate of this convergence.

Recall that the total variation distance of a signed measure $\mu$ on a measurable space $(E, \mathcal{E})$ is given by $$\|\mu\|_{TV} = \frac{1}{2} \sup \Big\{\int_E g d\mu, g: E \to \RR, \|g\|_{\infty} \leq 1 \Big\}.$$ 

In the whole paper, we use the notation $f(t,x,v)$ when $f$ is a $L^1$-function on $\RR_+ \times D \times \RR^n$ and $f_t(dx,dv)$ when $f$ is measure-valued. Our main result is the following, see Definition \ref{defiweaksol} and Theorem \ref{ThmGreenberg} for the precise meaning of weak solutions.

\begin{thm}
\label{MainTheorem}
Assume that Hypothesis \ref{HypoM} is satisfied.
Let the initial distribution, $f_0$, be a probability measure on $D \times \RR^n$.  Let $r: \mathbb{R}_+ \rightarrow \mathbb{R}_+$ be a continuous increasing function such that there exists a constant $C > 0$ satisfying
$r(x+y) \leq C(r(x) + r(y)) \text{ for all } x,y \in \RR_+$ and such that
\begin{align}
\label{EqHypoMomentThm}
\int_{\RR^n} r\Big(\frac{1}{\|v\|}\Big) M(v) dv < \infty \quad \text{and} \quad \int_{D \times \RR^n} r\Big(\frac{1}{\|v\|}\Big) f_{0}(dx,dv) < \infty.
\end{align}
Then, there exist some constant $\kappa > 0$ and a weak solution  $\rho(dt,dx,dv) = dt f_t(dx, dv)$ to (\ref{Problem1}) such that for all $t \geq 0$,
\begin{align*}
\|f_t - \mu_{\infty}\|_{TV} \leq  \frac{\kappa}{r(t)}.
\end{align*}
Moreover, in the case where $f_0$ admits a density in $L^1(D \times \RR^n)$, the solution $f$ is unique among ``regular'' solutions (see Theorem \ref{ThmUniq}).
\end{thm}

The typical example for the rate $r$ is $r(t) = (t+1)^{n}$, or rather $r(t) = (t+1)^{n-}$, as exemplified by the following situation.

\begin{coroll}
\label{CorollMaxwell}
We take the same hypotheses and notations as in Theorem \ref{MainTheorem}, and assume furthermore that $M$ is bounded (for instance, M is a Maxwellian distribution of the form (\ref{DefMaxwellian})). 
\begin{enumerate}[a)]
\item If $f_0$ has a bounded density, there exists a constant $\kappa > 0$ such that, for all $t \geq 0$,
\begin{align*}
\|f_t - \mu_{\infty}\|_{TV} \leq \frac{\kappa(1+ \log^2(t+1))}{(t+1)^{n}}.
\end{align*}

\item If there exists $d \in (0,n)$ such that 
\begin{align*}
\int_{D \times \RR^n} \frac{1}{\|v\|^{d}} f_{0}(dx,dv)  < \infty,
\end{align*}
 there exists a constant $\kappa > 0$ such that for all $t \geq 0$,
\begin{align*}
\|f_t - \mu_{\infty}\|_{TV} \leq \frac{\kappa}{(t+1)^d}.
\end{align*}
\end{enumerate}
\end{coroll}

Physically, the most interesting case is the following: consider a collisionless gas enclosed in a vessel represented by the domain $D$. The boundary of the domain is kept at temperature $\theta > 0$. A particle colliding with this boundary at $x \in \pD$ is either specularly reflected, with probability $1- \alpha(x)$, or exchanges energy with the boundary and is diffusively reflected with probability $\alpha(x)$, the distribution $M$ being the Maxwellian with temperature $\theta$.

\subsection{Bibliography and discussion}

Relaxation to equilibrium is a key aspect in statistical mechanics. In general, this relaxation, which is known since the H-theorem in the case of the Boltzmann equation, is the result of two main physical equilibrating effects: the collisions between gas molecules and their interactions with the boundary. In \cite{Desvillettes2005}, Desvillettes and Villani find that the distance between the distribution function of the gas at time $t$ and the final equilibrium state decays at a rate $\mathcal{O}(\frac{1}{t^m})$ for all $m > 0$, in the case of space inhomogeneous solutions to the Boltzmann equation satisfying strong conditions of regularity, positivity and decay at large velocities. The rate of \cite{Desvillettes2005} is completed by an exponential rate in the case where the initial data is close to equilibrium in Guo \cite{Guo2010}. In these works, the authors assume that the spatial domain is either the flat torus or a smooth bounded domain with specular or bounce-back reflection at the boundary. Hence the focus is on the equilibrating effect of the collisions between gas molecules rather than the interaction with the boundary, and the equilibrium is entirely determined by the total mass and energy. Later, in \cite{VillaniHypoco}, Villani works on the case of a diffuse or accomodation reflection at the wall of a bounded smooth domain, with a constant temperature at the boundary. The equilibrium is thus slightly changed, as the total mass is now the only conserved quantity. In this case, both collisions between gas molecules and interactions with the boundary play an important role in the relaxation to equilibrium, and give an example of the so-called ``hypocoercivity'' method.

\vspace{.5cm}

Concerning the model studied in this paper, here are the main available results. In \cite{Tsuji2010}, Tsuji, Aoki and Golse find numerically a rate of convergence  in $t^{-n}$ for bounded initial data. An upper bound for the convergence rate in $t^{-1}$ is obtained by Aoki and Golse in \cite{Aoki201187}, assuming some spherical symmetry on the domain and on the initial condition and that $\alpha \equiv 1$. Using a stochastic approach, Kuo, Liu and Tsai in \cite{Kuo2013} obtain the (optimal) convergence rate of $t^{-n}$ in a spherically symmetric domain for $n = 1,2,3$ with $\alpha \equiv 1$ and with bounded initial data satisfying some technical conditions. Later, Kuo \cite{Kuo2015} extended this work, in dimensions 1 and 2, to the case of Maxwell boundary conditions (with additionally some specular reflections). All the above results assume that $M$ is a Maxwellian distribution. We also refer to the connected paper by Mokhtar-Kharroubi and Seifert \cite{Mokhtar2017} who studied a similar problem in slab geometry (in dimension 1) using Ingham's tauberian theorem.
\vspace{.5cm}

Our rate confirms, up to a logarithmic term, both the suggestions made by \cite{Tsuji2010} in view of their numerical results, see Corollary \ref{CorollMaxwell}, and the rate obtained by Kuo \cite{Kuo2015}. It also extends this result to higher dimensions, considers more complicated domains and allows more general initial conditions.

For the most interesting case where $M$ is given by (\ref{DefMaxwellian}), we can sum up our conclusions as follows: if $f_0$ is bounded on $\{v \in \RR^n, \|v\| \leq \epsilon\}$ for some $\epsilon > 0$, e.g. if $f_0(x,v) = g_0(x) \delta_{v_0}(v)$ for some density $g_0$ on $D$ and some $v_0 \ne 0$, the convergence rate towards equilibrium is (up to a logarithmic factor) in $\frac{1}{t^n}$. On the other hand, if $f_0$ is unbounded around $0$, e.g. $f_0(x,v) = \frac{c}{\|v\|^{\alpha}} \mathbf{1}_{\{ \|v\| \leq 1\}}$ with $\alpha \in (0,n)$, the convergence rate towards equilibrium is $\frac{1}{t^{(n-\alpha)-}}$ using Theorem \ref{MainTheorem} with $r(t) = t^{(n-\alpha)-}$. 

In \cite{Kuo2013}, the authors point out that $f_0 - \mu_{\infty}$ (with $f_0$ bounded) is the limiting factor that prevents from a better rate of convergence. We believe that, indeed, our method might allow one to prove the following extension: when one considers two solutions $f_t$ and $g_t$ with $f_0 = \delta_{(x_0,v_0)}$ and $g_0 = \delta_{(y_0, w_0)}$, $\|f_t - g_t\|_{TV} \lesssim t^{-n-1}$ as soon as $v_0 \ne 0$ and $w_0 \ne 0$.

\vspace{.5cm}

 Stochastic billards have also been studied in details, see the works of Evans \cite{evans2001}, Comets, Popov, Schütz and Vachkovskaia \cite{Comets2009} 	and the recent work of F\'etique \cite{FetiqueExplicitspeedconvergence2019} in the convex setting. This corresponds to the monokinetic case of our model: the velocity of particles has a constant norm $1$ ($f_0$ and the distribution $M$ are carried by the unit sphere). They prove exponential convergence to equilibrium by coupling methods. Let us mention that we use a result from Evans on the geometry of $C^1$ domains. 
\vspace{.5cm}

The stochastic process studied in this paper is similar to the family of Piecewise Deterministic Markov Processes (PDMP) introduced by Davis \cite{davis1993markov}. However it does not entirely fit this framework, since the jumps are predictable in our case. In the past few years, several long time behaviours for models corresponding to PDMP have been studied, exhibiting a geometric convergence towards equilibrium. We refer to the study of the telegraph process by Fontbona, Guérin and Malrieu \cite{fontbona_guerin_malrieu_2012, FONTBONA20163077}, and on the recent work of Durmus, Guillin and Monmarché \cite{durmus}.

\vspace{.5cm}

In conclusion, our result is, to the best of our knowledge, the first quantitative result for this problem for a non-symmetric domain in dimension $d \geq 2$, in a non-monokinetic regime. We also consider a more general law $M$ for the reflection at the boundary, with a larger class of initial data $f_0$.

\subsection{Strategy for the proof and plan of the paper}

The next Section \ref{Preliminary} is devoted to the rigorous introduction of our notion of weak solutions, and to the proof of uniqueness under a regularity assumption on $f_0$, in the spirit of Greenberg, van der Mee and Protopopescu \cite{greenberg} and Mischler and Mellet \cite{MELLET2004827}.

\vspace{.5cm} 

In Section \ref{SectionProba}, we construct the stochastic process which we use in the proof of Theorem \ref{MainTheorem}. We show that the law of this stochastic process is a weak solution in the sense of measures to (\ref{Problem1}), and that it is the unique weak solution under further regularity assumptions of $f_0$. The unusual boundary conditions leads to rather non-standard difficulties.
\vspace{.5cm}

In Section \ref{SectionConvexCase}, we derive the proof of our large time result in the context of a uniformly convex domain with $C^2$ boundary, following the strategy described below, and we extend in Section \ref{Extension} the previous result to general domains. For the sake of clarity we start by proving the result in a uniformly convex domain, because the coupling is easier since from any point at the boundary of the domain, we can join any other point at the boundary in one step.
\vspace{.5cm}

It is worth mentioning that the coupling method which we use is close, at least in spirit, to methods based on the study of the Feller nature of the corresponding semigroup. Those methods are known since the work of Meyn and Tweedie \cite{Meyn:2009:MCS:1550713} for exponential rates of convergence, and have recently been extended by Douc, Fort and Guillin \cite{DOUC2009897} for subgeometric convergence rates. They involve the derivation of the modulated moments of the delayed hitting time of some ``petite'' set, a computation that is straightforward once the coupling time is estimated.

\vspace{.5cm}

In a companion paper \cite{BernouSemigroup}, we investigate the same problem by a purely analytic approach. Of course, the main issue is the absence of a spectral gap for the operator corresponding to (\ref{Problem1}), which is the key reason for the polynomial rate of convergence. 

\vspace{.5cm}

To prove Theorem \ref{MainTheorem}, we introduce a coupling $(X_t, V_t)_{t \geq 0}, (\xtilde{t}, \vtilde{t})_{t \geq 0}$ with $(X_t, V_t)$ distributed according to $f_t$ and $(\xtilde{t}, \vtilde{t})$ distributed according to $\mu_{\infty}$, in such a way that the coupling time 
$$ \tau = \inf\{t \geq 0, (X_{t + s})_{s \geq 0} = (\xtilde{t+ s})_{s \geq 0}, (V_{t + s})_{s \geq 0} = (\vtilde{t+ s})_{s \geq 0}\}, $$
is as small as possible. 
We show that it is possible to build a coupling such that the following occurs. 
\begin{enumerate}[i)]
\item When one process collides with the boundary (Proposition \ref{CouplingProp}), if the other one has a large enough speed, so that its next collision occurs sufficiently soon after the one of the first process, there is a positive probability that the two processes coincide for all times following the next collision with the boundary. 
\item We come back to the previous situation after a random number of collisions with the boundary for both processes, and this number of collisions is controlled by a geometric random variable.
\end{enumerate}

The construction of such a coupling is quite subtle. Indeed, the random nature of $(X_t,V_t)_{t \geq 0}$ only appears when $X_t \in \pD$. When one tries to couple two such processes, complex situations can occur, for instance one of the process can hit the boundary several times before the other one does so. To construct a global process satisfying the Markov property, we introduce an extra variable, $(Z_s)_{s \geq 0}$, in the process, see Definition \ref{DefiCouplingProcess}, which allows us to memorize the randomness generated at some rebound of $(X_t)_{t \geq 0}$ until $(\tilde{X}_t)_{t \geq 0}$ hits the boundary.

\vspace{.3cm}

We then show that $r(\tau)$ has finite expectation, roughly, as soon as
$$ \int_{D \times \RR^n } r\Big(\frac{1}{\|v\|}\Big) f_0(x,v) dv dx + \int_{\RR^n} r\Big(\frac{1}{\|v\|}\Big) M(v) dv < \infty. $$
This assumption is crucial: the velocity of a particle has roughly for law either $f_0$ or $M$, the time needed to cross the domain is proportional to the inverse of this velocity, and the coupling can occur only at the boundary.

We then conclude using the fact that:
\begin{align}
\label{EqMaxCoupling}
\|f_t - \mu_{\infty}\|_{TV} \leq \mathbb{P}\Big((X_t,V_t) \neq (\xtilde{t},\vtilde{t})\Big) \leq \mathbb{P}(\tau > t) = \mathbb{P}\Big(r(\tau) > r(t) \Big) \leq \frac{\mathbb{E}[r(\tau)]}{r(t)}
\end{align}
from Markov's inequality, leading us to the rate of convergence in Theorem \ref{MainTheorem}.

\section{Weak Solutions}
\label{Preliminary}

In this section, we give a definition of weak solutions in the sense of measures for (\ref{Problem1}). Existence of this weak solution for any initial probability measure, without further assumption, will be obtained in Section \ref{SectionProba} by a probabilistic method. We show uniqueness of sufficiently regular weak solutions. Let us mention that uniqueness for boundary value problems such as (\ref{Problem1}) cannot be derived in general. We refer to Greenberg, van der Mee, Protopopescu \cite[Chapter 11]{greenberg} for a discussion on those well-posedness issues.

\vspace{.5cm}

We recall that $D$ is a $C^2$ domain (open, connected) in $\RR^n$ and set $G = D \times \RR^n$, $\Sigma = (0,\infty) \times G$. We write $\cdot$ for the scalar product in $\RR^n$, $\|.\|$ for the Euclidian norm.   We also define
\begin{align*}
F_t &= \{(t,x,v), (x,v) \in G\}, \quad t \in \RR_+, \\
 \partial_{\pm} G &= \{(x,v) , \pm v \cdot n_x < 0, x \in \pD, v \in \RR^n \}, \\
 \partial_0 G &= \{(x,v) \in \pD \times \RR^n, v \cdot n_x = 0\}, 
\end{align*}
where we recall that $n_x$ is the unit normal vector at $x \in \pD$ pointing towards $D$. In words, $\partial_+ G$ corresponds to points coming from $D$ towards the boundary, while $\partial_- G$ is the set of points coming from the boundary towards $D$. For a topological space $A$, we write $\CM(A)$ for the set of non-negative Radon measures on $A$, $\CP(A)$ for the set of probability measures on $A$. We denote $\langle .,. \rangle$ the scalar product for the duality $\CM(A)$, $\CM(A)^*$. We write $\mathcal{B}(A)$ for the Borel sigma-algebra on $A$.  For any set $B$, we denote $\bar{B}$ for the closure of $B$, and set $d(D)$ to be the diameter of $D$ :
\begin{align*}
d(D) = \sup_{x,y \in \pD} \|x-y\|.
\end{align*}
 For any space $E$, we write $\CD(E) = C^{\infty}_c(E)$ for the space of test functions (smooth with compact support) on $E$. We set 
\begin{align}
\label{EqDefL}
L = \partial_t + v \cdot \nabla_x.
\end{align} 
We deal with two reference measures:
\begin{itemize}
\item the $n$-dimensional Lebesgue measure (on $D$, $\bar{D}$ and $\RR^n$).
\item the $(n-1)$-dimensional Hausdorff measure in $\RR^n$.
\end{itemize}
To lighten the notations, the same symbols $dx, dv, dz, \dots$ denote all of them. Possible ambiguity can be resolved by checking the space of integration. Similarly the volume of a set $A$, denoted $|A|$ in all cases, refer to the corresponding ambiant space endowed with the appropriate measure.

We let $K: \CM((0, \infty) \times \partial_+ G) \to \CM((0, \infty) \times \partial_- G)$, given, for any measure $\nu \in \CM((0, \infty) \times \partial_+ G)$, any test function $\phi \in \CD((0,\infty) \times \partial_- G)$, by
\begin{align}
\label{EqDefiningKMeasure}
\langle K\nu, \phi \rangle_{(0, \infty) \times \partial_- G} & = \int_{(0, \infty) \times \partial_+ G} \Big( \int_{\{v' \cdot n_x > 0\}} \alpha(x) \phi(t,x,v') c_0 M(v') |v' \cdot n_x| dv' \Big) \nu(dt,dv,dx) \\
&\quad + \int_{(0,\infty) \times \partial_+ G}(1-\alpha(x))\phi(t,x,\eta_x(v))  \nu(dt,dv,dx), \nonumber
\end{align}
for $c_0$ defined by (\ref{EqCM}).
The operator $\eta_x(.)$ is the one of specular reflection at $x \in \partial D$, given by
\begin{align}
\label{EqSpecularRef}
\eta_x(v) = v - 2 (v \cdot n_x)n_x, \quad v \in \RR^n.
\end{align} 
Hence, if $(x,v) \in \partial_{\pm} G, (x, \eta_x(v)) \in \partial_{\mp} G$. 

Whenever necessary, we extend the definition of $K$ to an operator $\bar{K}: \CM(\partial_+ G) \rightarrow \CM(\partial_- G )$ defined similarly. For any measure $\nu \in \CM( \partial_+ G)$, any test function $\phi \in \CD(\partial_- G)$, we set
\begin{align}
\label{EqBoundaryOperator}
\langle \bar{K}\nu, \phi \rangle_{\partial_- G} & = \int_{\partial_+ G} \Big( \int_{\{v' \cdot n_x > 0\}} \alpha(x) \phi(x,v') c_0 M(v') |v' \cdot n_x| dv' \Big) \nu(dv,dx)\\
&\quad + \int_{\partial_+ G}(1-\alpha(x))\phi(x,\eta_x(v))  \nu(dv,dx) \nonumber.
\end{align}

With this at hand, we define our notion of weak solution in the sense of measures. 

\begin{defi}
\label{defiweaksol}
We say that a non-negative Radon measure $\rho \in \CM(\bar{\Sigma})$ is a weak solution to (\ref{Problem1}) with non-negative initial datum $\rho_0 \in \CM(G)$ if
\begin{enumerate}[i)]
\item for all $T > 0$, $\rho((0, T) \times G) < \infty$;
\item there exists a couple of non-negative Radon measures $\rho_{\pm}$ on $(0, \infty) \times \partial_{\pm} G$ such that :
\begin{align}
\label{EqBoundaryWeakSolution}
\rho_- = K\rho_+,
\end{align}
and for all $ \phi \in \mathcal{D}(\bar{\Sigma})$ with $\phi = 0$ on $(0,\infty) \times \partial_0 G$,
\begin{align}
\label{GreenResult}
\langle \rho,L\phi \rangle_{\Sigma} =  - \langle \rho_0,\phi(0,\cdot) \rangle_{G}  + \langle \rho_+,\phi \rangle_{(0, \infty) \times \partial_+ G}- \langle \rho_-,\phi \rangle_{(0,\infty) \times \partial_- G}.
\end{align}
\end{enumerate}
\end{defi}

As we will see in Section \ref{SubsectionLawProcess}, such a solution always exists.
If $f \in C^{\infty}([0, \infty) \times \bar{D} \times \RR^n)$ is a strong solution to (\ref{Problem1}), then $\rho(dt,dx,dv) = f(t,x,v)dt dx dv$ on $(0, \infty) \times D \times \RR^n$ is a weak solution with 
\begin{align*}
\rho_+(dt,dx,dv) &= f(t,x,v)|v \cdot n_x| dt dx dv, \quad \text{ in } (0,\infty) \times \partial_+ G, \\
\rho_-(dt,dx,dv) &= f(t,x,v)|v \cdot n_x| dt dx dv, \quad \text{ in } (0, \infty) \times \partial_- G. 
\end{align*} 
Indeed, this can be understood reading the proof of Theorem \ref{ThmUniq} and mainly relies on the following fact: using that $\partial_t f + v \cdot \nabla_x f = 0 $ in $D \times \RR^n$ and Green's formula, we find that 
\begin{align*}
\langle \rho,L\phi \rangle_{\Sigma} &= \int_0^{\infty} \int_{D \times \RR^n} f L\phi dv dx dt \\
 &= - \int_{D \times \RR^n} f_0 \phi(0, \cdot) dv dx - \int_{0}^{\infty} \int_{\pD \times \RR^n} \phi f (n_x \cdot v) dv dx dt \\
& = - \langle \rho_0,\phi(0,\cdot) \rangle_{G}  + \langle \rho_+,\phi \rangle_{(0, \infty) \times \partial_+ G}- \langle \rho_-,\phi \rangle_{(0,\infty) \times \partial_- G}.
\end{align*} 
The fact that $\rho_- = K \rho_+$ is explained by the boundary condition in (\ref{Problem1}), see Remark \ref{rmkBoundaryOp} below.

In \cite[Proposition 1]{MELLET2004827}, Mellet and Mischler show uniqueness of the solution in an $L^1$ setting for a slightly harder case (namely the Vlasov equation rather than the free transport), with the additional hypothesis that the initial datum belongs to $L^1(D \times \RR^n) \cap L^2(D \times \RR^n)$. We adapt this proof in Theorem \ref{ThmUniq} below. 

When a weak solution can be identified with a function having a few regularity, we can define its trace on $\pD$ in a precise manner. We recall here a result of Mischler \cite{Mischler99}.

\begin{thm}\cite[Theorem 1, $E \equiv 0$, $G \equiv 0$] {Mischler99}
\label{ThmMisch}
 If $f \in L^{\infty}_{\loc} ([0, \infty); L^1_{loc}(\bar{D} \times \RR^n))$ satisfies
$$ L f = 0 \quad \text{ in } \mathcal{D}'((0,\infty) \times D \times \RR^n), $$ 
then there holds that $f \in C([0,\infty), L^1_{\loc}(\bar{D} \times \RR^n))$ and the trace $\gamma f$ of $f$ on $(0,\infty) \times \pD \times \RR^n$ is well defined, it is the unique function $$
\gamma f \in L^1_{\loc}([0,\infty) \times \pD \times \RR^n, (n_x \cdot v)^2 dv dx dt)
$$ 
satisfying the Green's formula: for all $0 \leq t_0 < t_1$, for all $\phi \in \CD(\bar{\Sigma})$ such that $\phi = 0$ on $(0,\infty) \times \partial_0 G,$  
\begin{align}
\label{GreenFormulaTrace}
\int_{t_0}^{t_1} \int_G f L\phi dv dx dt = \Big[\int_G f \phi dv dx \Big]^{t_1}_{t_0} - \int_{t_0}^{t_1}  \int_{\pD \times \RR^n} (\gamma f)  (n_x \cdot v) \phi dv dx dt.
\end{align}

\end{thm}

Observe that all the terms are well-defined in (\ref{GreenFormulaTrace}). In particular, our test functions satisfy $\phi(t,x,v) \leq C |v \cdot n_x|$ for all $(t,x,v) \in (0,\infty) \times \pD \times \RR^n$. 

\begin{rmk}

\label{rmkBoundaryOp}
For any $g \in L^1((0, \infty) \times \partial_+ G, |v \cdot n_x| dv dx dt)$, it holds that $K(g |v \cdot n_x|)$ belongs to $ L^1_{\loc}((0,\infty) \times \partial_- G, |v \cdot n_x| dv dx dt)$ and we have
\begin{align}
\label{EqKRegular}
{K}(|v \cdot n_x| g )(t,x,v) = \alpha(x) c_0 M(v) \int_{\{v' \cdot n_x < 0\}} g(t,x,v') |v' \cdot n_x| dv' + (1 - \alpha(x)) g(t,x, \eta_x(v)), 
\end{align}
for almost every $(t,x,v) \in (0,\infty) \times \partial_- G$.
\end{rmk}

\begin{proof}[Proof of (\ref{EqKRegular})] Set $\nu(dt,dx,dv) = g(t,x,v) |v \cdot n_x| dt dx dv$ on $(0,\infty) \times \partial_+ G$
and consider a test function $\phi \in \CD((0,\infty) \times \partial_-G)$. We have
\begin{align*}
&\langle K(\nu), \phi \rangle_{(0,\infty) \times \partial_- G} \\
 & \quad \quad = \int_0^{\infty} \int_{\partial_+ G} \alpha(x) \Big( \int_{\{v' \cdot n_x > 0\}} \phi(t,x,v')c_0M(v')|v' \cdot n_x| dv' \Big) g(t,x,v) |v \cdot n_x| dv dx  dt \\
 & \qquad \quad + \int_0^{\infty} \int_{\partial_+ G} (1-\alpha(x)) \phi(t,x,\eta_x(v)) g(t,x,v) |v \cdot n_x| dv dx dt \\
 & \qquad = \int_0^{\infty} \int_{\partial_- G} \phi(t,x,v) \Big(\alpha(x) c_0 M(v) \int_{\{v' \cdot n_x < 0\}} g(t,x,v') |v' \cdot n_x| dv' \Big) |v \cdot n_x| dv dx dt \\
 & \qquad \quad + \int_0^{\infty} \int_{\partial_- G} \phi(t,x,v) (1-\alpha(x)) g(t, x, \eta_x(v)) |v \cdot n_x|  dv dx dt.
\end{align*}
In the first integral, we only exchanged the roles of $v$ and $v'$. In the second one, we performed the involutive change of variables $v' = \eta_x(v)$ and used that $|\eta_x(v) \cdot n_x| = |v \cdot n_x|$ for all $(x,v) \in \partial_+ G$. Since this holds for any $\phi \in \CD((0, \infty) \times \partial_- G)$, (\ref{EqKRegular}) follows.
\end{proof}

For $f$ with the same regularity as in Theorem \ref{ThmMisch}, $\gamma_\pm f$ denote the restrictions of $\gamma f$ to $(0, \infty) \times \partial_\pm G$.
From (\ref{GreenResult}) and (\ref{GreenFormulaTrace}) and the uniqueness of this trace function it is clear that if the measures $\rho_{\pm}$ in Definition \ref{defiweaksol} admit two densities $f_{\pm}$ with respect to the measure $|v \cdot n_x| dv dx dt$ on $(0, \infty) \times \partial_{\pm} G$, those densities can be identified with $\gamma_{\pm} f$.

 We now adapt the uniqueness result in Proposition 1 in \cite{MELLET2004827}.

\begin{thm}
\label{ThmUniq}
 Consider $f \in C_w([0,\infty); L^1 (\bar{D} \times \RR^n))$ for all $T > 0$ (i.e. $f$ is weakly continuous in time in the sense of measures) admitting a trace function $\gamma f \in L^1([0,T] \times \pD \times \RR^n, |v \cdot n_x| dv dx dt)$ (for all $T > 0$) such that formula (\ref{GreenFormulaTrace}) holds. Assume that $\rho(dt,dx,dv) = f(t,x,v) dt dx dv$ is a weak solution to (\ref{Problem1}) with initial condition $f_0 \in L^1(D \times \RR^n)$. Then, we have
\begin{equation}
\label{distribsolution}
\left\{
\begin{aligned}
\begin{array}{lll}
L f = (\partial_t + v \cdot \nabla_x)f = 0 \quad &\text{ in } \mathcal{D}'((0,\infty) \times D \times \RR^n), \\
f(0,.) = f_0 \quad  &\text{a.e. in } D \times \RR^n, \\
(v \cdot n_x) \gamma_- f  = K \Big( |v \cdot n_x| \gamma_+ f \Big) \quad &\text{a.e. in } (0,\infty) \times \partial_- G. 
\end{array}
\end{aligned}
\right.
\end{equation} 
Moreover, $f$ is the unique solution to (\ref{distribsolution}) with this regularity.
\end{thm}

As we will see in Theorem \ref{ThmGreenberg}, such a solution always exists, assuming of course that $f_0$ is a probability density function. 

\begin{proof}

\textbf{Step 1.} Here, we prove that $f$ solves (\ref{distribsolution}). 

We first claim that we have the two equalities $\rho_+(dt,dx,dv) = \gamma_+ f(t,x,v) |v \cdot n_x| dt dx dv$ and $\rho_-(dt,dx,dv) = \gamma_- f(t,x,v) |v \cdot n_x| dt dx dv$. Indeed, consider a test function $\phi \in \CD((0, \infty) \times \bar{D} \times \RR^n)$, with $\phi = 0$ on $(0, \infty) \times \partial_0 G$. Using (\ref{GreenResult}), the definition of $\rho$ and (\ref{GreenFormulaTrace}), we obtain
\begin{align*}
\langle \rho_+, \phi \rangle_{(0, \infty) \times \partial_+ G} - \langle \rho_-, \phi \rangle_{(0, \infty) \times \partial_- G} = \langle \rho, L\phi \rangle_{\Sigma} &= \int_0^{\infty} \int_{D \times \RR^n} f L\phi dv dx dt \\
&=  - \int_0^{\infty} \int_{\pD \times \RR^n} (v \cdot n_x) (\gamma f) \phi  dv dx dt.
\end{align*}
from which we deduce that $\rho_+(dt,dx,dv) - \rho_-(dt,dx,dv) = \gamma f(t,x,v) (v \cdot n_x) dt dx dv$ whence the claim. With this at hand, the third equation of (\ref{distribsolution}) follows immediatly from (\ref{EqBoundaryWeakSolution}) and Remark \ref{rmkBoundaryOp}.

The first equation of (\ref{distribsolution}) follows from (\ref{GreenResult}) and the definition of $\rho$, since for all $T > 0$, the right-hand side of (\ref{GreenResult}) is $0$ for $\phi \in \mathcal{D}((0,T) \times D \times \RR^n)$.

For the second equation of (\ref{distribsolution}), we want to prove that for any $\phi \in \CD(D \times \RR^n)$, 
\begin{align}
\label{EqInitialConditionUniq}
\int_{D \times \RR^n} \phi(x,v) f(0,x,v) dv dx = \langle f_0, \phi \rangle_{D \times \RR^n}.
\end{align}
Using the definition of $\rho$ and the equation (\ref{GreenResult}) we obtain immediatly
\begin{align*}
\int_0^{\infty} \int_{D \times \RR^n} L\psi f  dv dx dt  = - \langle f_0, \psi(0,.) \rangle_{D \times \RR^n}
\end{align*}
for any $\psi \in \CD([0,\infty) \times D \times \RR^n)$. 
Let $\phi \in \CD(D \times \RR^n)$, $\epsilon \in (0,1)$ and define the function $\beta_{\epsilon}$ by $\beta_{\epsilon}(t) = e^{-\frac{t}{\epsilon - t}} \mathbf{1}_{\{t \in [0, \epsilon)\}}$.
Therefore $\beta_{\epsilon}$ is smooth with compact support in $[0, \infty)$ and we can apply the previous equation with $\psi(t,x,v) = \beta_{\epsilon}(t) \phi(x,v)$ to find
\begin{align}
\label{EqAlphaEps1}
\int_0^{\infty} \int_{D \times \RR^n} \Big(\beta_{\epsilon}'(t) \phi(x,v) + \beta_{\epsilon}(t) v \cdot \nabla_x \phi(x,v) \Big) f(t,x,v)  dv dx dt = - \langle f_0, \phi \rangle_{D \times \RR^n}.
\end{align}
We set
\begin{align*}
J_{\epsilon} = \int_0^{\infty} \int_{D \times \RR^n} \beta_{\epsilon}(t) v \cdot \nabla_x \phi(x,v) f(t,x,v)  dv dx dt,
\end{align*}
and
\begin{align*}
I_{\epsilon} = \int_0^{\infty} \int_{D \times \RR^n} \beta_{\epsilon}'(t) \phi(x,v) f(t,x,v)  dv dx dt,
\end{align*}
so that  (\ref{EqAlphaEps1}) writes
$$ I_{\epsilon} + J_{\epsilon} = - \langle f_0, \phi \rangle_{D \times \RR^n}. $$
Since $\phi \in \CD(D \times \RR^n)$, $\beta_{\epsilon} \leq 1$ and $\beta_{\epsilon}(t) \to 0$ a.e. as $\epsilon$ converges to $0$ the dominated convergence theorem gives immediatly $\lim \limits_{\epsilon \to 0} J_{\epsilon} = 0$.
On the other hand, since $\int_0^{\epsilon} |\beta'_{\epsilon}(t)|dt = - \int_0^{\epsilon} \beta'_{\epsilon}(t)dt = 1$,
\begin{align*}
I_{\epsilon} = \Delta_{\epsilon}  - \int_{D \times \RR^n} \phi(x,v) f(0,x,v) dv dx,
\end{align*} 
with
\begin{align*}
\Delta_{\epsilon} = \int_0^{\infty} \int_{D \times \RR^n} \beta'_{\epsilon}(t) \phi(x,v) \Big( f(t,x,v) - f(0,x,v) \Big) dv dx dt.
\end{align*}
We have,
\begin{align*}
|\Delta_{\epsilon}| &\leq \int_0^{\epsilon} |\beta'_{\epsilon}(t)| dt \Big| \int_{D \times \RR^n} \phi(x,v) \Big(f(t,x,v) - f(0,x,v)\Big) dv dx \Big| \\
&\leq \sup_{t \in [0, \epsilon]} \Big|\int_{D \times \RR^n} \phi(x,v) \Big(f(t,x,v) - f(0,x,v)\Big) dv dx \Big|.
\end{align*}
The resulting supremum converges to $0$ as $\epsilon$ goes to $0$ using the weak continuity of $f$.
Taking the limit as $\epsilon$ goes to $0$ in (\ref{EqAlphaEps1}) completes the proof of (\ref{EqInitialConditionUniq}).

\vspace{.5cm}

\textbf{Step 2.} We now show uniqueness of the solution through a contraction result in $L^1(D \times \RR^n)$. Consider two solutions $g_1, g_2$ of (\ref{distribsolution})  with the same initial datum $g_0$. By linearity, $f = g_1 - g_2$ is again a solution to (\ref{distribsolution}) the problem with an initial datum $f_0 \equiv 0$ (the trace being $\gamma f = \gamma g_1 - \gamma g_2$ by linearity of the Green's formula (\ref{GreenFormulaTrace})). Let $\beta \in W^{1, \infty}_{\loc}(\RR)$ such that $|\beta(y)| \leq C_{\beta}(1+|y|)$,  for some constant $C_{\beta} >0$ and for all $y \in \RR$. From \cite[Proposition 2]{Mischler99} (note that our hypothesis on $\gamma f$ implies $\gamma f \in L^1_{\loc}((0, \infty) \times \pD \times \RR^n, |v \cdot n_x|^2 dv dx dt)$), we know that
\begin{gather*}
L \beta(f) = (\partial_t + v \cdot \nabla_x)\beta(f) = 0, \quad \text{ in } \mathcal{D}'((0,\infty) \times D \times \RR^n), \\
\gamma \beta (f) = \beta (\gamma f), \hspace{2.2cm} \text{ in } (0,T) \times \pD \times \RR^n.
\end{gather*}
We now choose $\beta(y) = |y|$, which satisfies the previous requirements. We set $0 < t_0 < t_1$ and  for all $\epsilon \in (0,t_0)$, $\delta_{\epsilon}(t) = \mathbf{1}_{(t_0,t_1)}(t) + e^{-\frac{t-t_1}{\epsilon + t_1 - t}} \mathbf{1}_{[t_1, t_1 + \epsilon)} + e^{-\frac{t_0 - t}{\epsilon + t - t_0}} \mathbf{1}_{(t_0-\epsilon,t_0)}$ and apply the Green's formula (\ref{GreenFormulaTrace}) to $|f|$ with the test function $\psi(t,x,v) = \delta_{\epsilon}(t) \phi(x,v)$ for all $(t,x,v) \in [0,\infty) \times \bar{D} \times \RR^n$, where $\phi \in \CD(D \times \RR^n)$, so that $\psi \in \CD((0,\infty) \times D \times \RR^n)$ using that $\delta_{\epsilon}$ is smooth with support in $(t_0 - \epsilon, t_1 + \epsilon)$. We obtain
\begin{align*}
0 = \int_{t_0}^{t_1} \int_{D \times \RR^n} |f|L\psi dv dx ds = \Big[ \int_{D \times \RR^n} |f| \psi dv dx \Big]_{t_0}^{t_1}.
\end{align*}
Since $\delta_{\epsilon}(t_1) = \delta_{\epsilon}(t_0) = 1$, we deduce
$$ \int_{D \times \RR^n} \big(|f(t_1)| - |f(t_0)|\big) \phi(x,v) dx dv = 0. $$
Since $f$ is weakly continuous, we let $t_0 \to 0$, and, using $|f(0)| = |f_0| = 0$ almost everywhere in $D \times \RR^n$, we conclude that for all $t_1 > 0$
$$ \int_{D \times \RR^n} |f(t_1,x,v)| \phi(x,v) dx dv = 0, $$
for all $\phi \in \CD(D \times \RR^n)$. This completes the proof.
\end{proof}

In the next subsection, we construct a stochastic process from which we obtain a weak solution to the problem. Ultimately, we show the following well-posedness result, which follows from Theorem \ref{ThmUniq}, Propositions \ref{PropWeakSolution} and \ref{PropRegularitySol}. 

\begin{thm} \
\label{ThmGreenberg} 
\begin{enumerate}[(i)]
\item Let $\rho_0 \in \mathcal{P}(D \times \RR^n)$. There exists a weak solution $\rho$ in the sense of Definition \ref{defiweaksol} to (\ref{Problem1}) with inital data $\rho_0$. This solution writes $\rho(dt,dx,dv) = dt f_t(dx,dv)$ on $\Sigma$, with $t \to f_t$ right-continuous from $[0, \infty)$ to $\mathcal{P}(D \times \RR^n)$. 
\item If moreover $\rho_0$ admits a density $f_0 \in L^1(D \times \RR^n)$, then for all $t \geq 0$, $f_t$ admits a density $f(t,.)$ with respect to the Lebesgue measure on $D \times \RR^n$. We have, for all $T > 0$, $f \in C([0,T); L^1(\bar{D} \times \RR^n))$ and the trace measure of $f$, $\gamma f$ satisfies $\gamma f \in L^1([0,T] \times \pD \times \RR^n, |v \cdot n_x| dtdxdv)$. Hence $f$ is the unique weak solution to (\ref{Problem1}) with such regularity. 
\end{enumerate}
\end{thm}

%----------------------------------------------------
\section{Probabilistic setting}
\label{SectionProba}

In this section, we build a stochastic process which corresponds to the evolution of a gas particle. Then we show that its law (roughly speaking) is a weak solution in the sense of Definition \ref{defiweaksol} of (\ref{Problem1}), and enjoys the regularity requirements of Theorem \ref{ThmUniq} when the initial condition admits a density.

\subsection{Construction of the process}
\label{SubsectionConstructionProbabilisticSolution}
We start by setting some notations that will show useful in the construction of the stochastic process. We set $\mathcal{A} = (-\frac{\pi}{2}, \frac{\pi}{2}) \times [0,\pi)^{n-2}$. We recall that the Jacobian of the hyperspherical change of variables $v  \to (r,\theta_1, \dots, \theta_{n-1})$ from $\RR^n$ to the space  $(0, \infty) \times [- \pi, \pi) \times [0,\pi)^{n-2}$ is given by $r^{n-1} \prod_{j = 1}^{n-2} \sin(\theta_j)^{n-1-j}$. For $r \in \RR_+$, we abusively write $M(r) = M(v)$ with $v \in \RR^n$, $\|v\| = r$.

\begin{lemma} \
\label{NotatHM}
We define $h_{R}: \RR_+ \to \RR_+$ to be the density given by $h_R(r) = c_R r^n M(r) \mathbf{1}_{\{r \geq 0\}}$, where $c_R$ is a normalizing constant. Let also $h_{\Theta}$ the density on $\mathcal{A}$ defined by
$$h_{\Theta}(\theta_1, \dots, \theta_{n-1}) = c_{\Theta} \cos(\theta_1) \prod_{j = 1}^{n-2} \sin(\theta_j)^{n-1-j}.$$ 
We write $\Upsilon$ for the law of $(R,\Theta)$, $R$ having density $h_R$, $\Theta$ having density $h_{\Theta}$ independent of $R$.

There exists a measurable function $\vartheta: \pD \times \mathcal{A} \to \RR^n$ such that for any $x \in \pD,$ any $\Upsilon$-distributed random variable $(R, \Theta)$,
\begin{align}
\label{EqRVarTheta}
R \vartheta(x,\Theta) \sim c_0 M(v) |v \cdot n_x| \mathbf{1}_{\{v \cdot n_x > 0\}},
\end{align}
and such that for all $\theta = (\theta_1, \dots, \theta_{n-1}) \in \mathcal{A}$, $x \in \pD$
\begin{align}
\label{EqVarTheta}
\vartheta(x,\theta) \cdot n_x = \cos(\theta_1).
\end{align} 
\end{lemma}

\begin{proof}

For $(e_1, \dots, e_n)$ the canonical basis of $\RR^n$, we define, $P: \RR^n \to [0, \infty) \times [- \pi, \pi) \times [0,\pi]^{n-2}$, which, to a vector expressed in the $(e_1, \dots, e_n)$ coordinates, gives the associated hyperspherical coordinates (with polar axis $e_1$). 
For $x \in \pD$, we fix an orthonormal basis $(n_x, f_2 \dots, f_n)$ of $\RR^n$ and consider the isometry $\xi_x$ that sends $(e_1, \dots, e_n)$ to $(n_x, f_2, \dots, f_n)$.
We then set, for $\theta \in \mathcal{A}$,
\begin{align*}
\vartheta(x,\theta) = \Big( \xi_x \circ P^{-1} \Big)(1, \theta).  
\end{align*}
With this construction, $\vartheta$ is such that (\ref{EqRVarTheta}) holds. Finally, by definition of $P$ and $\xi_x$, we have
\begin{align*}
\cos(\theta_1) = P^{-1}(1,\theta) \cdot e_1 = \xi_x\Big(P^{-1}(1,\theta) \Big) \cdot n_x,
\end{align*}
as desired.
\end{proof} 

\begin{rmk}
Note that the fact that $\int_{\RR_+} s^n M(s) ds < \infty$ follows from $\int_{\RR^n} \|v\| M(v) dv < \infty$, see Hypothesis \ref{HypoM}, using hyperspherical coordinates.
\end{rmk}

\begin{notat}
We introduce two important deterministic maps. Define $\zeta: \bar{D} \times \RR^n \rightarrow \RR_+$ by
\begin{align}
\label{defsigma}
\zeta(x,v) &= \left\{ 
\begin{array}{ll}
\inf \{s > 0, x + sv \in \pD\},& \text{ if } (x,v) \in G \cup \partial_- G, \\
0, & \text{ if } (x,v) \in \partial_+ G \cup \partial_0 G.
\end{array} 
\right.  
\end{align}
We also define $q: \bar{D} \times \RR^n \rightarrow \pD$ by
\begin{align}
\label{defq}
q(x,v) &= \left\{ 
\begin{array}{ll}
x + \zeta(x,v)v,& \text{ if } (x,v) \in G \cup \partial_- G, \\
x, & \text{ if } (x,v) \in \partial_+ G \cup \partial_0 G.
\end{array} 
\right.  
\end{align}
\end{notat}

For a gas particle governed by the dynamics of (\ref{Problem1}), in position $(x,v) \in \bar{D} \times \RR^n$ at time $t = 0$, $\zeta(x,v)$ is the time of its first collision with the boundary, while $q(x,v)$ is the point of $\pD$ where this collison occurs.
The value attributed to those functions on $\partial_0 G$ has no consequences on our study, since our dynamic forbids the occurence of this situation.

 Recall that $\eta_x(v) = v - 2(v \cdot n_x) n_x$ for all $(x,v) \in \pD \times \RR^n$.

\begin{notat}
\noindent We define the map $w: \pD \times \RR^n \times [0,1] \times \RR_+ \times \mathcal{A}  \to \RR^n$ by
\begin{align}
\label{EqDefW}
 w(x,v,u,r,\theta) = \eta_x(v) \mathbf{1}_{\{u > \alpha(x)\}} + r \vartheta(x,\theta) \mathbf{1}_{\{u \leq \alpha(x)\}}. 
 \end{align}
We write $\mathcal{U}$ for the uniform distribution over $[0,1]$, and denote $\mathcal{Q}$ the measure $\mathcal{U} \otimes \Upsilon$. 
\end{notat}

Let us define, given an appropriate sequence of inputs, our process.

\begin{defi}
\label{DefiProcess}
Consider an initial distribution $\rho_0$ on $(D \times \RR^n) \cup \partial_- G$, a sequence  of i.i.d. random vectors $(U_i, R_i, \Theta_i)_{i \geq 1} $ of law $\mathcal{Q}$. We define the stochastic process $(X_t, V_t)_{t \geq 0}$ as follows:
\begin{labeling}{Step k+1:}
\item[Step 0:] Let $(X_0,V_0)$ be distributed according to $\rho_0$. 
\item[Step 1:] Set $T_1 = \zeta(X_0, V_0)$.  

\noindent For $t \in [0, T_1)$, set $ V_t = V_0$ and $ X_t = X_0 + t V_0$. 

\noindent Set $X_{T_1} = X_{T_1-}$ and $V_{T_1} = w(X_{T_1}, V_{T_1-}, U_1, R_1, \Theta_1)$.
%\begin{itemize}
%\item If $B_1 \geq \alpha(X_{T_1})$, set $V_{T_1} = V_0 - 2 (V_0 \cdot n_{X_{T_1}})n_{X_{T_1}}$ (specular reflection). 
%\item Otherwise, set $V_{T_1} = R_1 \vartheta(X_{T_1}, \Theta^1)$ (diffuse reflection). 
%\end{itemize}

\item[Step k+1:] Set $T_{k+1} = T_k + \zeta(X_{T_k}, V_{T_k})$. 

\noindent For all $t \in (T_k, T_{k+1})$, set $X_t = X_{T_k} + (t - T_k)V_{T_k},$ $V_t = V_{T_k}$. 

\noindent Set $X_{T_{k+1}} = X_{T_{k+1}-} \in \pD$ and 

\noindent $V_{T_{k+1}} = w(X_{T_{k+1}}, V_{T_{k+1}-}, U_{k+1}, R_{k+1}, \Theta_{k+1})$. 
\item[etc.]
\end{labeling}
We say that $(X_s,V_s)_{s \geq 0}$ is a free-transport process with initial distribution $\rho_0$.
\end{defi} 

\begin{rmk}
\label{RmkExtensionDefiProcess}
We extend the previous definition to the case where $(x,v) \in \partial_+ G$ and $\rho_0 = \delta_x \otimes \delta_v$, with, informally, $X_0 = x, V_{0-} = v$. In this case, we pick an extra triplet $(U_0, R_0, \Theta_0) \sim \mathcal{U} \otimes \Upsilon$  independent of everything else and we set
$$ X_0 = x, \quad V_0 = w(x,v, U_0, R_0, \Theta_0). $$
Step 1 and further remain the same.
\end{rmk}

\subsection{Non-explosion}

In this section, we show that the process constructed in Definition \ref{DefiProcess} is almost surely well defined for all times $t > 0$. For $m \geq 1$, we write $\mathbb{S}^m = \{x \in \RR^{m+1}, \|x\| = 1\}$ for the unit sphere in $\RR^{m+1}$. Recall that any $C^2$ bounded domain satisfies the uniform interior ball condition and therefore the following interior cone condition, see for instance Fornaro \cite[Proposition B.0.16 and its proof]{fornaro}. 

\begin{defi}
\label{DefiUniformCone}
We say that a bounded set $D \subset \mathbb{R}^n$ satisfies the uniform cone condition if there exist $\beta \in (0,1)$, $h > 0$, such that for all $x \in \pD$,
\begin{align*}
C_x &= \{x + tu, t \in (0,h), u \cdot n_x > \beta, u \in \mathbb{S}^{n-1} \} \subset D.
\end{align*}
\end{defi}

\begin{prop}
\label{PropExplosion}
Under Hypothesis \ref{HypoM}, the sequence $(T_i)_{i \geq 1}$ of Definition \ref{DefiProcess} almost surely satisfies $T_i \rightarrow + \infty$  as $i \to + \infty$. More precisely, for any $T > 0$, $\EE[\#\{i \geq 1: T_i \leq T\}] < \infty$.
\end{prop}

\begin{proof}
Let $h$ and  $\beta$ be the positive constants of the uniform cone condition corresponding to $D$.
Recall that there exists a constant $\alpha_0  > 0$ such that for any $x \in \pD$, $\alpha(x) \geq \alpha_0 $. For $N$ large enough, writing $\Theta_1 = (\Theta_1^1, \dots, \Theta_1^{n-1}) \in \RR^{n-1}$, we have
$$p = \PP\Big(U_1 \leq \alpha_0 , \cos(\Theta^1_1) > \beta, R_1 \in [0, N] \Big) > 0. $$
Using Borel-Cantelli's lemma, one concludes that almost surely, an infinite number of elements of the sequence $\Omega_i = \{U_i \leq \alpha_0, \cos(\Theta^1_i) > \beta, R_i \in [0, N]\}$  is realized. 
For all $i \geq 1$, on $\Omega_i$, $\vartheta(X_{T_i}, \Theta_i) \cdot n_{X_{T_i}} > \beta$, 
whence  $X_{T_i + t} = X_{T_i} + tV_{T_i} \in C_{X_{T_i}} \subset D$ for all $t \in [0, \frac{h}{N}]$, because $V_{T_i}= R_i \vartheta(X_{T_i},\Theta_i)$ has a norm smaller than $N$. 

Set $T_0 = 0$ and $\tau_i = T_{i+1} - T_{i}$ for all $i \geq 1$. By the previous observation, we have, on $\Omega_i$, 
$$ \tau_i = \zeta(X_{T_i}, V_{T_i}) =  \frac{|q(X_{T_i}, V_{T_i}) - X_{T_i}|}{R_i} \geq \frac{h}{N} > 0,$$
To conclude, note first that
\begin{align*}
\lim \limits_{i \to + \infty} T_i \geq \sum_{j \geq 1} \tau_j \mathbf{1}_{\Omega_j} \geq \frac{h}{N} \sum_{j \geq 1} \mathbf{1}_{\Omega_j} = + \infty \quad \text{ a.s. }
\end{align*}

For the second part of the propositon, we let $T > 0$ and we set $N_T := \sup \{i \geq 1, \tau_1 + \dots + \tau_i \leq T\}$. 
For all $i \geq 1$, we let $(\sigma_i)_{i \geq 1}$ be the i.i.d. sequence defined by $\sigma_i = \frac{h}{N} \mathbf{1}_{\Omega_i}$, and define the random variable $M_T$ by $M_T := \sup \{i \geq 1, \sigma_1 + \dots + \sigma_i < T\}$. We have
\begin{align*}
\EE[\#\{i \geq 1: T_i \leq T\}] \leq  \EE[N_T] + 1 \leq \EE[M_T]  + 1,
\end{align*}
since for all $i \geq 1, \tau_i \geq \sigma_i$ almost surely. Since the sequence $(\sigma_i)_{i \geq 1}$ is i.i.d., it follows from a classical result of renewal theory that $\EE[M_T] < \infty$, which terminates the proof.
\end{proof}

\subsection{Law of the process}
\label{SubsectionLawProcess}

\begin{prop}
\label{PropWeakSolution}
Let $\rho_0 \in \mathcal{P}(D \times \RR^n)$ and consider the process $(X_t, V_t)_{t \geq 0}$ from Definition \ref{DefiProcess}. Set, for all $t \geq 0$, $f_t$ to be the law of $(X_t,V_t)$, and define the measure $\rho$ on $\bar{\Sigma}$ by
$$ \rho(dt,dx,dv) = f_t(dx,dv) dt. $$
Then $\rho$ is a weak solution to (\ref{Problem1}) in the sense of Definition \ref{defiweaksol}. Moreover $t \to f_t(dx,dv)$ is right-continuous from $(0, \infty)$ to $\mathcal{P}(\bar{D} \times \RR^n)$ endowed with the weak convergence of measures. 
\end{prop}

\begin{rmk}
The boundary measures corresponding to $\rho$ in Definition \ref{defiweaksol} are given by
\begin{align*}
\rho_{\pm}(A) = \EE\big[\sum_{i \geq 1} \mathbf{1}_{(T_i,X_{T_i},V_{T_i}) \in A}\big], \quad A \in \CB((0, \infty) \times \partial_{\pm} G).
\end{align*}
\end{rmk}

\begin{proof} [Proof of Proposition \ref{PropWeakSolution}]

%Note that $\PP(\exists i \in \mathbb{N}^*, T_i = T) = 0$ by measure continuity and the definition of the $(T_i)_i$, so that 

From its definition, it is clear that $\rho$ is a  non-negative Borel measure on $\bar{\Sigma}$. For all $T > 0$,
\begin{align*}
\rho((0,T) \times G) = \int_0^T \EE[\mathbf{1}_{(t,X_t,V_t) \in \Sigma}] dt \leq T,
\end{align*}
so that $\rho$ is also Radon.  

\vspace{.5 cm  }

For $i \geq 1$, we introduce two probability measures $\rho_{\pm}^i$ on $\RR_+ \times \partial_{\pm}G$: $\rho_+^i$ is the law of the triple $(T_i, X_{T_i}, V_{T_i-})$ and $\rho_-^i$ is the law of the triple $(T_i, X_{T_i}, V_{T_i})$. 

We now prove that for all $i \geq 1$, $\rho_-^i = K \rho_+^i$. 
For $B \in \CB( \RR_+ \times \partial_- G)$, using the definition of $(V_t)_{t \geq 0}$, we have
\begin{align*}
\rho_-^i(B) &= \mathbb{E}[\mathbf{1}_{(T_i,X_{T_i},V_{T_i}) \in B}] \\
&= \EE\Big[\alpha(X_{T_i}) \mathbf{1}_{\big(T_i, X_{T_i}, R_i \vartheta(X_{T_i},\Theta_i)\big) \in B}\Big] + \EE\Big[(1-\alpha(X_{T_i})) \mathbf{1}_{\big(T_i, X_{T_i}, \eta_{X_{T_i}}(V_{T_i-})\big) \in B}\Big] 
\end{align*}
Using (\ref{EqRVarTheta}), we deduce,
\begin{align*}
\rho_-^i(B) &= \int_{(0,\infty) \times \partial_+ G} \int_{  \{w \in \RR^n, w \cdot n_x > 0\}} \alpha(x) \mathbf{1}_{\{(t,x,w) \in B\}} c_0 M(w) |w \cdot n_x| dw \rho^i_+(dt,dx,dv) \\
&\quad + \int_{(0,\infty) \times \partial_+ G} \mathbf{1}_{\{(t,x,\eta_x(v)) \in B\}} (1-\alpha(x)) \rho^i_+(dt,dx,dv) \\
&= K\rho^i_+(B),
\end{align*}
recall (\ref{EqDefiningKMeasure}).
Setting $\rho_+(A) = \sum_{i \geq 1} \rho_+^i (A) $ for all $A \in \CB(\RR_+ \times \partial_+ G)$, $\rho_-(B) = \sum_{i \geq 1} \rho_-^i(B)$ for all $B \in \CB(\RR_+ \times \partial_- G)$, we deduce that $\rho_- = K \rho_+$ on $\RR_+ \times \partial_- G$.

We now prove (\ref{GreenResult}). Let $\phi \in \CD(\bar{\Sigma})$. We have, by definition of $\rho$ and using Definition \ref{DefiProcess},
\begin{align*}
\langle \rho,L\phi \rangle_{\Sigma} &= \int_0^{\infty} \EE[L\phi(t,X_t,V_t)] dt \\
&= \int_0^{\infty} \EE\Big[\sum_{i = 0}^{\infty} \mathbf{1}_{\{T_i \leq t < T_{i+1}\}} L\phi(t, X_{T_i} + (t - T_i) V_{T_i}, V_{T_i})\Big] dt \\
&= \sum_{i = 0}^{\infty} \EE\Big[\int_{T_i}^{T_{i + 1}} (\partial_t + V_{T_i} \cdot \nabla_x ) \phi(t, X_{T_i} + (t - T_i)V_{T_i}, V_{T_i}) dt \Big] \\
&= \sum_{i = 0}^{\infty} \EE \Big[\int_{T_i}^{T_{i + 1}} \frac{d}{dt} \Big( \phi(t, X_{T_i} + (t - T_i)V_{T_i}, V_{T_i})\Big) dt \Big].
\end{align*}
As a conclusion,
 \begin{align*}
\langle \rho, L\phi \rangle_{\Sigma}
&= \EE \Big[\sum_{i = 0}^{\infty} \phi(T_{i+1}, X_{T_i} + (T_{i+1} - T_i) V_{T_i}, V_{T_i}) \Big] - \EE \Big[\sum_{i = 1}^{\infty}  \phi(T_i, X_{T_i}, V_{T_i}) \Big] \\
&\quad - \EE[\phi(0,X_0,V_0)]\\
& = \EE \Big[\sum_{i = 0}^{\infty} \phi(T_{i+1}, X_{T_{i+1}} , V_{{T_{i+1}-}}) \Big]
 - \langle  \rho_-, \phi \rangle_{(0,\infty) \times \partial_- G} - \langle  \rho_0,\phi(0,.) \rangle_{D \times \RR^n} \\
&= \langle  \rho_+,\phi \rangle_{(0,\infty) \times \partial_+ G} - \langle  \rho_-,\phi \rangle_{(0,\infty) \times \partial_- G} - \langle  \rho_0,\phi(0,.) \rangle_{D \times \RR^n},
\end{align*}
which concludes the proof that $\rho$ is a weak solution. Observe that all the computations above can easily be justified because there exists some $T > 0$ such that $\supp(\phi) \subset [0,T] \times \bar{D} \times \RR^n$. 

\vspace{.5cm}

The right-continuity of $t \to f_t$ on $(0, \infty)$ is a straightforward result given that $(X_t)_{t \geq 0}$ is continuous and $(V_t)_{t \geq 0}$ is càdlàg on $(0, \infty)$ according to Definition \ref{DefiProcess}. 
 \end{proof}

In the next proposition, we study the regularity of the solution given by Proposition \ref{PropWeakSolution} in the case where the initial data $\rho_0$ has a density in $D \times \RR^n$. 

\begin{prop}
\label{PropRegularitySol}
For $\rho_0$ having a density $f_0 \in L^1(D \times \RR^n)$, the Radon measure
 $\rho$ defined in Proposition \ref{PropWeakSolution} admits a density $f$ with respect to the Lebesgue measure in $\RR_+ \times \bar{D} \times \RR^n$. Moreover,  $f \in C([0,\infty); L^1(\bar{D} \times \RR^n))$. The non-negative measures $\rho_{\pm}$ satisfy 
$$ \rho_{\pm}(dt,dx,dv) = \gamma_{\pm} f(t,x,v) |v \cdot n_x| dt dx dv \text{ on } (0, \infty) \times \partial_{\pm} G, $$
where $\gamma f \in L^1([0,T] \times \pD \times \RR^n, |v \cdot n_x|dv dx  dt)$ for all $T > 0$ is the trace measure of $f$ given by Theorem \ref{ThmMisch} and where we write $\gamma_{\pm} f$ for its restrictions to $(0, \infty) \times \partial_{\pm} G$.
\end{prop} 

Observe that we can indeed apply Theorem \ref{ThmMisch} because (i) $L^1([0,T] \times \pD \times \RR^n, |v \cdot n_x|dv dx  dt) \subset L^1_{\loc}([0,\infty) \times \pD \times \RR^n, (v \cdot n_x)^2dv dx  dt)$, and (ii) $Lf = 0$ in $\mathcal{D}'((0,\infty) \times D \times \RR^n)$ since $\rho$ is a weak solution to (\ref{Problem1}), see Step 7 below.

\begin{proof}

Recall that for $i \geq 1$, $\rho^i_+$ denotes the law of $(T_i, X_{T_i}, V_{T_i-})$, with the sequence $(T_i)_{i \geq 1}$ and the process $(X_t,V_t)_{t \geq 0}$, of Definition \ref{DefiProcess}. For all $i \geq 1$, we also write $\rho^i_-$ for the law of $(T_i, X_{T_i}, V_{T_i})$. 
\vspace{.5cm}

\textbf{Step 1.} We show that $\rho^1_+$ has a density with respect to $|v \cdot n_x| dv dx dt$.
For $A \in \CB(\RR_+ \times \partial_+ G)$,
\begin{align*}
\rho^1_+(A) &= \EE[\mathbf{1}_{\{(T_1, X_{T_1}, V_{T_1-}) \in A \}}] = \EE[\mathbf{1}_{\{(\zeta(X_0,V_0), q(X_0,V_0), V_{0}) \in A \}}] \\
&= \int_{D \times \RR^n} \mathbf{1}_{\{(\zeta(x,v), q(x,v),v) \in A\}} f_0(x,v) dvdx.
\end{align*}
For any fixed $v \in \RR^n$, the map $x \to \big(y = q(x,v), s = \zeta(x,v)\big)$ is a $C^1$ diffeomorphism from $D$ to $\{(y,s): y \in \pD, v \cdot n_y < 0, s \in [0, \zeta(y,-v))\}$, and the Jacobian is given by $|v \cdot n_y|$, see Lemma 2.3 of \cite{Esposito2013} where $\tau_b(x,v) = \zeta(x,-v)$ with our notations. 
Applying this change of variables, we obtain
\begin{align*}
\rho^1_+(A) &= \int_{\partial_+ G} \int_0^{\zeta(y,-v)} \mathbf{1}_{\{(s, y,v) \in A\}} f_0(y - sv,v) |v \cdot n_y| ds dv dy.
\end{align*}
Hence $\rho^1_+$ has a density with respect to the measure $|v \cdot n_x| dv dx dt$ on $\RR_+ \times \partial_+ G$. 

\vspace{.5cm} 

\textbf{Step 2.} We show that for all $i \geq 1$, assuming that $\rho^i_+$ has a density $g^i_+$, $\rho^i_-$ has a density $g^i_-$ with respect to the measure $|v \cdot n_x| dv dx dt$ on $\RR_+ \times \partial_- G$. For $A \in \CB(\RR_+ \times \partial_- G)$,
\begin{align*}
\rho^i_-(A) &= \EE[\mathbf{1}_{\{(T_i, X_{T_i}, V_{T_i}) \in A \}}] \\
&= \EE\Big[\alpha(X_{T_i})\mathbf{1}_{\{(T_i, X_{T_i}, R_i \vartheta(X_{T_i}, \Theta_i)) \in A \}}\Big] + \EE\Big[(1-\alpha(X_{T_i}))\mathbf{1}_{\{(T_i, X_{T_i}, \eta_{X_{T_i}}(V_{T_i-})) \in A \}}\Big],
\end{align*}
where we recall that $\eta_x(v) = v - 2(v \cdot n_x) n_x$.
We obtain, recalling Lemma \ref{NotatHM},
\begin{align*}
\rho^i_-(A)  &=  \int_{\partial_+ G} \int_{\RR_+} \alpha(x) \Big( \int_{\{ v' \cdot n_x > 0\}} \mathbf{1}_{\{(\tau, x, v') \in A\}} c_0 M(v') |v' \cdot n_x| dv' \Big) g^i_+(\tau, x, v) |v \cdot n_x| d\tau dv dx \\
&\quad + \int_{\partial_+ G} \int_{\RR_+} (1 - \alpha(x)) \mathbf{1}_{\{(\tau, x, \eta_x(v)) \in A\}} g^i_+(\tau, x, v) |v \cdot n_x| d\tau dv dx \\
&= \int_{\partial_- G} \int_{\RR_+}   \mathbf{1}_{\{(\tau, x, v') \in A\}} \Big( \alpha(x) c_0 M(v') \int_{\{v \cdot n_x < 0\}} g^i_+(\tau,x,v) |v \cdot n_x| dv  \Big) |v' \cdot n_x| d\tau dv' dx \\
&\quad + \int_{\partial_- G} \int_{\RR_+} \mathbf{1}_{\{(\tau, x, v) \in A\}} \Big( (1-\alpha(x)) g^i_+(\tau,x,v - 2(v \cdot n_x) n_x)\Big) |v \cdot n_x|  d\tau dv dx,
\end{align*}
where we have used that the change of variable $v \to (w = \eta_x(v))$ is involutive for any $x \in \pD$. 
We conclude that for $(t,x,v) \in (0,\infty) \times \partial_- G$,
$$ g^i_-(t,x,v) = \alpha(x)  c_0 M(v) \int_{\{v' \cdot n_x < 0\}} g^i_+(t,x,v') |v' \cdot n_x| dv'  + (1-\alpha(x)) g^i_+(t,x,v - 2(v \cdot n_x) n_x), $$

\noindent and therefore for all $i \geq 1$, $\rho^i_-$ has a density with respect to $|v \cdot n_x| dv dx dt$ on $\RR_+ \times \partial_- G$. 

\vspace{.5cm}

\textbf{Step 3.} We show that for all $i \geq 1$, for all $t \geq 0$, assuming that $\rho^i_-$ has a density $g^i_-$, the law $f^i_t$ of $(X_t, V_t)$ restricted to $(T_i, T_{i+1})$ has a density on $D \times \RR^n$ with respect to the Lebesgue measure. For $A \in \CB(D \times \RR^n)$,
\begin{align*}
f^i_t(A) &= \EE[\mathbf{1}_{\{(X_t, V_t) \in A\}} \mathbf{1}_{\{T_i < t < T_{i+1} \}}] \\
&= \EE[\mathbf{1}_{\{(X_{T_i} + (t-T_i)V_{T_i}, V_{T_i}) \in A\}} \mathbf{1}_{\{T_i < t < T_i + \zeta(X_{T_i},V_{T_i}) \}}] \\
&= \int_{\partial_- G} \int_0^t \mathbf{1}_{\{(x + (t-\tau)v, v) \in A\}} \mathbf{1}_{\{ \tau < t < \tau + \zeta(x,v) \}} |v \cdot n_x| g^i_-(\tau,x,v) d\tau dv dx.
\end{align*}
For any fixed $v \in \RR^n$, $t \in (0, \infty)$, $(x,\tau) \to (y = x + (t - \tau)v)$, is a $C^1$-diffeomorphism from $\{(x,\tau) \in \pD \times (0, \infty): v \cdot n_x > 0$, $ \tau < t < \tau + \zeta(x,v)\}$ to $D$ such that $x = q(y, -v)$, $t-\tau = \zeta(y,-v)$ and is the inverse of the $C^1$-diffeomorphism of Step 1. Hence, its Jacobian is given by $\frac{1}{|v \cdot n_x|} \ne 0$, and we obtain,
\begin{align*}
f^i_t(A) = \int_{D \times \RR^n} \mathbf{1}_{\{(y, v) \in A\}} g^i_-(t - \zeta(y,-v),q(y,-v),v) dy dv,
\end{align*}
and therefore $f^i_t$ has a density $g^i_t$ over $D \times \RR^n$. 

\vspace{.5cm}
\textbf{Step 4.} One easily shows that for all $t \geq 0$, $f^0_t$, the law of $(X_t,V_t)$ restricted to $[0,T_1)$ also has a density with respect to the Lebesgue measure. Indeed, it is enough to write, for any $A \in \mathcal{B}(D \times \RR^n)$, 
$$ f_t^0(A) = \EE[\mathbf{1}_{\{t < T_1, (X_0 + t V_0, V_0) \in A\}}], $$
and to use that $(X_0,V_0)$ has a density.

\vspace{.5cm}
\textbf{Step 5.} We now prove that, for all $i \geq 0$, if $f^i_t$ has a density $g^i_t$ for all $t \geq 0$, then $\rho^{i+1}_+$ has a density with respect to the measure $|v \cdot n_x| dt dv dx$ on $\RR_+ \times \partial_+ G$. 
For $A \in \CB(\RR_+ \times \partial_+ G)$, 
\begin{align*}
\rho^{i+1}_+(A) &= \EE[\mathbf{1}_{\{(T_{i+1}, X_{T_{i+1}}, V_{T_{i+1}-}) \in A\}}] \\
&= \EE \Big[\int_{T_i}^{T_{i+1}} \mathbf{1}_{\{(T_{i+1}, X_{T_{i+1}}, V_{T_{i+1}-}) \in A\}} \frac{1}{T_{i+1} - T_i} dt \Big]\\
&= \int_0^{\infty} \EE \Big[\mathbf{1}_{\{T_i < t < T_{i+1}\}} \mathbf{1}_{\{(t + \zeta(X_t, V_t), q(X_t,V_t), V_t) \in A\}} \frac{1}{t+\zeta(X_t,V_t) - (t-\zeta(X_t,-V_t))} \Big] dt \\
&= \int_0^{\infty} \int_{D \times \RR^n} \mathbf{1}_{\{(t + \zeta(x,v), q(x,v),v) \in A\}} \frac{1}{t+\zeta(x,v) - (t-\zeta(x,-v))} g^i_t(x,v) dvdx dt.
\end{align*}
We used that for $t \in (T_i, T_{i+1}), T_i = t - \zeta(X_t, -V_t)$, $T_{i+1} = t + \zeta(X_t,V_t)$, $X_{T_{i+1}} = q(X_t,V_t)$ and $V_{T_{i+1}-} = V_t$.
We use a slightly modified change of variables compared to Step 1: for a fixed $t \in \RR^+$ and a fixed $v \in \RR^n$, we consider $x \to \big(y = q(x,v), \tau = t + \zeta(x,v)\big)$. This diffeomorphism from $D$ to $\{(y,\tau) \in \pD \times (0, \infty): v \cdot n_y < 0, t < \tau < t + \zeta(y,-v)\}$ has a Jacobian equal to $|v \cdot n_y|$, as in Step 1. Therefore, since $\zeta(x,v) + \zeta(x,-v) = \zeta(y, -v)$ and $x = y - (\tau-t)v$,
\begin{align*}
\rho^{i+1}_+(A) &= \int_0^{\infty} \int_{\partial_+ G} \int_{0}^{\infty} \mathbf{1}_{\{(\tau, y,v) \in A\}} \frac{1}{\zeta(y,-v)} \mathbf{1}_{\{\tau - \zeta(y,-v) < t < \tau\}} g^i_t(y-(\tau-t)v,v) |v \cdot n_y| d\tau dv dy  dt \\
&= \int_{\partial_+ G} \int_0^{\infty} \mathbf{1}_{\{(\tau, y,v) \in A\}} |v \cdot n_y| \frac{1}{\zeta(y,-v)} \Big(\int_{\tau - \zeta(y,-v)}^{\tau} g^i_t(y-(\tau-t)v,v) dt \Big) d\tau dv dy,
\end{align*}
and this shows that $\rho^{i+1}_+$ has a density with respect to the measure $|v \cdot n_x| dv dx dt$ on $(0,\infty) \times \partial_+ G$.

\vspace{.5cm}

\textbf{Step 6.} From Steps 1 to 5, we conclude that for all $i \geq 1$, $\rho^i_{\pm}$ have a density $g^{i}_{\pm}$ with respect to the measure $|v \cdot n_x| dv dx dt$ on $(0,\infty) \times \partial_{\pm} G$.  Thus, $\rho_{\pm} = \sum_{i \geq 1} \rho^i_{\pm}$ also have a density with respect to $|v \cdot n_x| dv dx dt$ on $(0,\infty) \times \partial_{\pm} G$ that we write $g_{\pm}$. The function defined by
\begin{align}
\label{EqDefTraceMeasure}
g(t,x,v) = g_+(t,x,v) \mathbf{1}_{\{v \cdot n_x < 0\}} + g_-(t,x,v) \mathbf{1}_{\{v \cdot n_x > 0\}}, \quad (t,x,v) \in \RR_+ \times \pD \times \RR^n,
\end{align}
belongs to $L^1([0,T) \times \pD \times \RR^n, |v \cdot n_x| dt dx dv)$ for all $T > 0$, because 
\begin{align*}
\rho_{\pm}([0,T] \times \partial_{\pm}G) = \EE[\#\{i: T_i \leq T\}] < \infty,
\end{align*}
by Proposition \ref{PropExplosion}. Consequently, 
$g$ belongs to $L^1_{\loc}([0,T) \times \pD \times \RR^n, |v \cdot n_x|^2 dt dx dv)$. A second conclusion from those steps is that the measure $f_t$ has a density on $D \times \RR^n$ for all $t \geq 0$. Hence $\rho$ has a density $f$ on $\RR_+ \times \bar{D} \times \RR^n$. 

\vspace{.5cm}

\textbf{Step 7.} Note that, because $\rho(dt,dx,dv) = f(t,x,v)dtdxdv$ satisfies (\ref{GreenResult}), we obviously have that $f$ satisfies 
$$ Lf = 0 \in \CD'((0,\infty) \times D \times \RR^n). $$
Using Theorem \ref{ThmMisch}, we conclude that $f \in C([0,\infty); L^1_{\loc}(\bar{D} \times \RR^n))$, and then to $C([0,\infty); L^1(\bar{D} \times \RR^n))$ since for all $t \geq 0$, $f(t,.)$ is a probability density.

\vspace{.5cm}

\textbf{Step 8.} There only remains to prove that the function $g$ defined by (\ref{EqDefTraceMeasure}) is the trace of $f$ in the sense of Theorem \ref{ThmMisch}.
We want to show that for any $0 \leq t_0 < t_1$, any $\phi \in \CD((0, \infty) \times \bar{D} \times \RR^n)$ such that $\phi = 0$ on $(0, \infty) \times \partial_0 G$, we have
\begin{align*}
\int_{t_0}^{t_1} \int_G f L\phi dv dx dt = \Big[\int_G f \phi dv dx \Big]^{t_1}_{t_0} - \int_{t_0}^{t_1}  \int_{\pD \times \RR^n} g(t,x,v)  (n_x \cdot v) \phi dv dx dt.
\end{align*}
By substraction, this can be reduced to proving that 
\begin{align}
\label{EqSubstraction}
\int_{0}^{t_1} \int_G f L\phi dv dx dt = \int_{G} f(t_1,x,v) \phi(t_1,x,v) dx dv -  \int_0^{t_1} \int_{\pD \times \RR^n} g(t,x,v)  (n_x \cdot v) \phi dv dx dt,
\end{align}
for any $t_1 > 0$, any $\phi \in \CD((0, \infty) \times \bar{D} \times \RR^n)$, $\phi = 0$ on $(0, \infty) \times \partial_0 G$.  

For any $\epsilon \in (0,1)$, let $\beta_{\epsilon}(t) = \mathbf{1}_{(0,t_1)}(t) + e^{-\frac{t-t_1}{\epsilon + t_1 - t}} \mathbf{1}_{[t_1,t_1 + \epsilon)}(t)$. Applying (\ref{GreenResult}) with the test function $\beta_{\epsilon} \phi$, recalling that $\rho(dt,dx,dv) = f(t,x,v) dt dx dv$, $\rho_{\pm}(dt,dx,dv) = g_{\pm}(t,x,v) |v \cdot n_x|$ so that $(\rho_+ - \rho_-) (dt,dx,dv) = - g(t,x,v) (v \cdot n_x) dt dx dv$, we find 
\begin{align*}
\int_0^{\infty} \int_{G} \beta_{\epsilon}' f \phi dv dx dt + \int_0^{\infty} \int_G \beta_{\epsilon} f L\phi dv dx dt = - \int_0^{\infty} \int_{\pD \times \RR^n} g (v \cdot n_x) \beta_{\epsilon} \phi dv dx dt.
\end{align*}
We rewrite this equation as
\begin{align*}
A_{\epsilon} + B + C_{\epsilon} = - D_{\epsilon},
\end{align*}
by setting
\begin{align*}
A_{\epsilon} &= \int_{t_1}^{t_1 + \epsilon} \int_G \beta_{\epsilon}'(t) \Big(f(t,x,v) \phi(t,x,v) - f(t_1,x,v) \phi(t_1,x,v)\Big) dv dx dt, \\
B &= - \int_{D \times \RR^n} f(t_1,x,v) \phi(t_1,x,v) dv dx , \\
C_{\epsilon} &= \int_0^{\infty} \int_G \beta_{\epsilon} f L\phi dv dx dt, \\
D_{\epsilon} &= \int_0^{\infty} \int_{\pD \times \RR^n} g (v \cdot n_x) \beta_{\epsilon} \phi dv dx dt,
\end{align*}
where we used that $\int_{t_1}^{t_1 + \epsilon} \beta'_{\epsilon}(t) dt = -1$.
We have
\begin{align*}
|A_{\epsilon}| \leq \sup_{t \in [t_1, t_1 + \epsilon]} \Big|\int_{D \times \RR^n}  \big(f(t,x,v) \phi(t,x,v) - f(t_1,x,v) \phi(t_1,x,v) \big) dv dx \Big| \times  \int_{t_1}^{t_1 + \epsilon} |\beta'_{\epsilon}(t)| dt.
\end{align*}
Hence $A_{\epsilon} \to 0$ as $\epsilon \to 0$, because $f \in C([0,\infty), L^1(\bar{D} \times \RR^n))$ and by regularity of $\phi$, see Step 7. 

Since $\beta_{\epsilon}(t) \leq 1$ for all $t \geq 0$, since $f \in L^1_{\loc}(\RR_+ \times \bar{D} \times \RR^n)$, by regularity of $\phi$, and since $\beta_{\epsilon}(t) \to \mathbf{1}_{[0,t_1]}(t)$, a straightforward application of the dominated convergence theorem gives that $C_{\epsilon} \to \int_0^{t_1} \int_G f L\phi dv dx dt$ as $\epsilon \to 0$. 

The same argument, along with the fact that $g \in L^1((0,T) \times \pD \times \RR^n, |v \cdot n_x| dt dv dx)$ allows us to conclude that
\begin{align*}
\lim \limits_{\epsilon \to 0} D_{\epsilon} =  \int_0^{t_1} \int_{\pD \times \RR^n} g \phi (v \cdot n_x) dv dx dt.
\end{align*}
Overall, we obtain that $g$ satisfies (\ref{EqSubstraction}) for any $t_1 \geq 0$, any $\phi \in \CD((0,\infty) \times \bar{D} \times \RR^n)$ with $\phi = 0$ on $(0,\infty) \times \partial_0 G$, so that $g$ is the trace of $f$ in the sense of Theorem \ref{ThmMisch}.

\end{proof}

\section{The convex case}
\label{SectionConvexCase}

In this section, we prove Theorem \ref{MainTheorem} in the easier case where $D$ is a $C^2$ uniformly convex bounded domain (open, connected) in $\RR^n$.

The strategy is to build a coupling of two stochastic processes with the dynamic of Definition \ref{DefiProcess},  $(X_t,V_t)_{t \geq 0}$ with initial distribution $f_0$, $(\xtilde{t}, \vtilde{t})_{t \geq 0}$ with initial distribution $\mu_{\infty}$, where $\mu_{\infty}$ is the equilibrium distribution. For this couple of processes, two different regimes can be identified: a low-speed regime and a high-speed regime.

In a first step, we collect several results on the high-speed regime. In this situation, we find a coupling which is successful, in a sense to be defined, with a probability admitting a positive lower bound. In a second step, we detail the construction of the processes. 
Finally, we prove that
\begin{align}
\label{EqDefTau}
\tau = \inf\{t \geq 0: (X_{t + s})_{s \geq 0} = (\xtilde{t+ s})_{s \geq 0}, (V_{t + s})_{s \geq 0} = (\vtilde{t+ s})_{s \geq 0}\},
\end{align}  
satisfies $\EE[r(\tau)] < \infty$.

\subsection{A coupling result.}
\label{SubsectionCoupling}

Recall the notations $h_R$, $\Upsilon$ introduced in Lemma \ref{NotatHM}. Since $M$ admits a density, there exists $a > 0$ such that, 
\begin{align}
\label{EqHypoM1}
\int_0^a h_R(x) dx > 0, \qquad \int_a^{\infty} h_R(x) dx > 0, 
\end{align}
and we assume for simplicity that $a = 1$ in the sequel. Recall also that $\mathcal{A} = (-\frac{\pi}{2}, \frac{\pi}{2}) \times [0,\pi]^{n-2}$, and $d(D) := \sup \limits_{(x,y) \in D^2} \|x-y\|$, which corresponds to the diameter of $D$. We introduce some more notations. 

\begin{notat}
\label{Notationyxi}
\noindent We define four maps:
\begin{enumerate}[i.]
\item the map $\xi: \pD \times \RR_+ \times \mathcal{A} \to \RR_+$, such that $$\xi(x,r,\theta) = \zeta(x, r\vartheta(x, \theta)),$$
\item the map $y: \pD \times \mathcal{A} \to \pD$, such that $$y(x,\theta) = q(x, \vartheta(x, \theta)),$$
\item the map $\tilde{\xi}: \bar{D} \times \RR^n \times \RR_+ \times \mathcal{A} \to \RR_+$, such that $$\tilde{\xi}(x,v,r,\theta) = \zeta(x,v) + \zeta\Big(q(x,v), r\vartheta(q(x,v), \theta)\Big),$$
\item the map $\tilde{y}: \bar{D} \times \RR^n \times \mathcal{A} \to \pD$, such that $$\tilde{y}(x,v,\theta) = q\Big(q(x,v), \vartheta(q(x,v), \theta)\Big).$$
\end{enumerate}
\end{notat}

\noindent The main result in this section is the following proposition:

\begin{prop}
\label{CouplingProp}
There exists a constant $c > 0$ such that for all $x_0 \in \pD$, $ \tilde{x}_0 \in D$, $\tilde{v}_0 \in \RR^n$ with $\|\tilde{v}_0\| \geq 1$, there exists $\Lambda_{x_0, \tilde{x}_0, \tilde{v}_0} \in \mathcal{P}(((0,\infty) \times \mathcal{A})^2)$ such that, if $(R, \Theta, \tilde{R}, \tilde{\Theta})$ has law $\Lambda_{x_0, \tilde{x}_0, \tilde{v}_0}$, both $(R, \Theta)$ and $(\tilde{R}, \tilde{\Theta})$ have law $\Upsilon$, and for
\begin{align*}
E_{x_0, \tilde{x}_0, \tilde{v}_0} := \Big\{(r, \theta, \tilde{r}, \tilde{\theta}) \in (\RR_+ \times \mathcal{A})^2:  y(x_0,\theta) = \tilde{y}(\tilde{x}_0, \tilde{v}_0, \tilde{\theta}), \xi(x_0,r,\theta) = \tilde{\xi}(\tilde{x}_0, \tilde{v}_0, \tilde{r}, \tilde{\theta})\Big\},
\end{align*} 
we have
\begin{align}
\label{IneqCouplingLemma} 
\PP \big((R,\Theta, \tilde{R}, \tilde{\Theta}) \in E_{x_0, \tilde{x}_0, \tilde{v}_0} \big) \geq c.
\end{align}
\end{prop}

The rest of this subsection is devoted to the proof of this proposition. 

\vspace{.5cm}

\begin{lemma}
\label{LemmaPositiveMeasure}
There exist two constants $r_1 > 0$ and $ c_1 > 0$ such that for all $(x,y) \in (\pD)^2$,
\begin{align}
\label{EqLemmaPositiveMeasure}
 \int_{\{ \scriptscriptstyle{z \in \pD, \|z-x\| \wedge \|z-y\| \geq r_1}\} } \big(|(z-x) \cdot n_x|| (z-x) \cdot n_z|\big) \wedge \big(| (z-y) \cdot n_y|| (z-y) \cdot n_z| \big) dz  \geq c_1 .
\end{align}
\end{lemma}

\begin{proof}

Without loss of generality we assume that $0 \in D$. Recall that we write $\mathcal{H}$ for the $n-1$ dimensional Hausdorff measure. 

We show first that there exists $c > 0$ such that for all $(x,y) \in (\pD)^2 $, $\mathcal{H}(A_{x,y}) \geq c$, where $A_{x,y} := \{z \in \pD, \|z-x\| \wedge \|z-y\| \geq r_1\}$ for some $r_1 > 0$. Set $r_0 := \inf_{z \in \pD} \|z\|$.

Note that for all $(x,y) \in (\pD)^2$, for $\delta \in (0,1)$, with $r_1 = r_0 \sqrt{2 - 2 \delta}$, we have the inclusion $A_{x,y} \subset A'_{x,y} := \{z \in \pD, \frac{z}{\|z\|} \cdot \frac{x}{\|x\|} < \delta, \frac{z}{\|z\|} \cdot \frac{y}{\|y\|} < \delta\}$ since for all $z \in A'_{x,y}$, 
\begin{align}
\label{IneqrDLemmaMeasure}
 \|x - z\|^2 \geq \|x\|^2 - 2\delta \|x\|\|z\| + \|z\|^2 \geq  (\|x\| - \|z\|)^2 + (2 - 2\delta)\|z\|\|x\| \geq r_1^2,
 \end{align}
 and $\|y - z\| \geq r_1$ as well. 
 
 Let $\phi: \RR^n \to \RR^n$ defined by $\phi(x) = \frac{x}{(2\|x\|) \vee r_0} r_0$ for any $x \in \RR^n$. Note that $\phi$ is the projection on the closed ball $\bar{B}(0,\frac{r_0}{2}) := \{z \in \RR^n, \|z\| \leq \frac{r_0}{2}\}$ and is thus $1$-Lispschitz. 
By definition of $r_0$, setting $S:= \{y \in \RR^n, \|y\| = \frac{r_0}{2}\}$, we have $\phi(\pD) = S$.

We apply the following statement: for $m \in \mathbb{N}^*$, for any Lipschitz map $f: \RR^m \to \RR^m$ with Lipschitz constant $L > 0$, for any $A \subset \RR^m$, \begin{align}
\label{EqMattila}
\mathcal{H}(f(A)) \leq L^m \mathcal{H}(A),
\end{align} see \cite[Theorem 7.5] {mattila_1995}. We obtain that
$$ \mathcal{H} \Big(\phi(A'_{x,y}) \Big) \leq \mathcal{H}(A'_{x,y}). $$

Observe that $\phi(A'_{x,y}) = \{u \in S, \frac{u}{\|u\|} \cdot \frac{x}{\|x\|} < \delta, \frac{u}{\|u\|} \cdot \frac{y}{\|y\|} < \delta \} $ so that
$$ \mathcal{H}\Big(\phi(A'_{x,y}) \Big) \geq \mathcal{H}(S) - 2 \mathcal{H}\Big(\Big\{u \in S, \frac{u \cdot e_1}{\|u\|} < \delta \Big\} \Big) \geq \frac{1}{2} \mathcal{H}(S), $$
if $\delta < \delta_0$ for some $\delta_0 > 0$ not depending on $x$ and $y$, since $\mathcal{H}(\{u \in S, \frac{u \cdot e_1}{\|u\|} < \delta\})$ converges to $0$ when $\delta$ goes to $0$. 

To conclude, it suffices to use that $$\inf_{(a,b) \in (\pD)^2, \|a-b\| \geq r_1} |(a-b) \cdot n_a| > 0,$$ which follows by compactness from the fact that $D$ is $C^1$, bounded and uniformly convex. 
\end{proof}

Recall that the constant $c_0$ is defined by (\ref{EqCM}).

\begin{lemma}
\label{LemmaDensityCoupling}
For $x \in \pD$ and $V$ having density $c_0 M(v) |v \cdot n_x| \mathbf{1}_{\{v \cdot n_x > 0\}}$, the law of $(\zeta(x,V), q(x,V))$ admits a density $\mu_x$ on $\RR_+^* \times (\pD \setminus \{x\})$ given by
$$ \mu_x(\tau,z) = c_0 M\Big(\frac{z-x}{\tau}\Big) \frac{1}{\tau^{n+2}} |(z-x) \cdot n_x| | (z-x) \cdot n_z|. $$
\end{lemma}

\begin{proof}
Let $A \in \mathcal{B}(\RR_+ \times (\pD \setminus\{x\}))$. We have
\begin{align}
\label{EqProbaLemmaDensity}
\PP\Big((\zeta(x,V),q(x,V)) \in A\Big) = \int_{\scriptstyle{\{v \cdot n_x > 0\}}} \mathbf{1}_{\{(\zeta(x,v), q(x,v)) \in A\}} c_0 M(v) |v \cdot n_x| dv.
\end{align}
We show that this quantity is equal to
$$I := \int_0^{\infty} \int_{\pD} \mathbf{1}_{\{(\tau,z) \in A\}} c_0 M\Big(\frac{z-x}{\tau}\Big) \frac{1}{\tau^{n+2}} |(z-x) \cdot n_x| | (z-x) \cdot n_z| dz d\tau. $$
Consider the change of variable $(\tau,z) \to v$ given by $v = \frac{z-x}{\tau} =: \phi(\tau,z)$.
 Note that by uniform convexity, we have $v \cdot n_x > 0$ and $(\tau,z) = (\zeta(x,v),q(x,v))$.
  The map $\phi$ is a $C^1$ diffeomorphism between $\RR_+ \times (\pD \setminus \{x\})$ and $\{v \in \RR^n, v \cdot n_x > 0\}$. Note that
\begin{enumerate}
\item the tangent space to $\RR_+$ at $\tau \in \RR_+$ is $\RR$,
\item the tangent space to $\pD \setminus \{x\}$ at $z \in \pD \setminus \{x\}$ is $n_z^{\perp} \subset \RR^n$,
\item the tangent space to $\{v \in \RR^n, v \cdot n_x > 0\}$ at $v$ is $\RR^n$.
\end{enumerate}  
  
  For $(\tau,z) \in \RR_+ \times (\pD \setminus \{z\})$, the differential of $\phi$ in the direction $(s,y)$ with $s \in \RR$, $y \in n_z^{\perp}$ is given by $$D\phi_{(\tau,z)}(s,y) = \frac{y}{\tau} - \frac{(z-x) s}{\tau^2}.$$
Let $(f_1, \dots, f_{n-1})$ be an orthonormal basis of $n_z^{\perp}$, $f_n$ such that $(f_1, \dots, f_{n-1},f_n)$ is an orthonormal basis of $n_z^{\perp} \times \RR$. The Jacobian matrix of $\phi$ in the bases $(f_1, \dots, f_n)$ for $n_z^{\perp} \times \RR$ and $(f_1, \dots, f_{n-1}, n_z)$ for $\RR^n$ is thus
\begin{align*}
J_{\phi}(\tau,z)= \begin{pmatrix}
\frac{1}{\tau} & 0 &\dots & 0 & -\frac{(z-x) \cdot f_1}{\tau^2}  \\
0 & \frac{1}{\tau} & \dots & 0 & -\frac{(z-x) \cdot f_2}{\tau^2} \\
 & \dots & \dots \\
 0 & \dots & \dots & \frac{1}{\tau} & -\frac{(z-x)\cdot f_{n-1}}{\tau^2}  \\
 0 & \dots & 0 & 0 & -\frac{(z-x) \cdot n_z}{\tau^2}
\end{pmatrix}.
\end{align*} 
The Jacobian at the point $(\tau,z)$ is therefore given by $\frac{|(z-x) \cdot n_z|}{\tau^{n+1}}$.

Recalling (\ref{EqProbaLemmaDensity}), using that $(\tau,z) = (\zeta(x,v),q(x,v))$, we find
\begin{align*}
I &=  \int_{\scriptstyle{\{v \cdot n_x > 0\}}} \mathbf{1}_{\{(\zeta(x,v), q(x,v)) \in A\}} c_0 M(v) \frac{1}{\zeta(x,v)} \big|\zeta(x,v) (v \cdot n_x) \big|  dv = \PP \Big((\zeta(x,V), q(x,V) \in A \Big),
\end{align*}
as desired.
\end{proof}

With the help of Lemmas \ref{LemmaPositiveMeasure} and \ref{LemmaDensityCoupling}, we prove Proposition \ref{CouplingProp}.

\begin{proof}[Proof of Proposition \ref{CouplingProp}]

In a first step, we derive an inequality from which we will conclude in the second step, using the classical framework of maximal coupling.

\vspace{.5cm}

\textbf{Step 1.} We show that, for $A = (\pD)^2 \times [0,d(D))$, there exists $c > 0$ such that 
\begin{align}
\label{IneqMaxCouplingNew}
\inf_{(x,\tilde{x}, \tilde{t}) \in A} \int_{ \pD} \int_{\tilde{t}}^{\infty} [ \mu_x(\tau,z) \wedge \mu_{\tilde{x}}(\tau-\tilde{t}, z) ]  d\tau  dz \geq c. 
\end{align}
We have, using Lemma \ref{LemmaDensityCoupling}, for any $(x,\tilde{x}, \tilde{t}) \in A$,
\begin{align*}
J &:= \int_{\pD} \int_{\tilde{t}}^{\infty} [\mu_x(\tau,z) \wedge \mu_{\tilde{x}}(\tau-\tilde{t}, z)] d\tau dz  \\
&\geq c_0 \int_{\scriptstyle{\{z \in \pD, \|z-x\| \wedge \|z-\tilde{x}\| \geq r_1\}}} \int_{b_0}^{b_1} \Big( \Big[ M\Big(\frac{z-x}{\tau}\Big) \frac{1}{\tau^{n+2}} |(z-x) \cdot n_x| | (z-x) \cdot n_z| \Big] \\
&\qquad  \wedge  \Big[ M\Big(\frac{z-\tilde{x}}{\tau - \tilde{t}}\Big) \frac{1}{(\tau - \tilde{t})^{n+2}} |(z-\tilde{x}) \cdot n_{\tilde{x}}| |(z-\tilde{x}) \cdot n_z| \Big] \Big) d\tau  dz,
\end{align*}
 where $b_0 = d(D) (\frac{\delta_1 + 1}{\delta_1})$ and $b_1 = d(D) (\frac{2\delta_1 + 1}{\delta_1})$, recalling the definition of $\delta_1$ from Hypothesis \ref{HypoM}. Indeed, $b_0 \geq d(D) \geq \tilde{t}$. For $\tau \in (b_0, b_1)$, $z, y \in \pD$ with $\|z-y\| \geq r_1$, we have
 $$0 < \frac{r_1}{b_1} \leq \frac{\|z-y\|}{\tau} \leq \frac{\|z-y\|}{\tau - \tilde{t}} \leq \delta_1 \frac{\|z-y\|}{d(D)} \leq \delta_1, $$ whence, recalling the definition of $\bar{M}$ from Hypothesis \ref{HypoM},
$$ M \Big(\frac{z-x}{\tau}\Big) \wedge M \Big(\frac{z- \tilde{x}}{\tau - \tilde{t}}\Big) \geq \bar{M} \Big(\frac{z-x}{\tau}\Big) \wedge \bar{M}\Big(\frac{z- \tilde{x}}{\tau - \tilde{t}}\Big) \geq \kappa_1,$$
where $\kappa_1 = \underset{\frac{r_1}{b_1} \leq \|v\| \leq \delta_1}{\min} \bar{M}(v) > 0$ not depending on $(x,\tilde{x},\tilde{t})$. We obtain, using Tonelli's theorem, that 
\begin{align*}
J &\geq  c_0 \kappa_1 \int_{b_0}^{b_1} \frac1{\tau^{n+2}} d\tau \\
& \quad \times \int_{\{z \in \pD, \|z-x\| \wedge \|z-\tilde{x}\| \geq  r_1\}}  \Big[ |(z-x) \cdot n_x| |(z-x) \cdot n_z| \Big] \wedge  \Big[|(z-\tilde{x}) \cdot n_{\tilde{x}}| |(z-\tilde{x}) \cdot n_z| \Big] dz. 
\end{align*}
We conclude by applying Lemma \ref{LemmaPositiveMeasure}.

\vspace{.5cm}

\textbf{Step 2.} Recall that $x_0 \in \pD$, $\tilde{x}_0 \in D$, $\tilde{v}_0 \in \RR^n$ such that $\|\tilde{v}_0\| \geq 1$ are fixed. Set $x = x_0$, $\tilde{x} = q(\tilde{x}_0,\tilde{v}_0)$ and $\tilde{t} = \zeta(\tilde{x}_0,\tilde{v}_0) \leq \frac{d(D)}{\|\tilde{v}_0\|} \leq d(D)$, .  Classicaly, using (\ref{IneqMaxCouplingNew}), one can couple $(S, Y) \sim \mu_{x}$ and $(\tilde{S}, \tilde{Y}) \sim \mu_{\tilde{x}}$ so that $\mathbb{P}(Y = \tilde{Y}, S = \tilde{S} + \tilde{t})  \geq c$. Recalling that, if $(R,\Theta) \sim \Upsilon$ and $(\tilde{R}, \tilde{\Theta}) \sim \Upsilon$, $(\xi(x,R,\Theta), q(x,\Theta)) \sim \mu_x$ and $(\tilde{\xi}(\tilde{x}_0, \tilde{v}_0, \tilde{R}, \tilde{\Theta}) - \tilde{t}, \tilde{y}(\tilde{x}_0,\tilde{v}_0, \tilde{\Theta})) \sim \mu_{\tilde{x}}$, the conclusion follows.
\end{proof}

\subsection{Some more preliminary results.}

Recall that the function $r: \RR_+ \to \RR_+$ is non-decreasing, continuous, and that there exists $C > 0$ satisfying, for all $(x,y) \in (\RR_+)^2$, $r(x+y) \leq C( r(x) + r(y) ).$

\begin{rmk}
\label{RmkPolynomialR}
There exist $C > 0$, $\beta > 0$ such that for all $n \geq 1$, for all $x_1, \dots, x_n \geq 0$, 
\begin{align}
\label{EqConclusionRmkBeta}
r\Big(\sum_{i = 1}^n x_i \Big) \leq Cn^{\beta} \sum_{i = 1}^n r(x_i).
\end{align} 
\end{rmk}

\begin{proof}
If $n = 2^p$, $p \in \mathbb{N}$, we have
\begin{align*}
r\Big(\sum_{i = 1}^{2^p} x_i \Big)\leq C^p \sum_{i=1}^{2^p} r(x_i). 
\end{align*}
In the general case, setting $x_j = 0$ for all $j \in \{1, \dots, 2^{[\log_2(n)] + 1}\} \setminus \{1, \dots n\}$, we obtain
\begin{align*}
r\Big(\sum_{i = 1}^n x_i \Big) = r\Big(\sum_{i = 1}^{2^{[\log_2(n)] + 1}} x_i\Big) &\leq C^{[\log_2(n)] + 1} \Big( \sum_{i = 1}^n r(x_i) + (2^{[\log_2(n)] + 1}-n)r(0) \Big) \\
&\leq 2C n^{\log_2(C)} \sum_{i = 1}^n r(x_i),
\end{align*}
where we used that $r(0) \leq r(x_i)$, that $2^{[\log_2(n)] + 1} - n \leq n$, and that $C^{[\log_2(n)] + 1} \leq C n^{\log_2(C)}$.
\end{proof}

\begin{lemma}
\label{LemmaSumGeometric}
Let $(\CG_k)_{k \geq 0}$ be a non-decreasing family of $\sigma$-algebras, $(\tau_k)_{k \geq 1}$ a family of random times such that $\tau_k$ is $\CG_{k}$-measurable for all $k \geq 1$. Let $(E_k)_{k \geq 1}$ a family of events such that for all $k \geq 1,$ $E_k \in \CG_k$ and assume there exists $c > 0$ such that a.s.
\begin{align}
\label{HypoProbaEk}
\forall k \geq 1, \quad \PP(E_k|\CG_{k-1}) \geq c.
\end{align}
Set $G = \inf\{k \geq 1, E_k \text{ is realized}\}$, which is almost surely finite. Assume there exists a positive $\mathcal{G}_0$-measurable random variable $L$ such that for all $k \geq 1$, (note that $\{G \geq k\} \in \CG_{k-1}$),
\begin{align}
\label{IneqHypoLemmaSumGeom}
 \mathbf{1}_{\{G \geq k\}} \EE[r(\tau_{k+1}-\tau_k)|\CG_{k-1}] \leq L \quad \text{ and } \quad \EE[r(\tau_1)|\CG_0] \leq L.
 \end{align}
Then
$$ \EE[r(\tau_G) | \CG_0] \leq \kappa L, $$
for some constant $\kappa > 0$ depending only on $c$ and the function $r$.
\end{lemma}

\begin{proof}
For all $j \geq 1$, on $\{G=j\}$, setting $\tau_0 = 0$, we have $\tau_G = \sum_{i = 0}^{j-1} (\tau_{i+1}-\tau_i)$. Hence, using (\ref{EqConclusionRmkBeta}),
\begin{align}
\label{EqLemmaGeom}
\EE \Big[r \big(\tau_G \big) \Big| \CG_0 \Big] &= \sum_{j = 1}^{\infty} \EE \Big[ r \Big( \sum_{i = 0}^{j-1} (\tau_{i+1} - \tau_i) \Big) \mathbf{1}_{\{G = j\}} \Big| \CG_0 \Big] \nonumber \\
 &\leq C \sum_{j=1}^{\infty} j^{\beta} \sum_{i=0}^{j-1} \EE \Big[ r(\tau_{i+1}-\tau_i) \Big( \prod_{k = 1}^{j-1} \mathbf{1}_{E_k^c} \Big) \mathbf{1}_{E_j} \Big| \CG_0 \Big] =  C\sum_{j=1}^{\infty} j^{\beta} \sum_{i=0}^{j-1} u_{i,j},
\end{align}
the last equality standing for the definition of $u_{i,j}$. By convention, we give the value 1 to any product indexed by the empty set. Note that for any $ l \geq m \geq 1$, using (\ref{HypoProbaEk}),
\begin{align*}
\EE \Big[ \Big(\prod_{k = m}^l \mathbf{1}_{E_k^c} \Big) \Big| \CG_{m-1} \Big] = \EE \Big[ \Big(\prod_{k = m}^{l-1} \mathbf{1}_{E_k^c} \Big) \EE[\mathbf{1}_{E_l^c}|\CG_{l-1}] \Big| \CG_{m-1} \Big] &\leq (1-c) \EE \Big[ \Big(\prod_{k = m}^{l-1} \mathbf{1}_{E_k^c} \Big) \Big| \CG_{m-1} \Big].
\end{align*}
Iterating the argument, 
\begin{align}
\label{IneqCondLemmaSumGeom}
 \EE \Big[ \Big(\prod_{k = m}^l \mathbf{1}_{E_k^c} \Big) \Big| \CG_{m-1} \Big] \leq (1-c)^{l-m+1}. 
 \end{align}
We first bound $u_{i,j}$ in the case where $i \geq 1$ and $j \geq i+2$. We have, using that $\mathbf{1}_{E_j} \leq 1$ and that $\{G \geq i\}$ on $E_1^c \cap \dots \cap E_{j-1}^c$,
\begin{align*}
u_{i,j} &\leq \EE \Big[ r(\tau_{i+1}-\tau_i) \Big( \prod_{k = i+2}^{j-1} \mathbf{1}_{E_k^c} \Big) \Big( \prod_{k = 1}^{i+1} \mathbf{1}_{E_k^c} \Big) \mathbf{1}_{\{G \geq i\}} \Big| \CG_0 \Big] \\
&\leq \EE \Big[ r(\tau_{i+1}-\tau_i)  \Big( \prod_{k = 1}^{i+1} \mathbf{1}_{E_k^c} \Big) \mathbf{1}_{\{G \geq i\}} \EE \Big[  \prod_{k = i+2}^{j-1} \mathbf{1}_{E_k^c} \Big| \CG_{i+1} \Big] \Big| \CG_0 \Big] \\
&\leq (1-c)^{j-i-2} \EE \Big[ r(\tau_{i+1}-\tau_i)  \Big( \prod_{k = 1}^{i+1} \mathbf{1}_{E_k^c} \Big) \mathbf{1}_{\{G \geq i\}} \Big| \CG_0  \Big]
\end{align*}
by (\ref{IneqCondLemmaSumGeom}). Using  (\ref{IneqHypoLemmaSumGeom}), that $\mathbf{1}_{E_{i+1}^c} \mathbf{1}_{E_i^c} \leq 1$ and the fact that $\{G \geq i\} \in \CG_{i-1}$, we deduce that
\begin{align*}
u_{i,j} &\leq  (1-c)^{j-i-2} \EE \Big[ \mathbf{1}_{\{G \geq i\}} \EE \Big[ r(\tau_{i+1}-\tau_i)   \Big| \CG_{i-1} \Big] \prod_{k = 1}^{i-1} \mathbf{1}_{E_k^c} \Big| \CG_0 \Big] \\
&\leq L (1-c)^{j-i-2}  \EE \Big[\prod_{k = 1}^{i-1} \mathbf{1}_{E_k^c} \Big| \CG_0 \Big]  \leq L (1-c)^{j-3},
\end{align*}
where we used (\ref{IneqCondLemmaSumGeom}). Using similar (easier) computations, one can show that
$$u_{0,1} \leq L, \quad \text{and for } j \geq 2, \quad u_{0,j} \leq L(1-c)^{j-2} \quad \text{ and } \quad u_{j-1,j} \leq L (1-c)^{j-2}. $$

\noindent We plug-in those results into (\ref{EqLemmaGeom}) to conclude, splitting the sum over the cases $j = 1$, $j = 2$ and $j \geq 3$, that there exists a constant $\kappa  > 0$ depending only on $r$ and $c$ such that
$$
\EE \Big[ r(\tau_G)  \Big| \CG_0 \Big] \leq C \Big( L + 2^{\beta +1} L + \sum_{j = 3}^{\infty} j^{\beta +1} L (1-c)^{j-3} \Big) \leq \kappa L, $$
as desired.
\end{proof}

Recall, for $(x,\theta) \in \pD \times \mathcal{A}$, the notation $\vartheta(x,\theta)$ introduced in Lemma \ref{NotatHM}. For any filtration $(\mathcal{F}_t)_{t \geq 0}$, any stopping time $\nu$  we introduce the $\sigma$-algebra $\mathcal{F}_{\nu-} := \sigma( A \cap \{t < \nu \}, t \in \RR_+, A \in \mathcal{F}_t)$, see \cite[Definition 1.11]{jacod1987limit}. We set $\mathcal{F}_{0-}$ to be the completion of the trivial $\sigma$-algebra. 
\begin{lemma}
\label{LemmaDensityTi}
Let $x \in \pD$ and $V = R \vartheta(x,\Theta)$, with $(R,\Theta) \sim \Upsilon$. Let $(X_t,V_t)_{t \geq 0}$ be a free-transport process (see Remark \ref{RmkExtensionDefiProcess}) with $(X_0, V_0) = (x,V) \in \partial_- G$. Set $T_0 = 0$, $T_{i+1} = \inf\{t > T_i, X_t \in \pD\}$ for all $i \geq 0$. Then, for all $i \geq 1$, $T_i$ admits a density with respect to the Lebesgue measure on $\RR_+$.
\end{lemma}

\begin{proof}

We set for all $t \geq 0$, $\mathcal{F}_t = \sigma((X_s,V_s)_{0 \leq s \leq t})$. Let $A \in \mathcal{B}(\RR_+)$ with $\lambda(A) = 0$, where $\lambda$ is the Lebesgue measure on $\RR_+$.  We have $T_1 = \frac{\|x - q(x, \vartheta(x,\Theta))\|}{R}$, so that
\begin{align*}
\PP(T_1 \in A) = \int_{\mathcal{A}} \PP \Big(\frac{\|x - q(x, \vartheta(x,\theta))\|}{R} \in A \Big) h_{\Theta}(\theta) d\theta.
\end{align*}
For $\theta \in \mathcal{A} = (-\frac{\pi}{2}, \frac{\pi}{2}) \times [0,\pi]^{n-2}$, we set $A_{x,\theta} = \{s \in \RR_+, \frac{\|x - q(x, \vartheta(x,\theta))\|}{s} \in A\}$, so that
$$ \PP(T_1 \in A) = \int_{\mathcal{A}} \PP (R \in A_{x,\theta})h_{\Theta}(\theta) d\theta. $$
Note that $\lambda(A_{x,\theta}) = 0$ for all $\theta \in \mathcal{A}$. Since $R$ has a density $h_R$ with respect to the Lebesgue measure on $\RR_+$, we conclude that $\PP(T_1 \in A) = 0$, so that $T_1$ admits a density with respect to the Lebesgue measure on $\RR_+$.

\vspace{.5cm}
Concerning $T_2$, we introduce the event 
$B = \{\text{Specular}$ $ \text{reflection at } X_{T_1}\}$. Note that $B$ is independent of $R$, see Definition \ref{DefiProcess}. We fix $A \in \mathcal{B}(\RR_+)$ with $\lambda(A) = 0$. 
\begin{enumerate}[i)]
\item On the event $B$, since $T_2 = T_2 - T_1 + T_1$, setting $Y = q(x,\vartheta(x,\Theta))$ and recalling (\ref{EqSpecularRef}),
$$T_2 = \frac{ \Big\|Y - q(Y, \eta_Y(\vartheta(Y,\Theta)) \Big\|}{R} + \frac{\|x - Y\|}{R}.$$
Proceeding as for $T_1$, we find, with the notation $y = q(x,\vartheta(x,\theta))$,
$$ \PP(\{T_2 \in A\} \cap B) = \int_{\mathcal{A}} \big(1- \alpha(y) \big) \PP \Big( \frac{\|x - y\| + \|y - q(y, \eta_y(\vartheta(y,\theta)))\|}{R} \in A \Big) h_{\Theta} d\theta = 0. $$ 

\item On the event $B^c$, we introduce the process $(\xtilde{t},\vtilde{t})_{t \geq 0}$ with, $\xtilde{t} = X_{T_1 + t}$, $\vtilde{t} = V_{T_1 + t}$. By the strong Markov property for the process $(X_s,V_s)_{s \geq 0}$, we have that, setting $$\tilde{T}_1 = \inf\{t > 0, \xtilde{t} \in \pD\}  = T_2 - T_1,$$ $\tilde{T}_1$ admits a density with respect to $\lambda$, conditionally on $\mathcal{F}_{T_1-}$ on $B^c$. Indeed, $X_{T_1} \in \pD$ and is $\mathcal{F}_{T_1-}$-measurable, $V_{T_1} = R_1 \vartheta(X_{T_1}, \Theta_1)$ on $B^c$, with $(R_1, \Theta_1) \sim \Upsilon$ independent of $\mathcal{F}_{T_1-}$, so that we can apply the previous study for $T_1$. We obtain, since $T_1$ is $\mathcal{F}_{T_1-}$ measurable.
$$ \PP(T_2 \in A \cap B^c) =  \PP(\{\tilde{T}_1 + T_1 \in A\} \cap B^c) = 0. $$
\end{enumerate}
Hence, $\PP(\{T_2 \in A\}) = 0$. The conclusion follows by induction. 
\end{proof}

\subsection{Construction of the coupling.} 
\label{SubsectionConstruction}
In this section, we define the coupling of the two processes that we will use to prove Theorem \ref{MainTheorem}, and show two of its properties.

We recall that $\mathcal{U}$ is the uniform distribution over $[0,1]$ and $\mathcal{Q}$ is the law on $[0,1] \times \RR_+ \times \mathcal{A}$ such that $\mathcal{Q} = \mathcal{U} \otimes \Upsilon$, where $\Upsilon$ is defined in Lemma \ref{NotatHM}. For $x \in \pD$, $\tilde{x} \in D$, $\tilde{v} \in \RR^n$ with $\|\tilde{v}\| \geq 1$, recall the law $\Lambda_{x,\tilde{x}, \tilde{v}}$ on $(\RR_+ \times \mathcal{A})^2$ defined in Proposition \ref{CouplingProp} .

Let $(x, v, \tilde{x}, \tilde{v})$ in $(\bar{D} \times \RR^n)^2 $ with $x \in \pD$ or $\tilde{x} \in \pD$.  We define the law $\Gamma_{x, v, \tilde{x}, \tilde{v}}$ on the space $([0,1] \times \RR_+ \times \mathcal{A})^2$ by:
\begin{align}
\label{EqLawGlobalCoupling}
\Gamma_{x,v,\tilde{x}, \tilde{v}}&(du, dr, d\theta, d\tilde{u}, d\tilde{r}, d\tilde{\theta}) = \mathbf{1}_{\{x = \tilde{x}\}} \mathcal{Q} (du, dr,d\theta) \delta_u (d\tilde{u}) \delta_{r}(d\tilde{r}) \delta_{\theta}(d\tilde{\theta})  \\ 
& \quad + \mathbf{1}_{\{x \in \pD\}} \mathbf{1}_{\{ \tilde{x} \in D\} \cap \{ \|v\|  \geq 1, \|\tilde{v}\| \geq 1\}} \mathcal{U}(du) \Lambda_{x,\tilde{x}, \tilde{v}} (  dr, d\theta, d\tilde{r},  d\tilde{\theta}) \delta_u(d\tilde{u}) \nonumber \\
& \quad + \mathbf{1}_{\{x \ne \tilde{x}\}} \mathbf{1}_{\{ \tilde{x} \in \pD\} \cup \{ \|\tilde{v}\| < 1\} \cup \{\|v\| < 1\}} (\mathcal{Q} \otimes \mathcal{Q}) (du, dr, d\theta, d\tilde{u}, d\tilde{r}, d\tilde{\theta}). \nonumber
\end{align}

\noindent We can now describe the global coupling procedure with the help of this law. In order to obtain a Markov process, we introduce an additional random process $(Z_s)_{s \geq 0}$ with values in the set $\{\emptyset\} \cup ([0,1] \times \RR_+ \times \mathcal{A})$.

\begin{defi}
\label{DefiCouplingProcess}
We define a coupling process $(X_s,V_s,\xtilde{s},\vtilde{s}, Z_s)_{s \geq 0}$ by the following steps:
\begin{labeling}{Step k+1:}
\item [Step 0:] Simulate $(X_0,V_0) \sim f_0$, $(\xtilde{0},\vtilde{0}) \sim \mu_{\infty}$, set $Z_0 = \emptyset$ and $S_0 = 0$.

$\dots$

\item [Step k+1:]  Set $S_{k+1} = S_k + \zeta(X_{S_k},V_{S_k}) \wedge \zeta(\xtilde{S_k},\vtilde{S_k})$.

\noindent Set, for all $t \in (S_k,S_{k+1})$, $X_t = X_{S_k} + (t-S_k) V_{S_k}$, $V_t = V_{S_k}$, 

\noindent $\hspace{4.13cm} \xtilde{t} = \xtilde{S_k} + (t-S_k) \vtilde{S_k}$, $\vtilde{t} = \vtilde{S_k}$,

\noindent $\hspace{4.13cm} Z_t = Z_{S_k}$.

\noindent Set $X_{S_{k+1}} = X_{S_{k+1}-}$, $\xtilde{S_{k+1}} = \xtilde{S_{k+1}-}$.

\noindent Simulate $(Q_{k+1}, \tilde{Q}_{k+1}) \sim \Gamma_{X_{S_{k+1}}, V_{S_{k+1}-}, \xtilde{S_{k+1}}, \vtilde{S_{k+1}-}}$. 

\noindent Set $V_{S_{k+1}} = V_{S_{k+1}-} \mathbf{1}_{\{X_{S_{k+1}} \not \in \pD\}} + w(X_{S_{k+1}}, V_{S_{k+1}-}, Q_{k+1}) \mathbf{1}_{\{X_{S_{k+1}} \in \pD\}}$.

\noindent Set $\tilde{Q}'_{k+1} = \tilde{Q}_{k+1} \mathbf{1}_{\{Z_{S_{k+1}-} = \emptyset\}} + Z_{S_{k+1}-} \mathbf{1}_{\{Z_{S_{k+1}-} \ne \emptyset\}}$.

\noindent Set $\vtilde{S_{k+1}} = \vtilde{S_{k+1}-} \mathbf{1}_{\{\xtilde{S_{k+1}} \not \in \pD\}} + w(\xtilde{S_{k+1}}, \vtilde{S_{k+1}-}, \tilde{Q}'_{k+1}) \mathbf{1}_{\{\xtilde{S_{k+1}} \in \pD\}}$. 

\noindent Set $Z_{S_{k+1}} = \emptyset \mathbf{1}_{\{\xtilde{S_{k+1}} \in \pD\}} + \tilde{Q}'_{k+1} \mathbf{1}_{\{\xtilde{S_{k+1}} \not \in \pD\}}. $

\end{labeling}
\end{defi}

Observe that the last line of Definition \ref{DefiCouplingProcess} rewrites as 
$$ Z_{S_{k+1}} = \emptyset \mathbf{1}_{\{\xtilde{S_{k+1}} \in \pD\}} + Z_{S_{k+1}-} \mathbf{1}_{\{\xtilde{S_{k+1}} \not \in \pD, Z_{S_{k+1}-} \ne \emptyset \}} + \tilde{Q}_{k+1} \mathbf{1}_{\{\xtilde{S_{k+1}} \not \in \pD, Z_{S_{k+1}-} = \emptyset \}}. $$

\begin{rmk}
One can readily see from Definition \ref{DefiCouplingProcess} that the process $(X_s,V_s,\xtilde{s},\vtilde{s},Z_s)_{s \geq 0}$ is a strong Markov process. 
\end{rmk}

Let us explain informally this definition. The sequence $(S_k)_{k \geq 1}$ is the sequence of collisions with the boundary of $(X_s,V_s)_{s \geq 0}$ and $(\xtilde{s},\vtilde{s})_{s \geq 0}$. The behavior of the coupling process is clear between $S_k$ and $S_{k+1}$ for all $k \geq 0$. 
For all $k \geq 1$, at time $S_k$, we set $(X,V_-) = (X_{S_k},V_{S_k-})$, $(\xtilde{}, \vtilde{-}) = (\xtilde{S_k},\vtilde{S_k-})$, $Z_- = Z_{S_k-}$ and we have $X \in \pD$ or $\xtilde{} \in \pD$. We explain in the following table how we choose the new velocities $(V,\tilde{V})$ and update the value of $Z$.

\vspace{4.5cm}

\hspace{-.5cm}

\begin{table}[h]
  \centering
  \caption{Update when $X \in \pD$ or $\xtilde{} \in \pD$. }
\begin{tabular}{c c c c c}

\toprule
$X $ & $\xtilde{} $ & $Z_-$ & $\|V_-\| \wedge \|\vtilde{-}\|$ & Update \\
\midrule
& & & &  Simulate $(R,\Theta, \tilde{R},\tilde{\Theta}) \sim \Lambda_{X,\tilde{X},\vtilde{-}}$, $U \sim \mathcal{U}$. \\
$\in \pD$ & $\not \in \pD$ & $ \emptyset$ & $\geq 1$ & Set $(Q,\tilde{Q}) = ((U,R,\Theta),(U,\tilde{R},\tilde{\Theta})). $ \\
& & & & Update $V$ using $Q$, set $\vtilde{} = \vtilde{-}$ and store $\tilde{Q}$ in $Z$: $Z = \tilde{Q}$. \\
\midrule
& & & & Simulate $(Q,\tilde{Q}) \sim \mathcal{Q} \otimes \mathcal{Q}$. \\
$\in \pD$ & $\not \in \pD$ & $ \emptyset$ & $< 1$ & Update $V$ using $Q$, set $\vtilde{} = \vtilde{-}$, store $\tilde{Q}$ in $Z$: $Z = \tilde{Q}$ \\
& & & & (this is quite artificial since $\tilde{Q}$ is independent of $Q$). \\
\midrule
& & & &  Simulate $Q \sim \mathcal{Q}$. \\
$\in \pD$ & $ \in \pD$ & $ \emptyset$ & all values & Update $V$ and $\tilde{V}$ using $Q$ (if $V_- = \vtilde{-}$ then $V = \tilde{V}$).\\
& $\xtilde{} = X$ & & & Set $Z = \emptyset$.  \\
\midrule
& & & &  Simulate $(Q,\tilde{Q}) \sim \mathcal{Q} \otimes \mathcal{Q}$. \\
$\in \pD$ & $ \in \pD,$ & $ \emptyset$ & all values & Update $V$ using $Q$ and $\tilde{V}$ using $\tilde{Q}$.\\
& $ \xtilde{} \ne X$& & & Set $Z = \emptyset$.  \\
\midrule 
& & & &   Simulate $(R,\Theta, \tilde{R},\tilde{\Theta}) \sim \Lambda_{X,\tilde{X},\vtilde{-}}$, $U \sim \mathcal{U}$. \\
$\in \pD$ & $ \not \in \pD$ & $ \ne \emptyset$ & $\geq 1$ & Set $(Q,\tilde{Q}) = ((U,R,\Theta),(U,\tilde{R},\tilde{\Theta})$, update $V$ using $Q$. \\
& & & & Set $\tilde{V} = \vtilde{-}$. Leave $Z$ unchanged: $Z = Z_-$ ($\tilde{Q}$ is useless). \\
\midrule 
& & & &   Simulate $(Q, \tilde{Q}) \sim \mathcal{Q} \otimes \mathcal{Q}.$ \\
$\in \pD$ & $ \not \in \pD$ & $ \ne \emptyset$ & $< 1$ & Update $V$ using $Q$, set $\tilde{V} = \vtilde{-}$. \\
& & & & Leave $Z$ unchanged: $Z = Z_-$ ($\tilde{Q}$ is useless). \\
\midrule
& & & &   Simulate $Q \sim \mathcal{Q}$. \\
$\in \pD$ & $\in \pD$ & $ \ne \emptyset$ & all values & Update $V$ using $Q$, update $\tilde{V}$ using $Z_-$. \\
& $\xtilde{} = X$ & & & Clear $Z$ by setting $Z = \emptyset$.  \\
\midrule 
& & & &   Simulate $(Q, \tilde{Q}) \sim \mathcal{Q} \otimes \mathcal{Q}.$ \\
$\in \pD$ & $\in \pD$ & $ \ne \emptyset$ & all values & Update $V$ using $Q$, update $\tilde{V}$ using $Z_-$. \\
& $\xtilde{} \ne X$ & & & Clear $Z$ by setting $Z = \emptyset$ ($\tilde{Q}$ is useless). \\
\midrule
& & & &   Simulate $(Q, \tilde{Q}) \sim \mathcal{Q} \otimes \mathcal{Q}.$ \\
$\not \in \pD$ & $\in \pD$ & $ \emptyset$ & all values & Update $\tilde{V}$ using $\tilde{Q}$, set $V = V_-$. \\
& & & & Set $Z = \emptyset$ ($Q$ is useless). \\
\midrule 
& & & &   Simulate $(Q, \tilde{Q}) \sim \mathcal{Q} \otimes \mathcal{Q}.$ \\
$\not \in \pD$ & $\in \pD$ & $ \ne \emptyset$ & all values & Update $\tilde{V}$ using $Z_-$, set $V = V_-$. \\
& & & & Clear $Z$ by setting $Z = \emptyset$ ($Q, \tilde{Q}$ are useless). \\
\midrule
\end{tabular}
\label{Tab:TableUpdate}
\end{table}

Observe that all those cases are treated in a rather concise way in Definition \ref{DefiCouplingProcess}. This leads to simpler notations and hopefully allows for a clearer proof.

\begin{lemma}
\label{LemmaMarginalCode}
Let $(X_s,V_s, \xtilde{s},\vtilde{s}, Z_s)_{s \geq 0}$ be a coupling process. Then $(X_s,V_s)_{s \geq 0}$ is a free-transport process with initial distribution $f_0$ (see Definition \ref{DefiProcess}). Moreover, $(\xtilde{s}, \vtilde{s})_{ s \geq 0}$ is a free-transport process with initial distribution $\mu_{\infty}$.  
\end{lemma}

\begin{proof}
We write, for all $s \geq 0$, $\mathcal{G}_s = \sigma((X_t,V_t,\xtilde{t},\vtilde{t}, Z_t)_{0 \leq t \leq s})$, $\mathcal{F}_s = \sigma((X_t,V_t)_{0 \leq t \leq s})$ and $\tilde{\mathcal{F}}_s = \sigma((\xtilde{t},\vtilde{t})_{0 \leq t \leq s})$.
Note first that for all $i \geq 1$, $X_{S_i} \in \pD$ or $\xtilde{S_i} \in \pD$. We have, a.s., recalling (\ref{EqLawGlobalCoupling}) and Proposition \ref{CouplingProp},
\begin{align}
\label{EqMarginalGamma}
 \int_{(\tilde{u}, \tilde{r}, \tilde{\theta}) \in [0,1] \times \RR_+ \times \mathcal{A}} \Gamma_{X_{S_i}, V_{S_i-}, \xtilde{S_i}, \vtilde{S_i-}}(du, dr, d\theta, d\tilde{u}, d\tilde{r}, d\tilde{\theta}) = \mathcal{Q} (du, dr, d\theta).
 \end{align}
Hence, with a similar argument for $\tilde{Q}_i$, 
\begin{align}
\label{EqConditionalLawQTilde}
\mathcal{L}(Q_i | \mathcal{G}_{S_i-}) = \mathcal{Q}, \qquad \mathcal{L}(\tilde{Q}_i | \CG_{S_i-}) = \mathcal{Q}. 
\end{align}

We focus first on the process $(\xtilde{t},\vtilde{t})_{t \geq 0}$. We introduce the subsequence $(\nu_k)_{k \geq 0}$ defined by $\nu_0 = 0$ and $\nu_{k+1} = \inf \{j > \nu_k, \xtilde{S_j} \in \pD\}$. Comparing Definitions \ref{DefiProcess} and \ref{DefiCouplingProcess}, one realizes that the only difficulty is to verify that for all $k \geq 1$, $\tilde{Q}_{\nu_k}'$ is $\mathcal{Q}$-distributed and independent of $\tilde{\mathcal{F}}_{S_{\nu_k}-} = \tilde{\mathcal{F}}_{S_{\nu_{k-1}}}$.

Note first that, for all $k \geq 1$, $\{Z_{S_{\nu_k}-} = \emptyset\} \in \mathcal{G}_{S_{\nu_{k-1}}}$. Indeed, we have $Z_{S_{\nu_{k-1}}} = \emptyset$ a.s. and thus
\begin{align}
\label{EqSetZEmpty}
\{Z_{S_{\nu_k}-} = \emptyset\} = \Big \{ \zeta(X_{S_{\nu_{k-1}}}, V_{S_{\nu_{k-1}}}) \geq \zeta(\xtilde{S_{\nu_{k-1}}}, \vtilde{S_{\nu_{k-1}}}) \Big \} \in \mathcal{G}_{S_{\nu_{k-1}}}.
\end{align} 
We claim that for all $k \geq 1$,
$$ \tilde{Q}'_{\nu_k} = \mathbf{1}_{\{Z_{S_{\nu_k}-} = \emptyset \}} \tilde{Q}_{\nu_k} + \mathbf{1}_{\{Z_{S_{\nu_k}-} \ne \emptyset \}} \tilde{Q}_{\nu_{k-1} + 1}. $$
Indeed, we clearly have $\tilde{Q}_{\nu_k}' = \tilde{Q}_{\nu_k}$ on $\{Z_{S_{\nu_k}-} = \emptyset\}$, and, by (\ref{EqSetZEmpty}) and since $Z_{S_{\nu_{k-1}}} = \emptyset$ a.s., 
 \begin{align*}
 \{Z_{S_{\nu_k}-} \ne \emptyset \} &= \big\{\zeta(X_{S_{\nu_{k-1}}}, V_{S_{\nu_{k-1}}}) < \zeta(\xtilde{S_{\nu_{k-1}}}, \vtilde{S_{\nu_{k-1}}}) \big\} \\
 & \subset \{X_{S_{\nu_{k-1} + 1}} \in \pD, \xtilde{S_{\nu_{k-1} + 1}} \not \in \pD, Z_{S_{\nu_{k-1}+1}-} = \emptyset \} \\
 &\subset \{ Z_{S_{\nu_{k-1}+1}} = \tilde{Q}_{\nu_{k-1}+1}, \quad \nu_{k} > \nu_{k-1} + 1 \} \\
 &\subset \{ Z_{S_{\nu_k}-} = \tilde{Q}_{\nu_{k-1}+1} \}.
 \end{align*} 
 This concludes the proof of the claim.
 
Using (\ref{EqConditionalLawQTilde}), for all $k \geq 1$, $\mathcal{L}(\tilde{Q}_{\nu_k} | \mathcal{G}_{S_{\nu_k}-}) = \mathcal{Q}$ and $\mathcal{L}(\tilde{Q}_{\nu_{k-1}+1} | \mathcal{G}_{S_{\nu_{k-1}+1}-}) = \mathcal{Q}$. 
Consider a function $\phi \in C^{\infty}_c([0,1] \times \RR_+ \times \mathcal{A})$. For $k \geq 1$, we compute
\begin{align*}
\EE\big[\phi(\tilde{Q}'_{\nu_k}) \big| \tilde{\mathcal{F}}_{S_{\nu_{k-1}}} \big] &= \EE \big[\phi(\tilde{Q}_{\nu_k}) \mathbf{1}_{\{Z_{S_{\nu_k}-} = \emptyset \}} \big| \tilde{\mathcal{F}}_{S_{\nu_{k-1}}} \big] + \EE\big[\phi(\tilde{Q}_{\nu_{k-1} + 1}) \mathbf{1}_{\{Z_{S_{\nu_k}-} \ne \emptyset \}} \big| \tilde{\mathcal{F}}_{S_{\nu_{k-1}}} \big] \\
& = \EE \Big[\mathbf{1}_{\{Z_{S_{\nu_k}-} = \emptyset \}} \EE \big[\phi(\tilde{Q}_{\nu_k}) \big|\mathcal{G}_{S_{\nu_{k}}-} \big] \Big| \tilde{\mathcal{F}}_{S_{\nu_{k-1}}} \Big] \\ 
& \quad +\EE \Big[\mathbf{1}_{\{Z_{S_{\nu_k}-} \ne \emptyset \}} \EE \big[\phi(\tilde{Q}_{\nu_{k-1}+1}) \big| \mathcal{G}_{S_{\nu_{k-1}+1}-} \big]  \Big| \tilde{\mathcal{F}}_{S_{\nu_{k-1}}} \Big],
\end{align*}
using (\ref{EqSetZEmpty}) and the fact that $\tilde{\mathcal{F}}_{S_{\nu_{k-1}}} \subset \mathcal{G}_{S_{\nu_{k-1}}} \subset \mathcal{G}_{S_{\nu_{k-1}+1}-} \subset \mathcal{G}_{S_{\nu_k}-}$. From the previous remarks on the conditional law of $\tilde{Q}_{\nu_k}$, $\tilde{Q}_{\nu_{k-1}+1}$, we obtain
$$ \EE \big[\phi(\tilde{Q}'_{\nu_k}) \big| \tilde{\mathcal{F}}_{S_{\nu_{k-1}}} \big] = \int_{[0,1] \times \RR_+ \times \mathcal{A}} \phi(x) \mathcal{Q}(dx) \Big( \EE \big[\mathbf{1}_{\{Z_{S_{\nu_k}-} = \emptyset \}} \big| \tilde{\mathcal{F}}_{S_{\nu_{k-1}}} \big] + \EE \big[\mathbf{1}_{\{Z_{S_{\nu_k}-} \ne \emptyset \}} \big| \tilde{\mathcal{F}}_{S_{\nu_{k-1}}} \big] \Big), $$
from which we conclude that $\mathcal{L}(\tilde{Q}_{\nu_k}'|\tilde{\mathcal{F}}_{S_{\nu_k}-}) = \mathcal{Q}$, as desired.

\vspace{.5cm}

The argument for $(X_s,V_s)_{s \geq 0}$ is similar and much easier since for all $j \geq 1$ such that $X_{S_j} \in \pD$, $V_{S_j} = w(X_{S_j},V_{S_j-},Q_j)$ with $\mathcal{L}(Q_j|\mathcal{F}_{S_j-}) = \mathcal{Q}$ using (\ref{EqConditionalLawQTilde}) and that $\mathcal{F}_{S_j-} \subset \CG_{S_j-}$. 
\end{proof}

\begin{lemma}
\label{Lemmataudefinitive}
Let $(X_s,V_s, \xtilde{s}, \vtilde{s}, Z_s)_{s \geq 0}$ be a coupling process. Then for all $t \geq 0$,
$$ \{(X_t,V_t) = (\xtilde{t},\vtilde{t}), Z_t = \emptyset\} \subset \{(X_{t+s}, V_{t+s})_{s \geq 0} = (\xtilde{t+s},\vtilde{t+s})_{s \geq 0} \}. $$
\end{lemma}

\begin{proof}
According to Definition \ref{DefiCouplingProcess}, on the event $\{(X_t,V_t) = (\xtilde{t},\vtilde{t}), Z_t = \emptyset \}$,  there exists $k \geq 1$ such that $S_k = t + \zeta(X_t,V_t) = t + \zeta(\xtilde{t},\vtilde{t})$ and we have
$$ \Big\{(X_t,V_t) = (\xtilde{t},\vtilde{t}), Z_t = \emptyset \Big\} \subset \Big\{(X_{S_k-},V_{S_k-}) = (\xtilde{S_k-},\vtilde{S_k-}), Z_{S_k-} = \emptyset \Big\}, $$
and $(X_s,V_s)_{t \leq s < S_k} = (\xtilde{s},\vtilde{s})_{t \leq s < S_k}$. 
We then have, according to the definition,  the equality

\noindent $X_{S_k} = X_{S_k-} = \xtilde{S_k-} = \xtilde{S_k}$ and $Z_{S_k-} = \emptyset$. Also, by definition of $\Gamma_{X_{S_k-},V_{S_k-},\xtilde{S_k-},\vtilde{S_k-}}$, since $X_{S_k-} = \xtilde{S_k-}$, we have $Q_k = \tilde{Q}_k$ with the notations of the definition. From there we obtain
$$ V_{S_k} = w(X_{S_k},V_{S_k-},Q_k) = w(\xtilde{S_k},\vtilde{S_k-}, \tilde{Q}_k)= \vtilde{S_k}, \quad \text{ and } \quad Z_{S_k} = \emptyset. $$
Hence $(X_s,V_s) = (\xtilde{s},\vtilde{s})$ and $Z_s = \emptyset$ for all $s \in (S_k, S_{k+1}]$. We conclude by iterating this procedure.
\end{proof}

\subsection{Proof of Theorem \ref{MainTheorem} in the convex case. }

We recall that the set $D$ is a bounded $C^2$ domain, uniformly convex in this section. The function $r$ defined on $\RR_+$ is such that there exists $C > 0$ satisfying, for all $(x,y) \in (\RR_+)^2,$ $r(x+y) \leq C (r(x) + r(y))$. The function $M : \RR^n \to (0,\infty)$ is radially symmetric and of mass 1 with $\int_{\RR^n} \|v\| M(v) dv <\infty$.
The function $\alpha$ defined on $\pD$ is uniformly bounded from below by $\alpha_0 > 0$. Finally, $\mu_{\infty}(dx,dv) = \frac{M(v)}{|D|} dx dv$ is the equilibrium distribution. Recall that $h_R$ is defined by $h_R(s) = c_R s^n M(s)$ for all $s \in \RR_+$ with $c_R$ a normalization constant, see Lemma \ref{NotatHM}. 
We define the constant $C_0 > 0$ by 
\begin{align}
\label{EqDefC0}
 C_0 = \max \Big( \int_{D \times \RR^n} r\Big(\frac{d(D)}{\|v\|}\Big) f_0(dx,dv), \int_{D \times \RR^n} r\Big(\frac{d(D)}{\|v\|}\Big) \mu_{\infty}(dx,dv), \int_{\RR_+} r\Big( \frac{d(D)}{s} \Big) h_R(s) ds \Big),
 \end{align} 
which is finite using (\ref{EqHypoMomentThm}) and since
\begin{align*}
 \int_{\RR_+} r\Big(\frac{d(D)}{s} \Big) h_R(s) ds &= \kappa \int_{\RR^n} r \Big( \frac{d(D)}{\|v\|} \Big) \|v\| M(v) dv \\
 &\leq \kappa \int_{\{\|v\| \leq 1\}} r\Big(\frac{d(D)}{\|v\|} \Big) M(v) dv + \kappa r(d(D)) \int_{\{\|v\| > 1\}} \|v\| M(v) dv. 
 \end{align*}
In this whole subsection $\kappa$ and $L$ denote some positive constants depending on $r$, $D$ and $\alpha_0$,  whose value is allowed to vary from line to line. Recall Remark \ref{RmkExtensionDefiProcess} for the definition of a free-transport process with initial distribution $\delta_x \otimes \delta_v$ with $(x,v) \in \partial_+ G$.

\begin{lemma}
\label{LemmaControlCalT}
There exists $\kappa > 0$ such that if $(x,v), (\tilde{x},\tilde{v}) \in (D \times \RR^n) \cup \partial_+ G$ and $(X_t,V_t)_{t \geq 0}$, $(\xtilde{t},\vtilde{t})_{t \geq 0}$ are two possibly correlated free-transport processes with initial distributions $\delta_x \otimes \delta_v$ and $\delta_{\tilde{x}} \otimes \delta_{ \tilde{v}}$ respectively, setting
$$ \mathcal{T} = \inf \{t > 0, \|V_{t}\| \ne \|v\|, \|\vtilde{t}\| \ne \|\tilde{v}\|\}, $$
we have $$ \EE[ r(\mathcal{T}) ] \leq \kappa \Big(1 + r\Big(\frac{d(D)}{\|v\|}\Big) + r\Big(\frac{d(D)}{\|\tilde{v}\|} \Big) \Big). $$
\end{lemma}

\begin{proof}
We introduce the sequence $(T_k)_{k \geq 0}$ by setting first $T_0 = \zeta(x, v)$ so that $T_0 = 0$ in the case where $(x,v) \in \partial_+ G$, and for $k \geq 0$, $T_{k+1} = \inf \{t > T_k, X_t \in \pD\}$. We introduce the filtration $\mathcal{F}_t = \sigma((X_s,V_s)_{0 \leq s \leq t} )$.  We also set $S_1 = \inf \{t > 0, \|V_{t}\| \ne \|v\| \}$  and $\tilde{S}_1 = \inf \{t > 0, \|\vtilde{t}\| \ne \|\tilde{v}\|\}$. Note that $\mathcal{T} = S_1 \vee \tilde{S}_1$. 

\vspace{.3cm}

\textbf{Step 1.}  We prove that
$$\EE[r(S_1)] \leq \kappa \Big(r \Big(\frac{d(D)}{\|v\|} \Big) + 1\Big).$$ We write $(U_i, R_i, \Theta_i)_{i \geq 0}$ for the sequence of $\mathcal{Q}$-distributed vectors such that for all $i \geq 0$, 
$$ V_{T_i} = w(X_{T_i}, V_{T_i-}, U_i, R_i, \Theta_i), $$
with $V_{0-} = v$. Set $A_n = \{\|V_{T_n}\| \ne \|V_{T_{n}-}\|\}$ for all $n \geq 0$, and $N = \inf\{n \geq 1, A_n \text{ is realized}\}$ so that $S_1 \leq T_N$ ($S_1$ may differ from $T_N$ if $x \in \pD$). We first use Lemma \ref{LemmaSumGeometric} to prove that
\begin{align}
\label{EqProvisoryS_1}
 \EE[r(T_N - T_1)|\mathcal{F}_{T_1-}] \leq \kappa \Big(1 + r \Big(\frac{d(D)}{\|V_{T_0}\|} \Big) \Big).
 \end{align}
\begin{enumerate}
\item We set for all $k \geq 0$, $\mathcal{G}_k = \mathcal{F}_{T_{k+1}-}$,  and for $k \geq 1$, $\tau_k = T_{k+1} - T_1 $ which is $\mathcal{G}_k$-measurable, $E_k = A_{k} \in\mathcal{G}_k$, so that $G = N$, with $G = \inf\{ k \geq 1, E_k \text{ is realized} \}$ corresponding to the notation of Lemma \ref{LemmaSumGeometric}.
\item For all $k \geq 1$, we have $\PP(E_k|\mathcal{G}_{k-1}) = \PP(A_k | \mathcal{F}_{T_k-}) = \PP (U_k \leq \alpha(X_{T_k})) \geq \alpha_0$, whence (\ref{HypoProbaEk}).
\item We have, by definition of $C_0$,
\begin{align*}
 \EE[r(\tau_1)|\CG_0] = \EE[r(T_2 - T_1)|\mathcal{F}_{T_1-}]
\leq \EE \Big[ r \Big(\frac{d(D)}{\|V_{T_1}\|}\Big) \Big | \mathcal{F}_{T_1-} \Big]  \leq  C_0 + r\Big(\frac{d(D)}{\|V_{T_0}\|} \Big),
\end{align*}  
since $\|V_{T_1}\| = \|V_{T_1}\| \mathbf{1}_{A_1} + \|V_{T_0}\|\mathbf{1}_{A_1^c}$ with $\mathcal{L}(\|V_{T_1}\| | A_1) = h_R$. 

For all $k \geq 1$, since $\|V_{T_{k}-}\| = \|V_{T_0}\|$ on $\{N \geq k\}$, we obtain,
\begin{align*}
 &\mathbf{1}_{\{G \geq k\}} \EE[r(\tau_{k+1}-\tau_k)|\CG_{k-1}] = \mathbf{1}_{\{N \geq k\}} \EE \Big[r(T_{k+2} - T_{k+1}) \Big|\mathcal{F}_{T_{k}-} \Big] \\ 
 &\quad \leq \EE \Big[r \Big( \frac{d(D)}{\|V_{T_{k+1}}\|} \Big) \Big( \mathbf{1}_{A_{k+1}^c \cap A_{k}^c } + \mathbf{1}_{A_{k+1} \cap A_{k}^c }  + \mathbf{1}_{A_{k+1}^c \cap A_{k} } + \mathbf{1}_{A_{k+1} \cap A_{k} } \Big) \mathbf{1}_{\{\|V_{T_{k}-}\| = \|V_{T_0}\|\}}    \Big|\mathcal{F}_{T_{k}-} \Big] \\
&\quad \leq r \Big(\frac{d(D)}{\|V_{T_0}\|} \Big) + 3 C_0,
\end{align*}
because $\|V_{T_{k+1}}\| = \|V_{T_0}\|$ on $A_k^c \cap A_{k+1}^c$ and the last three terms are bounded by $C_0$
since we clearly have $\mathcal{L}(\|V_{T_n}\| | A_k ) = h_R$ for all $n \geq k \geq 0$. We have proved (\ref{IneqHypoLemmaSumGeom}). 
\end{enumerate}
 
Applying Lemma \ref{LemmaSumGeometric} we conclude that there exists $\kappa > 0$ such that (\ref{EqProvisoryS_1}) holds. To conclude this step, note that
\begin{align*}
\EE[r(S_1)] &\leq C \big(\EE[ \EE[r(T_N - T_1)| \mathcal{F}_{T_1-}] ] + \EE[r(T_1 - T_0)] + \EE[r(T_0)] \big) \\
&\leq C \Big( \EE \Big[ \kappa \Big( 1 + r \Big(\frac{d(D)}{\|V_{T_0}\|}\Big) \Big) \Big] + \EE \Big[r \Big(\frac{d(D)}{\|V_{T_0}\|} \Big) \Big] + r\Big(\frac{d(D)}{\|v\|} \Big) \Big) \\
&\leq \kappa \Big( 1 + 2C_0 + 3r \Big(\frac{d(D)}{\|v\|} \Big) \Big),
\end{align*}
since $\|V_{T_0}\| = \|V_{T_0}\| \mathbf{1}_{A_0} + \|v\|\mathbf{1}_{A_0^c}$ with $\mathcal{L}(\|V_{T_0}\| | A_0) = h_R$.

\vspace{.3cm}

\textbf{Step 2. } We apply the previous step with the process $(\xtilde{s},\vtilde{s})_{s \geq 0}$ and conclude that
$$ \EE[r(\tilde{S}_1)] \leq \kappa \Big(1 + r\Big(\frac{d(D)}{\|\tilde{v}\|} \Big) \Big). $$
 
\vspace{.3cm}

\textbf{Step 3. } Since $\mathcal{T} = S_1 \vee \tilde{S}_1$, we conclude that 
$$\EE[r(\mathcal{T})] \leq \EE[r(S_1)] + \EE[r(\tilde{S}_1)] \leq \kappa \Big(1 + r \Big(\frac{d(D)}{\|v\|} \Big) + r\Big(\frac{d(D)}{\|\tilde{v}\|} \Big) \Big). $$

\vspace{-.5cm}
\end{proof}

\begin{lemma}
\label{LemmaIndepProcess}
There exists $\kappa > 0$ such that if $(x,v), (\tilde{x},\tilde{v}) \in (D \times \RR^n) \cup \partial_+ G$ and $(X_t,V_t)_{t \geq 0}$, $(\xtilde{t},\vtilde{t})_{t \geq 0}$ are two independent free-transport processes with initial distributions $\delta_x \otimes \delta_v$ and $\delta_{\tilde{x}} \otimes \delta_{ \tilde{v}}$ respectively, setting
$$ S = \inf \{t > 0, \xtilde{t} \in \pD, X_t \in D, \|V_{t-}\| \wedge \|\vtilde{t-}\| \geq 1, \|V_{t-}\| \ne \|v\|, \|\vtilde{t-}\| \ne \|\tilde{v}\|\}, $$
we have $$ \EE[ r(S) ] \leq \kappa \Big(1 + r\Big(\frac{d(D)}{\|v\|}\Big) + r\Big(\frac{d(D)}{\|\tilde{v}\|} \Big) \Big). $$
\end{lemma}

\begin{proof}
We introduce the filtration $\mathcal{F}_t = \sigma((X_s,V_s,\xtilde{s},\vtilde{s})_{0 \leq s \leq t} )$.  We also introduce the stopping times $\mathcal{T} = \inf \{t > 0, \|V_t\| \ne \|v\|, \|\vtilde{t}\| \ne \|\tilde{v}\| \}$ and $\tilde{S}_1 = \inf \{t > 0, \xtilde{t} \in \pD, \|\vtilde{t-}\| \ne \|\tilde{v}\|, \|V_{t-}\| \ne \|v\| \}$.

\vspace{.5cm}
\textbf{Step 1.} We prove that
$$\EE[r(\tilde{S}_1)] \leq \kappa \Big(1 + r \Big(\frac{d(D)}{\|v\|} \Big) + r \Big(\frac{d(D)}{\|\tilde{v}\|} \Big) \Big).$$ 

Note first that $\tilde{S}_1 \leq \mathcal{T} + \zeta(\xtilde{\mathcal{T}}, \vtilde{\mathcal{T}})$ since for all $t \geq \mathcal{T}$, almost surely, $\|V_t\| \ne \|v\|$, $\|\vtilde{t}\| \ne \|\tilde{v}\|$ and because $\xtilde{\mathcal{T} + \sigma(\xtilde{\mathcal{T}}, \vtilde{\mathcal{T}})} \in \pD$. 

Applying Lemma \ref{LemmaControlCalT}, we find that 
$$\EE[r(\mathcal{T})] \leq \kappa \Big(1 + r\Big(\frac{d(D)}{\|v\|} \Big) + r\Big(\frac{d(D)}{\|\tilde{v}\|} \Big) \Big). $$

Hence, noting that $\mathcal{L}(\|\vtilde{\CT}\|) = h_R$, we obtain
\begin{align*}
\EE[r(\tilde{S}_1)] \leq C \Big(\EE[r(\mathcal{T})] + \EE[r(\zeta(\xtilde{\mathcal{T}}, \vtilde{\mathcal{T}}))] \Big) &\leq \kappa \Big(1 + r\Big(\frac{d(D)}{\|v\|} \Big) + r\Big(\frac{d(D)}{\|\tilde{v}\|} \Big) + \EE \Big[ r\Big(\frac{d(D)}{\|\vtilde{\mathcal{T}}\|} \Big) \Big]  \Big) \\
&\leq \kappa \Big(1 + r\Big(\frac{d(D)}{\|v\|} \Big) + r\Big(\frac{d(D)}{\|\tilde{v}\|}  \Big) \Big),
\end{align*}
where we used that $\EE[r(\frac{d(D)}{\|\vtilde{\CT}\|})] \leq C_0$, see (\ref{EqDefC0}).

\vspace{.5cm}

\textbf{Step 2.} 
We set $\tilde{S}_0 = 0$, define $ \tilde{S}_1$ as in Step 1, and set, for $n \geq 1$, $$\tilde{S}_{n+1} = \inf \{ t > \tilde{S}_{n}, \xtilde{t} \in \pD, \|\vtilde{t-}\| \ne \|\vtilde{\tilde{S}_{n}-}\|, \|V_{t-}\| \ne \|V_{\tilde{S}_n-}\| \}.$$
We set, for all $n \geq 1$, $B_n = \{ \|V_{\tilde{S}_{n}-}\| \wedge \| \vtilde{\tilde{S}_{n}-}\| \geq 1 \}$ and $G := \inf \{n \geq 1: B_n \text{ is realized} \}.$ The aim of this step is to check that 
$$\EE[r(\tilde{S}_{G})] \leq \kappa \Big(1 + r\Big(\frac{d(D)}{\|v\|}\Big) + r\Big(\frac{d(D)}{\|\tilde{v}\|} \Big) \Big). $$
We plan to apply Lemma \ref{LemmaSumGeometric}.

\begin{enumerate}
\item We set, for all $k \geq 0$, $\mathcal{G}_k = \mathcal{F}_{\tilde{S}_{k}-},$ and for all $k \geq 1$, $\tau_k = \tilde{S}_{k}$ which is $\mathcal{G}_k$-measurable, $E_k = B_{k} \in \mathcal{G}_k$ so that $G$ corresponds to the notation in Lemma \ref{LemmaSumGeometric}. 
\item For all $k \geq 1$, using that $\mathcal{L}(\|V_{\tilde{S}_{k}-}\| | \mathcal{F}_{\tilde{S}_{k-1}-}) = \mathcal{L}(\|\vtilde{\tilde{S}_{k}-}\| | \mathcal{F}_{\tilde{S}_{k-1}-}) = h_R$ since both processes have a diffuse reflection between $\tilde{S}_{k-1}-$ and $\tilde{S}_{k}-$,
$$ \PP(E_k | \mathcal{G}_{k-1}) = \PP(B_{k} | \mathcal{F}_{\tilde{S}_{k-1}-}) = \Big(\int_1^{\infty} h_R(r) dr\Big)^2 =: c, $$
and $c > 0$ by hypothesis, see (\ref{EqHypoM1}), whence (\ref{HypoProbaEk}).
\item Using the strong Markov property and Step 1, we have, for all $k \geq 0$,
\begin{align}
\label{EqIterSn}
\EE[r(\tilde{S}_{k+1}-\tilde{S}_{k})|\mathcal{F}_{\tilde{S}_{k}-}] \leq \kappa \Big(1 + r \Big(\frac{d(D)}{\|V_{\tilde{S}_k-}\|} \Big) + r \Big(\frac{d(D)}{\|\vtilde{\tilde{S}_{k}-}\|} \Big) \Big) =: K_k.
\end{align} 
Moreover, $K_0 = \kappa (1 + r(\frac{d(D)}{\|v\|}) + r(\frac{d(D)}{\|\tilde{v}\|}))$ and for $k\geq 1$,
\begin{align*}
\EE[r(\tau_{k+1}-\tau_k)| \mathcal{G}_{k-1}] = \EE \Big[ K_{k} \Big| \mathcal{F}_{\tilde{S}_{k-1}-}\Big]
&\leq \kappa \EE \Big[1 + r\Big( \frac{d(D)}{\|V_{\tilde{S}_{k}-}\|} \Big) + r\Big( \frac{d(D)}{\|\vtilde{\tilde{S}_{k}-}\|} \Big) \Big| \mathcal{F}_{\tilde{S}_{k-1}-} \Big] \\
&\leq \kappa (1 + 2C_0).
\end{align*} 
We used again that $\mathcal{L} (\|V_{\tilde{S}_{k}-}\| | \mathcal{F}_{\tilde{S}_{k-1}-} ) = \mathcal{L} (\|\vtilde{\tilde{S}_{k}-}\| | \mathcal{F}_{\tilde{S}_{k-1}-} )  = h_R$. Finally, we have
$$ \EE[r(\tau_1)|\mathcal{G}_0] = \EE[r(\tilde{S}_1)] \leq \kappa \Big(1 + r \Big(\frac{d(D)}{\|v\|} \Big) + r \Big(\frac{d(D)}{\|\tilde{v}\|} \Big) \Big). $$
\end{enumerate}
We conclude by applying Lemma \ref{LemmaSumGeometric}. 

\vspace{.5cm}
\noindent\textbf{Step 3.} 
We prove that, for all $i \geq 1$, $X_{\tilde{S}_i} \not \in \pD$ almost surely. Since $\xtilde{\tilde{S}_G} \in \pD$ and $\|V_{\tilde{S}_G-}\| \wedge \|\vtilde{\tilde{S}_G-}\| \geq 1$ by definition, by Step 2, this will conclude the proof. Set $S_1 = \inf\{t > 0, X_t \in \pD, \|V_t\| \ne \|v\|\}$ and note that $S_1 \leq \tilde{S}_1$ by definition.

\noindent We set $(X'_t,V'_t) = (X_{S_1 + t}, V_{S_1 + t})$, $(\tilde{X}'_t, \tilde{V}'_t) = (\xtilde{S_1 + t}, \vtilde{S_1 + t})$ for all $t \geq 0$.  Set $T'_0 = 0$, and for all $i \geq 1$,  $T'_{i+1} = \inf \{t > T'_i, X'_t \in \pD\}$. Set also $\tilde{T}'_0 = 0$ and for all $i \geq 0$,  $\tilde{T}'_{i+1} = \inf \{t > \tilde{T}'_i, \tilde{X}'_t \in \pD \}$. Since $(X'_t, V'_t)_{t \geq 0}$ and $(\xtilde{t}', \vtilde{t}')_{t \geq 0}$ are, conditionally on $\mathcal{F}_{S_1-}$, two independent processes, $T'_i$ is independent of $\tilde{T}'_j$ for all $(i,j) \in (\mathbb{N}^*)^2$ conditionally on this $\sigma$-algebra. Moreover, by Lemma \ref{LemmaDensityTi}, $T'_i$ has a density conditionally on $\mathcal{F}_{S_1-}$, since $X_{S_1} \in \pD$ and $V_{S_1} = R \vartheta(X_{S_1},\Theta)$ with $(R,\Theta) \sim \Upsilon$ independent of $\mathcal{F}_{S_1-}$. We thus have, for $(i,j) \in (\mathbb{N}^*)^2,$ 
$$ \PP( T'_i = \tilde{T}'_j | \mathcal{F}_{S_1-}) = 0. $$
Since we have $\{ X_{\tilde{S}_G} \in \pD\} \subset \cup_{i, j \geq 1} \{T'_i = \tilde{T}'_j\}$, we obtain $X_{\tilde{S}_G} \not \in \pD$ a.s. as desired.
\end{proof}

Let us introduce some notations for the remaining part of this section.

\begin{notat}
\label{NotatConvexProof} Let $(X_s,V_s, \xtilde{s},\vtilde{s}, Z_s)_{s \geq 0}$ a coupling process, see Definition \ref{DefiCouplingProcess}. We use the same sequences $(S_i,Q_i,\tilde{Q}_i)_{i \geq 1} $ as in the definition, as well as $(\tilde{Q}'_i)_{i \geq 1}$, and we recall that, for all $i \geq 1$,
% we write $(Q_i, \tilde{Q}'_i)$ for the couple of $\mathcal{Q}$-distributed random variables such that
$$ V_{S_i} = w(X_{S_i},V_{S_i-},Q_i) \mathbf{1}_{\{X_{S_i} \in \pD\}} +V_{S_i-}  \mathbf{1}_{\{X_{S_i} \not \in \pD\}}, $$
$$ \vtilde{S_i} = w(\xtilde{S_i},\vtilde{S_i-},\tilde{Q}'_i) \mathbf{1}_{\{\xtilde{S_i} \in \pD\}} + \vtilde{S_i-} \mathbf{1}_{\{\xtilde{S_i} \not \in \pD\}} .$$

\noindent a) We set $T_0 = 0$, $\tilde{T}_0 = 0$ and for $k \geq 0$,
$$T_{k+1} = \inf\{t > \tilde{T}_{k}, X_t \in \pD\}, \qquad \tilde{T}_{k+1} = \inf\{t > T_{k+1}, \xtilde{t} \in \pD\}. $$
For all $k \geq 1$, we have $Z_{T_k-} = \emptyset$ and $X_{T_k} \in \pD$ so $Z_{T_k} \ne \emptyset$ if $\xtilde{T_k} \not \in \pD$. We always have $Z_{\tilde{T}_k} = \emptyset$. 
For all $k \geq 1$, we write $(\underline{Q}_k, \underline{\tilde{Q}}_k) = (\underline{U}_k,\underline{R}_k,\underline{\Theta}_k, \underline{\tilde{U}}_k, \underline{\tilde{R}}_k, \underline{\tilde{\Theta}}_k)$ for the random vector such that
$$V_{T_k} = w(X_{T_k}, V_{T_k-} ,\underline{Q}_k), \quad \text{ and } \quad \vtilde{\tilde{T}_k} = w(\xtilde{\tilde{T}_k}, \vtilde{\tilde{T}_k-}, \underline{\tilde{Q}}_k).$$
Note that $(\underline{Q}_k, \underline{\tilde{Q}}_k)_{k \geq 1}$ is a subsequence of $(Q_i, \tilde{Q}_i')_{i \geq 1}$.

\noindent b) For all $t \geq 0$, we set $$\mathcal{F}_t = \sigma \Big( (X_s,V_s,\xtilde{s},\vtilde{s},Z_s)_{0 \leq s \leq t}, (Q_i \mathbf{1}_{\{S_i \leq t\}})_{i \geq 1}, (\tilde{Q}_i \mathbf{1}_{\{S_i \leq t\}})_{i \geq 1} \Big).$$

\noindent c) We set 
$ \sigma_1 = \inf\{ t > 0, X_t = \xtilde{t} \in \pD, Z_{t-} = \emptyset, \|V_{t-} \| \ne \|V_0\|, \|\vtilde{t-}\| \ne \|\vtilde{0}\|\}. $

\noindent d) We set $\nu_0 = 0$ and for all $k \geq 0$,
$$ \nu_{k+1} = \inf \{n \geq \nu_k + 1, \xtilde{T_n} \not \in \pD, \|V_{T_n-}\| \wedge \|\vtilde{T_n-}\| \geq 1\}. $$ 
\end{notat} 
Note that, according to Definition \ref{DefiCouplingProcess}, we have for all $n \geq 1$, conditionally on $\mathcal{F}_{T_{\nu_n}-}$, 
$$(\underline{R}_{\nu_n}, \underline{\Theta}_{\nu_n}, \underline{\tilde{R}}_{\nu_n}, \underline{\tilde{\Theta}}_{\nu_n}) \sim \Lambda_{X_{T_{\nu_n}-}, \xtilde{T_{\nu_n}-}, \vtilde{T_{\nu_n}-}},$$
where we recall that $\Lambda$ is defined in Proposition \ref{CouplingProp}. We also have $Z_{T_{\nu_n}} \ne \emptyset$, see (\textit{a}).

\begin{lemma}
\label{LemmaTnu}
There exist three constants $\kappa, L, c > 0$ such that the following holds.
\begin{enumerate}[i)]
\item For all $m \geq 1$, $$\mathbf{1}_{\{T_{\nu_m} < \sigma_1\}} \EE[r(T_{\nu_{m + 1}} \wedge \sigma_1 - T_{\nu_m})| \mathcal{F}_{T_{\nu_m}-}] \leq L. $$
\item  $\EE[r(T_{\nu_1} \wedge \sigma_1)] \leq  \kappa \Big(1 + \EE \Big[r\Big(\frac{d(D)}{\|V_{0}\|}\Big) + r\Big(\frac{d(D)}{\|\vtilde{0}\|}\Big) \Big] \Big). $ 

\item For all $m \geq 1$, setting 
$$A_m = \{\underline{U}_{\nu_m} \leq \alpha_0,  X_{T_{\nu_{m} +1}}= \xtilde{T_{\nu_{m} + 1}}, T_{\nu_{m}+1} = T_{\nu_m} + \zeta(X_{T_{\nu_m}}, V_{T_{\nu_m}})\}, $$
we have $$\PP \Big(A_m \Big| \mathcal{F}_{T_{\nu_m}-} \Big) \geq c, $$
and  $A_m \subset \{\sigma_1 \leq T_{\nu_m + 1}\}$ outside a $\PP$-null set.
\end{enumerate}
\end{lemma}

\begin{proof}

We prove i). Recall Remark \ref{RmkExtensionDefiProcess} which defines a free-transport process with initial distribution $\delta_x \otimes \delta_v$, with $(x,v) \in \partial_+ G$. For all $k \geq 1$, we have  $\|V_{T_{\nu_k}-}\| \wedge \|\vtilde{T_{\nu_k}-}\| \geq 1$, $Z_{T_{\nu_k-}} = \emptyset$ and $X_{T_{\nu_k}} \in \pD$, $\xtilde{T_{\nu_k}} \not \in \pD$. Thus, using the strong Markov property, we only need to prove that there exists some $L > 0$ such that for all $(x,v) \in \partial_+ G$,  $\tilde{x} \not \in \pD$, $\tilde{v} \in \RR^n$ with $\|v\| \wedge \|\tilde{v}\| \geq 1$,  if $(X_0,\xtilde{0},V_{0-}, \tilde{V}_{0-}, Z_0-) = (x,\tilde{x},v,\tilde{v}, \emptyset)$,
\begin{align}
\label{PreliminaryEquationTnu1}
 \EE\Big[r\big(T_{\nu_1} \wedge \sigma_1 \big) \Big] \leq L.
 \end{align}
 
We set $\mathcal{T} = \inf\{t > 0, \|V_t\| \ne \|v\|, \|\vtilde{t}\| \ne \|\tilde{v}\|\}$. By Lemma \ref{LemmaControlCalT} and since $\|v\| \wedge \|\tilde{v}\| \geq 1$, $\EE[r(\mathcal{T})] \leq L$.

It thus suffices to prove that $$ \EE[r(T_{\nu_1} \wedge \sigma_1 - \mathcal{T}) \mathbf{1}_{\{T_{\nu_1} \wedge \sigma_1 > \mathcal{T}\}} ] \leq L. $$
To this end, we will use Lemma \ref{LemmaIndepProcess}.

\vspace{.3cm}

Set, for all $t \geq 0$, $(X'_t,V'_t) = (X_{\CT + t}, V_{\CT + t})$ and $(\xtilde{t}',\vtilde{t}') = (\xtilde{\CT + t},\vtilde{\CT + t})$. Conditionally on $\mathcal{F}_{\CT-}$, on the event $\{T_{\nu_1} \wedge \sigma_1 > \CT \}$, the processes $(X'_t,V'_t)_{0 \leq t < T_{\nu_1} \wedge \sigma_1 - \CT}$ and $(\xtilde{t}',\vtilde{t}')_{0 \leq t < T_{\nu_1} \wedge \sigma_1 - \CT}$ are two independent (killed) free-transport processes with initial distributions $\delta_{X_{\CT}} \otimes \delta_{V_{\CT-}}$ and $\delta_{\xtilde{\CT}} \otimes \delta_{\vtilde{\CT-}}$. Indeed, by definition of $\sigma_1$ and $\nu_1$, the first and third lines of Table \ref{Tab:TableUpdate} are never used during $[\CT, T_{\nu_1} \wedge \sigma_1)$, so that the innovations $(\mathbf{Q},\tilde{\mathbf{Q}})$ are always independent or one of them is useless.

\vspace{.3cm}

Using Lemma \ref{LemmaIndepProcess}, since $T_{\nu_1} \wedge \sigma_1 - \CT \leq  T_{\nu_1}-\CT \leq S$ with the notation of the Lemma, we conclude that $$ \mathbf{1}_{\{T_{\nu_1} \wedge \sigma_1 > \CT\}} \EE[r(T_{\nu_1} \wedge \sigma_1 - \CT)|\mathcal{F}_{\CT-}] \leq \kappa \Big( 1 + r \Big(\frac{d(D)}{\|V_{\CT-}\|} \Big) + r \Big(\frac{d(D)}{\|\vtilde{\CT-}\|} \Big) \Big). $$ 
We obtain
\begin{align}
\label{EqStep12Tnu1}
\EE[r(T_{\nu_1} \wedge \sigma_1 - \CT)&\mathbf{1}_{\{T_{\nu_1} \wedge \sigma_1 > \CT\}}] \leq  \kappa \Big(  \EE \Big[ r \Big(\frac{d(D)}{\|V_{\CT-}\|} \Big) + r \Big(\frac{d(D)}{\|\vtilde{\CT-}\|} \Big) \Big]  + 1\Big) \\
&\leq \kappa \Big( \EE \Big[ r \Big(\frac{d(D)}{\|v\|} \Big) \mathbf{1}_{\{\|V_{\CT-}\| = \|v\|\}} + r \Big(\frac{d(D)}{\|V_{\CT-}\|} \Big) \mathbf{1}_{\{\|V_{\CT-}\| \ne \|v\|\}}\nonumber \\ 
& \qquad  +r \Big(\frac{d(D)}{\|\tilde{v}\|} \Big) \mathbf{1}_{\{\|\vtilde{\CT-}\| = \|\tilde{v}\|\}} 
 + r \Big(\frac{d(D)}{\|\vtilde{\CT-}\|} \Big) \mathbf{1}_{\{\|\vtilde{\CT-}\| \ne \|\tilde{v}\|\}} \Big) \Big]  + 1\Big) \nonumber \\
&\leq L, \nonumber
\end{align}
using (\ref{EqDefC0}), that $\mathcal{L} (\|V_{\CT-}\| | \|V_{\CT-}\| \ne \|v\| ) = \mathcal{L}(\|\vtilde{\CT-}\| | \|\vtilde{\CT-}\| \ne \|\tilde{v}\| ) = h_R $ and that $\|v\| \wedge \|\tilde{v}\| \geq 1$. This concludes the proof of (\ref{PreliminaryEquationTnu1}) and thus of i).

\vspace{.5cm}

For ii), we apply the same proof as for i), replacing everywhere $(v,\tilde{v})$ by $(V_{0},\tilde{V}_{0})$. We conclude that 
\begin{align*} \EE[r(T_{\nu_1} \wedge \sigma_1)] \leq \kappa \Big(1 + \EE \Big[r\Big(\frac{d(D)}{\|V_{0}\|}\Big) + r\Big(\frac{d(D)}{\|\vtilde{0}\|}\Big) \Big] \Big). 
\end{align*}

\vspace{.5cm}

We prove iii). Set, for all $k \geq 1$, $W_k = (\underline{U}_k, \underline{R}_k, \underline{\Theta}_k, \underline{\tilde{U}}_k, \underline{\tilde{R}}_k, \underline{\tilde{\Theta}}_k)$. Recall that $T_k < \tilde{T}_k$. We deduce that $W_k$ is independent of $\mathcal{F}_{T_k-}$ and is $\mathcal{F}_{\tilde{T}_k}$-measurable. Also, we have $\xtilde{T_{\nu_k}} \not \in \pD$ and $\|V_{T_{\nu_k}-}\| \wedge \|\vtilde{T_{\nu_k}-}\| \geq 1$ by definition of $\nu_k$. Hence $W_{\nu_k} \sim \Gamma_{X_{T_{\nu_k-}}, V_{T_{\nu_k}-}, \xtilde{T_{\nu_k}-}, \vtilde{T_{\nu_k}-}}$ and its law is given by the second line of (\ref{EqLawGlobalCoupling}). Thus, conditionally on $\mathcal{F}_{T_{\nu_k}-}$, $$(\underline{R}_{\nu_k}, \underline{\Theta}_{\nu_k}, \underline{\tilde{R}}_{\nu_k}, \underline{\tilde{\Theta}}_{\nu_k}) \sim \Lambda_{X_{T_{\nu_k}}, \xtilde{T_{\nu_k}}, \vtilde{T_{\nu_k}-}},$$ the random variable $\underline{U}_{\nu_k}$ satisfies $\underline{U}_{\nu_k} \sim \mathcal{U}$, is independent of $(\underline{R}_{\nu_k}, \underline{\Theta}_{\nu_k}, \underline{\tilde{R}}_{\nu_k}, \underline{\tilde{\Theta}}_{\nu_k})$ and we have $\underline{U}_{\nu_k} = \underline{\tilde{U}}_{\nu_k}$ . Recall, for $(x,\tilde{x}, \tilde{v}) \in \pD \times D \times \RR^n$, the notation $E_{x,\tilde{x},\tilde{v}}$ from Proposition \ref{CouplingProp}. We set 
$$C_{x,\tilde{x}, \tilde{v}} = \Big\{ (u,\tilde{u}, r, \tilde{r}, \theta, \tilde{\theta}) \in [0,1]^2 \times \RR_+^2 \times \mathcal{A}^2: u \leq \alpha_0, \tilde{u} \leq \alpha_0, (r,\theta,\tilde{r},\tilde{\theta}) \in E_{x,\tilde{x},\tilde{v}} \Big\}. $$
 We have
$$  \Big\{W_{\nu_k} \in C_{X_{T_{\nu_k}},  \xtilde{T_{\nu_k}}, \vtilde{T_{\nu_k}-}}\Big\} \subset A_k.  $$
Indeed, if $W_{\nu_k} \in C_{X_{T_{\nu_k}},  \xtilde{T_{\nu_k}}, \vtilde{T_{\nu_k}-}}$, we have first $\underline{U}_{\nu_k} \leq \alpha_0$, so that $V_{T_{\nu_k}} = \underline{R}_{\nu_k} \vartheta(X_{T_{\nu_k}}, \underline{\Theta}_{\nu_k})$.
In this configuration, after $T_{\nu_k}$, $(X,V)$ has its first collision at time $ T_{\nu_k} + \zeta(X_{T_{\nu_k}}, V_{T_{\nu_k}})$ while $(\xtilde{},\vtilde{})$ collides for the first time after $T_{\nu_k}$ at time $\tilde{T}_{\nu_k} = T_{\nu_k} + \zeta(\xtilde{T_{\nu_k}}, \vtilde{T_{\nu_k}})$. Moreover, recalling Definition \ref{DefiCouplingProcess}, $$\vtilde{\tilde{T}_{\nu_k}} = \underline{\tilde{R}}_{\nu_k} \vartheta(\xtilde{\tilde{T}_{\nu_k}}, \underline{\tilde{\Theta}}_{\nu_k}).$$
We obtain, recalling Notation \ref{Notationyxi} and Proposition \ref{CouplingProp}, that
\begin{align*}
 T_{\nu_{k+1}} = T_{\nu_k} + \zeta(X_{T_{\nu_k}}, V_{T_{\nu_k}}) = T_{\nu_k} + \xi(X_{T_{\nu_k}}, \underline{R}_{\nu_k}, \underline{\Theta}_{\nu_k}) &= T_{\nu_k} + \tilde{\xi}(\xtilde{T_{\nu_k}}, \vtilde{T_{\nu_k}-}, \underline{\tilde{R}}_{\nu_k}, \underline{\tilde{\Theta}}_{\nu_k}) \\
 & = \tilde{T}_{\nu_k} + \zeta(\xtilde{\tilde{T}_{\nu_k}}, \vtilde{\tilde{T}_{\nu_k}}). 
\end{align*}
and  
$$X_{T_{\nu_k +1 }} = q(X_{T_{\nu_k}}, V_{T_{\nu_k}}) = y(X_{T_{\nu_k}}, \underline{\Theta}_{\nu_k}) = \tilde{y}(\xtilde{T_{\nu_k}}, \vtilde{T_{\nu_k}}, \underline{\tilde{\Theta}}_{\nu_k}) = q(\xtilde{\tilde{T}_{\nu_k}}, \vtilde{\tilde{T}_{\nu_k}}) = \xtilde{T_{\nu_{k+1}}}. $$

We have, for all $k \geq 1$,
\begin{align*}
\PP(A_k | \mathcal{F}_{T_{\nu_k}-})&\geq \PP\Big(W_{\nu_k} \in C_{X_{T_{\nu_k}-}, \xtilde{T_{\nu_k}-}, \vtilde{T_{\nu_k}-}} \Big) \\
&=  \PP(\underline{U}_{\nu_k} \leq \alpha_0 | \mathcal{F}_{T_{\nu_k}-}) \PP \Big((\underline{R}_{\nu_k}, \underline{\Theta}_{\nu_k}, \underline{\tilde{R}}_{\nu_k}, \underline{\tilde{\Theta}}_{\nu_k}) \in E_{X_{T_{\nu_k-}}, \xtilde{T_{\nu_k}-}, \vtilde{T_{\nu_k}-}} \Big| \mathcal{F}_{T_{\nu_k}-} \Big) \\
&\geq \alpha_0 c,
\end{align*}
with $c > 0$ given by Proposition \ref{CouplingProp}.

On $A_k$, we have $X_{T_{\nu_k+1}} = \xtilde{T_{\nu_k + 1}}$, $Z_{T_{\nu_k +1}-} = \emptyset$ because, for all $i \geq 1$, $Z_{T_i-} = \emptyset$, 
and, since $\underline{U}_{\nu_k} = \underline{\tilde{U}}_{\nu_k} \leq \alpha_0$, and $\xtilde{T_{\nu_k}} \not \in \pD$,
$$ \|V_{T_{\nu_k + 1}-}\| = \|V_{T_{\nu_k}}\| = \underline{R}_{\nu_k} \ne \|V_0\|, $$
$$ \|\vtilde{T_{\nu_k+1}-}\|  = \|\tilde{V}_{\tilde{T}_{\nu_k}}\| \mathbf{1}_{\{ \|\tilde{V}_{T_{\nu_k +1}-}\| = \|\tilde{V}_{\tilde{T}_{\nu_k}}\|\}} + \|\vtilde{T_{\nu_k+1}-}\| \mathbf{1}_{\{ \|\tilde{V}_{T_{\nu_k +1}-}\| \ne \|\tilde{V}_{\tilde{T}_{\nu_k}}\|\}}, $$
with $\mathcal{L}(\|\vtilde{T_{\nu_k +1}-}\| | \|\tilde{V}_{T_{\nu_k +1}-}\| \ne \|\tilde{V}_{\tilde{T}_{\nu_k}}\| ) = \mathcal{L}(\|\vtilde{\tilde{T}_{\nu_k}}\|| \underline{\tilde{U}}_{\nu_k} \leq \alpha_0) = h_R$ from which we obtain $\PP(\|\vtilde{T_{\nu_k+1}-}\| = \|\vtilde{0}\|) = 0$.
We conclude that $A_k \subset \{\sigma_1 \leq T_{\nu_{k+1}}\}$ outside a $\PP$-null set.

\vspace{-.3cm}
\end{proof}

\begin{proof}[Proof of Theorem \ref{MainTheorem} in the convex case] \
We fix $f_0 \in \mathcal{P}(D \times \RR^n)$. We consider the coupling $(X_s,V_s,\xtilde{s},\vtilde{s}, Z_s)_{s \geq 0}$ given by Definition \ref{DefiCouplingProcess}. By Lemma \ref{LemmaMarginalCode}, for any $t > 0$, $(X_t,V_t) \sim f_t$ and $(\xtilde{t},\vtilde{t}) \sim \mu_{\infty}$. 

We prove, with the help of Lemma \ref{LemmaTnu}, that, setting $$\tau = \inf\{t > 0, (X_{t+s},V_{t+s})_{s \geq 0} = (\xtilde{t+s}, \vtilde{t+s})_{s \geq 0} \}, $$ we have $\EE[r(\tau)] < \infty$. We then conclude the proof of Theorem \ref{MainTheorem} in Step 4.

\vspace{.5cm}
\textbf{Step 1. }
Recall Notation \ref{NotatConvexProof} for $\sigma_1$ and for the sequence $(\nu_k)_{k \geq 0}$. We plan to apply Lemma \ref{LemmaSumGeometric} to show that $\EE[r(\sigma_1)] \leq \kappa$.
\begin{enumerate}
\item Set, for $k \geq 0$, $\mathcal{G}_k = \mathcal{F}_{T_{\nu_{k}} \wedge \sigma_1-}$, and for $k \geq 1$, $\tau_k = T_{\nu_{k}} \wedge \sigma_1$, which is $\mathcal{G}_k$-measurable. Also, set 
$$E_k = \{\sigma_1 \leq T_{\nu_{k}}\} \in \CG_k.$$
Set $G = \inf\{k \geq 1, E_k \text{ is realized}\}$.
\item Recall, for all $k \geq 1$, the notation $A_k$ from Lemma \ref{LemmaTnu}, iii). Observe that, according to the Lemma and since $\nu_{k+1} \geq \nu_k + 1$, there holds $A_{k-1} \subset \{\sigma_1 \leq T_{\nu_{k-1}+1}\} \subset \{\sigma_1 \leq T_{\nu_k}\} \subset E_k$. We have, for all $k \geq 1$, by Lemma \ref{LemmaTnu} iii),
$$ \PP(E_k | \mathcal{G}_{k-1}) = \PP(E_k | \mathcal{F}_{T_{\nu_{k-1}} \wedge \sigma_1-}) \geq \EE \Big[ \PP(A_{k-1} | \mathcal{F}_{T_{\nu_{k-1}}-}) \Big| \mathcal{F}_{T_{\nu_{k-1}} \wedge \sigma_1 -} \Big]  \geq c,$$
whence (\ref{HypoProbaEk}).
\item From Lemma \ref{LemmaTnu} ii) and (\ref{EqDefC0}), we have
\begin{align*}
\EE[r(\tau_1)|\CG_0] = \EE[r(T_{\nu_1} \wedge \sigma_1) ] \leq L.
\end{align*}
Moreover, by Lemma \ref{LemmaTnu} i), for all $k \geq 1$, we have, using $\mathcal{F}_{T_{\nu_{k-1}} \wedge \sigma_1 -} \subset \mathcal{F}_{T_{\nu_{k-1}}-} \subset \mathcal{F}_{T_{\nu_k}-}$,
\begin{align*}
\mathbf{1}_{\{G \geq k\}} \EE[r(\tau_{k+1} - \tau_k)|\CG_{k-1}] &\leq \EE[r(T_{\nu_{k+1}} \wedge \sigma_1 - T_{\nu_{k}} \wedge \sigma_1) | \mathcal{F}_{T_{\nu_{k-1}} \wedge \sigma_1-}] \\
&\leq r(0) + \EE \Big[ \mathbf{1}_{\{\sigma_1 > T_{\nu_{k}}\}} \EE[r(T_{\nu_{k+1}} \wedge \sigma_1 - T_{\nu_{k}}) | \mathcal{F}_{T_{\nu_{k}}-}] \Big| \mathcal{F}_{T_{\nu_{k-1}} \wedge \sigma_1 -} \Big] \\
&\leq r(0) + L,
\end{align*}
whence (\ref{IneqHypoLemmaSumGeom}). 
\end{enumerate}

We apply Lemma \ref{LemmaSumGeometric} and conclude that
$$ \EE[r(\tau_G)] \leq \kappa, $$
from which we deduce, by definition of $G$, that
$$ \EE[r(\sigma_1)] = \EE[r(\sigma_1 \wedge T_{\nu_G})] = \EE[r(\tau_G)] \leq \kappa.$$

\vspace{.5cm}

\textbf{Step 2.} We introduce the sequence $(\sigma_i)_{i \geq 0}$ defined by $\sigma_0 = 0$, $\sigma_1$ defined by Notation \ref{NotatConvexProof}, iii), and for all $k \geq 1$,
$$ \sigma_{k+1} = \inf\{t > \sigma_k, X_t = \xtilde{t} \in \pD, Z_{t-} = \emptyset, \|V_{t-}\| \ne \|V_{\sigma_k}\|, \|\vtilde{t-}\| \ne \|\vtilde{\sigma_k}\|\}. $$
We plan to apply Lemma \ref{LemmaSumGeometric}. 
\begin{enumerate}
\item We set $\mathcal{G}_0$ to be the completion of the trivial $\sigma$-algebra and, for $k \geq 1$, $\CG_k = \mathcal{F}_{\sigma_{k+1}-}$. We also set, for all $k \geq 1$, $\tau_k = \sigma_{k+1}$ which is $\CG_k$-measurable, and $E_k = \{V_{\sigma_k} = \vtilde{\sigma_k} \} \in \CG_k$. 
We set $N = \inf \{k \geq 1, E_k \text{ is realized } \}$.
\item Let, for all $k \geq 1$, $(\mathbf{Q}_k, \tilde{\mathbf{Q}}_k) = ((\mathbf{U}_k, \mathbf{R}_k, \mathbf{\Theta}_k), (\tilde{\mathbf{U}}_k, \tilde{\mathbf{R}}_k, \tilde{\mathbf{\Theta}}_k))$ be the couple random variables used to define $V_{\sigma_k}$ and $\vtilde{\sigma_k}$. Since $X_{\sigma_k} = \xtilde{\sigma_k}$ and $Z_{\sigma_k-} = \emptyset$, we are in the situation of line 3 of Table \ref{Tab:TableUpdate}, hence $\mathbf{Q}_k = \tilde{\mathbf{Q}}_k$, so that if $\mathbf{U}_k \leq \alpha_0$, $$V_{\sigma_k} =  w (X_{\sigma_k},V_{\sigma_k-}, \mathbf{Q}_k) = w(X_{\sigma_k}, \vtilde{\sigma_k-}, \mathbf{Q}_k) = \vtilde{\sigma_k}.$$ Since $\mathbf{Q}_k$ is independent of $\mathcal{F}_{\sigma_k-}$,
$$ \PP(E_k | \mathcal{G}_{k-1}) \geq \PP(\mathbf{U}_k \leq \alpha_0 |  \mathcal{F}_{\sigma_k-}) = \alpha_0, $$
whence (\ref{HypoProbaEk}).
\item Note that for $k \geq 1$, 
\begin{align}
\label{EqControlVsigma}
\EE \Big[r \Big( \frac{d(D)}{\|V_{\sigma_k-}\|} \Big) + r \Big( \frac{d(D)}{\|\vtilde{\sigma_k-}\|} \Big) \Big| \mathcal{F}_{\sigma_{k-1}-} \Big] \leq 2 C_0,
\end{align}
using that $\mathcal{L}(\|V_{\sigma_k-}\| | \mathcal{F}_{\sigma_{k-1}-}) = \mathcal{L}(\|\vtilde{\sigma_k-}\| | \mathcal{F}_{\sigma_{k-1}-}) = h_R$, since $\|V_{\sigma_k-}\| \ne \|V_{\sigma_{k-1}}\|$ and $\|\vtilde{\sigma_k-}\| \ne \|\vtilde{\sigma_{k-1}}\|$ by definition of $\sigma_k$. 
By Step 1, Lemma \ref{LemmaTnu}, ii), the strong Markov property and the definition of $(\sigma_i)_{i \geq 0}$, we have, for all $k \geq 1$,
$$ \EE[r(\sigma_{k+1} - \sigma_k)|\mathcal{F}_{\sigma_{k}-}] \leq \kappa \Big(1 + r\Big(\frac{d(D)}{\|V_{\sigma_k-}\|} \Big) + r \Big(\frac{d(D)}{\|\vtilde{\sigma_k-}\|} \Big) \Big), $$
so that, using (\ref{EqControlVsigma}) and that $\mathcal{F}_{\sigma_{k-1}-} \subset \mathcal{F}_{\sigma_k-}$,
$$ \EE[r(\sigma_{k+1} - \sigma_k)|\mathcal{F}_{\sigma_{k-1}-}] \leq \kappa. $$
With this at hand, we show that (\ref{IneqCondLemmaSumGeom}) holds. First, by Step 1,
$$\EE[r(\tau_1)|\mathcal{G}_0] = \EE[r(\sigma_2)] \leq C \Big( \EE[r(\sigma_2 - \sigma_1)] + \kappa \Big) \leq L. $$
Moreover, for $k \geq 1$,
\begin{align*}
\mathbf{1}_{\{N \geq k\}} \EE[r(\tau_{k+1}-\tau_k)|\mathcal{G}_{k-1}] & \leq \EE [r(\sigma_{k+2} - \sigma_{k+1}) |\mathcal{F}_{\sigma_k-}] \leq \kappa,
\end{align*}
whence (\ref{IneqCondLemmaSumGeom}).
\end{enumerate}
We conclude by Lemma \ref{LemmaSumGeometric} that $\EE[r(\sigma_N)] \leq \kappa$. 

\vspace{.5cm}

\textbf{Step 3} Using Lemma \ref{Lemmataudefinitive}, since $(X_{\sigma_N}, V_{\sigma_N}) = (\xtilde{\sigma_N},\vtilde{\sigma_N})$ and $Z_{\sigma_N-} = \emptyset$, we conclude that $\tau \leq \sigma_N$, hence $\EE[r(\tau)] \leq \kappa$ by Step 2.

\vspace{.5cm} 

\textbf{Step 4. } Recall that for two probability measures $\mu, \nu$,
\begin{align*}
\| \mu - \nu \|_{TV} = \inf_{X \sim \mu, Y \sim \nu} \PP(X \ne Y).
\end{align*}
Hence for all $t \geq 0$,
\begin{align*}
\|f_t - \mu_{\infty}\|_{TV} \leq \PP\Big((X_t,V_t) \ne (\xtilde{t}, \vtilde{t}) \Big) = \PP(\tau > t),
\end{align*}
according to our definition of $\tau$. Finally, we use Step 3 and Markov's inequality to conclude that
\begin{align*}
\|f_t - \mu_{\infty}\|_{TV} \leq \frac{\EE[r(\tau)]}{r(t)} \leq \frac{\kappa}{r(t)},
\end{align*}
for all $t \geq 0$.
\end{proof}

\vspace{.5 cm}

\section{Extension to a general regular domain}
\label{Extension}

In this section, we extend the previous result on a convex bounded domain (open, connected) to the general case of a $C^2$ bounded domain. 

\subsection{Notations and preliminary results.}

In this subsection, we introduce the notion of communication between boundary points, derive an important corollary from this definition and prove a preliminary lemma that will be key to obtain a result similar to Proposition \ref{CouplingProp} in the general setting.

We introduce first a notion of communicating boundary points taken from Evans \cite{evans2001}.
\begin{defi}
\label{DefCommunication} We say that two points $x \in \pD$, $y \in \pD$ communicate, and write $x \leftrightarrow y$ if $tx + (1-t)y \in D$ for all $t \in (0,1)$, $n_x \cdot (y-x) > 0$ and $n_y \cdot (x-y) > 0$. Given a set $E \subset \pD$ we say that $x \in \pD$ communicates with $E$ and write $x \leftrightarrow E$ if $x \leftrightarrow y$ for all $y \in E$. Given two sets $E_1, E_2 \subset \pD$, we say that $E_1$ and $E_2$ communicate, and write $E_1 \leftrightarrow E_2$ if $x \leftrightarrow y$ for all $(x,y) \in E_1 \times E_2$.
\end{defi}

Since $D$ is regular, the condition $tx + (1-t)y \in D$ for all $t \in (0,1)$ implies that $n_x \cdot (y-x) \geq 0$. The previous notion forbids the case where $(y-x)$ is tangent to $\pD$ at $x$. 

Recall that we denote by $\mathcal{H}$ the $n-1$ dimensional Hausdorff measure. The goal of this subsection is to prove the following lemma.

\begin{lemma}
\label{LemmaDefiningCriticalRegion}
There exists $\kappa_0$, $d_0 > 0$, $F \subset \pD$, $\mathcal{R} \subset \pD$ with $F$, $\mathcal{R}$ compact and $F \leftrightarrow \mathcal{R}$ such that $\underset{(x,y) \in F \times \mathcal{R}}{\inf} \|x-y\| \geq d_0$ and $\mathcal{H}(F) \wedge \mathcal{H}(\mathcal{R}) \geq \kappa_0$. 

\end{lemma}

Recall that $d(D)$ denotes the diameter (in the usual sense) of $D$ and that for $x \in \RR^n$ and $r > 0$, we write $B(x,r) = \{y \in \RR^n, \|x-y\| < r\}$ for the Euclidian ball centered at $x$, with radius $r$, in $\RR^n$. We denote $\bar{B}(x,r)$ the corresponding closed ball.
\begin{notat}
For $x \in \pD$, $r > 0$, we set $B_{\pD}(x,r) := B(x,r) \cap \pD$. 
\end{notat} 

\begin{lemma}
\label{LemmaBallCommunication}
Let $x, y \in \pD$ with $x \leftrightarrow y$. There exists $\epsilon_0 > 0$ such that $\BpD(x,\epsilon_0) \leftrightarrow \BpD(y, \epsilon_0)$.
\end{lemma}

\begin{proof}
\textbf{Step 1.} Recall first that since $D$ is $C^2$, $D$ satisfies the uniform ball condition: there exists $r_D > 0$ such that for all $z \in \pD$, there exists ${B}_z$ a ball of radius $r_D$ with center $z + r_D n_z$ such that $B_z \subset D$ and $\bar{B}_z \cap \pD = \{z\}$.  As a consequence, for $\beta > 0$ to choose later, setting $t_0 = \frac{r_D \beta}{2d(D)} \wedge \frac14$, there holds that for all $x, z \in \pD$ with $n_z \cdot \frac{x-z}{\|x-z\|} \geq \frac{\beta}{2}$, $(1-t)z + tx \in B_z \subset D$. Indeed
\begin{align*}
\|(1-t)z + tx - z - r_D n_z\|^2 &= t^2 \|x-z\|^2 + r_D^2 - 2t r_D \big( n_z \cdot (x-z) \big) \\
&\leq r_D^2 + t\|x-z\| \Big(td(D) - r_D \beta \Big)
\end{align*} 
and since $t < t_0 < \frac{r_D\beta}{d(D)}$, the result follows.

\vspace{.5cm} 

\textbf{Step 2.}
Let $x, y \in \pD$ with $x \leftrightarrow y$. We have $n_x \cdot (y-x) > 0$, $n_y \cdot (x-y) > 0$ and $x \ne y$, hence $\beta := (n_y \cdot \frac{(x-y)}{\|x-y\|}) \wedge (n_x \cdot \frac{(y-x)}{\|y-x\|})  > 0$. Since $z \to n_z$ is continuous by regularity of $D$, there exists $\delta > 0$ such that for all $x' \in \BpD(x,\delta)$, $y' \in \BpD(y,\delta)$, $(n_{y'} \cdot \frac{(x'-y')}{\|x'-y'\|}) \wedge (n_{x'} \cdot \frac{(y'-x')}{\|y'-x'\|}) \geq \frac{\beta}{2}$. By Step 1, for all $t \in (0,t_0)$,
$$ (1-t)y' + tx' \in B_{y'} \subset D, $$
and, for all $t \in (1-t_0, 1)$, 
$$ (1-t)y' + tx' \in B_{x'} \subset D. $$
We conclude that for all $t \in (0,t_0) \cup (1-t_0, 1)$, $tx' + (1-t)y' \in D$.

\vspace{.5cm}

\textbf{Step 3.}
Since $x \leftrightarrow y$ by assumption, for all $t \in [t_0, 1-t_0]$, $tx + (1-t)y \in D$. By compactness and continuity of $a \to d(a,\pD) := \inf_{z \in \pD} \|a-z\|$, there exists $\eta > 0$ such that for all $t \in [t_0, 1-t_0]$, $B((1-t)y + tx, \eta) \subset D$. Hence, for $\delta$ given by Step 2, for all $x' \in \BpD(x,\delta \wedge \eta)$, $y' \in \BpD(y,\delta \wedge \eta)$, for all $t \in [t_0, 1-t_0]$, 
$$ \|(1-t)y' + tx' - (1-t)y - tx\| \leq \max(\|y'-y\|,\|x'-x\|) < \eta, $$
so that $(1-t)y' + tx' \in B((1-t)y+tx, \eta) \subset D$. Setting $\epsilon_0 = \delta \wedge \eta$, we conclude that, for all $x' \in \BpD(x,\epsilon_0)$, $y' \in \BpD(y,\epsilon_0)$,
$$n_{y'} \cdot (x' - y') >  0 \qquad \text{and} \qquad n_{x'} \cdot (y'-x') > 0$$ by Step 2 and for all $t \in (0,1)$, $tx' + (1-t)y' \in D$ by Steps 1 and 2. 
\end{proof}

\begin{proof}[Proof of Lemma \ref{LemmaDefiningCriticalRegion}.] Let $x,y \in \pD$ such that $x \leftrightarrow y$. 
Set $\bar{d} = \|x-y\|$. Using Lemma \ref{LemmaBallCommunication}, there exists $\epsilon_0> 0$ such that, setting $V_x := \BpD(x,\epsilon_0)$, $V_y:= \BpD(y,\epsilon_0)$, $V_x \leftrightarrow V_y$. Upon reducing the value of $\epsilon_0$, we can assume that for any $x' \in V_x, y' \in V_y$, $\|x' - y' \| \geq \frac{\bar{d}}{2}$. 
We conclude by setting $F = \bar{B}(y, \frac{\epsilon_0}{2}) \cap \pD$, $\mathcal{R} = \bar{B}(x,\frac{\epsilon_0}{2}) \cap \pD$ and $d_0 = \frac{\bar{d}}{2}$.
\end{proof}

\subsection{Uniform lower bound on the density of the $n_0$-th collision.}

%Recall from Remark \ref{RmkExtensionDefiProcess} the definition of a free-transport process $(X_t, V_t)_{t \geq 0}$ with initial distribution $\mu_0 \in \mathcal{P}(\partial_+ G)$.
We introduce the following notation.
\begin{notat}
Let $(x_0,v_0) \in \partial_+ G \cup (D \times \RR^n)$. For a free-transport process $(X_t, V_t)_{t \geq 0}$ with initial condition $X_0 = x_0$, $V_{0-} = v_0$, we set $T_0 = \zeta(X_0,V_0)$ and for $i \geq 0$, $T_{i+1} = \inf \{t > T_i, X_t \in \pD\}$. For all $k \geq 1$, we denote $P^k_{v_0}(x_0,dz)$ the law of $X_{T_k}$.  
\end{notat}

The goal of this section is to prove the following property.
\begin{prop}
\label{corollaryEvans}
There exist $n_0 \geq 1$, $\nu_0 > 0$ and $\delta_0 > 0$ such that, for all $(x_0,v_0) \in \partial_+ G \cup (D \times \RR^n)$,
\begin{align*}
P^{n_0}_{v_0}(x_0,dz) \geq \nu_0 dz,
\end{align*}
where we recall that $dz$ stands for the $n-1$ dimensional Hausdorff measure.
Moreover, for all $A \subset \pD$, setting $O_0 = \{\|V_{T_0}\| \ne \|V_{T_0-}\|, \dots, \|V_{T_{n_0-1}}\| \ne \|V_{T_{n_0-1}-}\|\}$,
\begin{align*}
\PP\Big(X_{T_{n_0}} \in A, \underset{i \in \{1,\dots,n_0\}}{\min} \|X_{T_{i}} - X_{T_{i-1}}\| \geq \delta_0 & \Big|(X_0, V_{0-}) = (x_0,v_0), O_0\Big)
\geq \nu_0 \mathcal{H}(A). 
\end{align*}
\end{prop}

We recall first a result from Evans from which we will derive a key feature of our model:
\begin{prop}[\cite{evans2001}, Proposition 2.7]
\label{PropEvans}
For any $C^1$ bounded domain $D$, there exist an integer $N$ and a finite set $\Delta \subset \partial D$ for which the following holds: for all $z', z'' \in \partial D$, there exist $z_0, \dots, z_N$ with $z' = z_0$, $z'' = z_N$, $\{z_1, \dots, z_{N-1}\} \subset \Delta$ and $z_k \leftrightarrow z_{k+1}$ for $0 \leq k \leq N-1$.
\end{prop}

\begin{coroll}
\label{CorollGivingDelta}
There exist $\delta > 0$ and $\eta > 0$ such that for all $(x_0,y_0) \in (\pD)^2$, there exists $z_1, \dots, z_{N+1} \in \Delta$, with $N$ and $\Delta$ given by Proposition \ref{PropEvans}, such that, setting $z_0 = x_0$, $z_{N+2} = y_0$, $z_i \leftrightarrow z_{i+1}$ for all $i \in \{0, \dots, N+1\}$ and 
\begin{align}
\label{Eq1CorollDelta}
|(z_i - z_{i+1}) \cdot n_{z_i}| |(z_i - z_{i+1}) \cdot n_{z_{i+1}}| \geq 2\delta,
\end{align}
moreover, for all $z_1' \in \BpD(z_1,\eta), \dots, z_{N+1}' \in \BpD(z_{N+1},\eta)$, setting $z_0' = z_0$, $z_{N+2}'= z_{N+2}$, $z'_i \leftrightarrow z'_{i+1}$ for all $i \in \{0, \dots, N+1\}$ and 
\begin{align}
\label{Eq2CorollDelta}
|(z'_i - z'_{i+1}) \cdot n_{z'_i}| |(z'_i - z'_{i+1}) \cdot n_{z'_{i+1}}| \geq \delta.\end{align}
\end{coroll}

\begin{proof}

\textbf{Step 1.} By \cite[Lemma 2.3]{evans2001}, for $z \in \pD$, the set $U_z = \{z' \in \pD, z' \leftrightarrow z\}$ is open in $\pD$ and non-empty. Using this result and the fact that $D$ is $C^1$, we find that for all $z \in \Delta$,
$$ x \to |(z-x) \cdot n_z| |(z-x) \cdot n_x| \mathbf{1}_{U_z}(x), $$
is lower semi-continuous, and positive on $U_z$. Using Proposition \ref{PropEvans}, that $\Delta$ is finite, and that the maximum of two lower semi-continuous functions is lower semi-continuous, we deduce that the function $I: \pD \to \RR_+$ defined by
$$I(x) =  \max_{z \in \Delta} \Big( |(z-x) \cdot n_z| |(z-x) \cdot n_x| \mathbf{1}_{U_z}(x) \Big), $$
is lower semi-continuous. Moreover, since for all $x \in \pD$, there exists $z \in \Delta$ such that $x \leftrightarrow z$ by Proposition \ref{PropEvans}, $I > 0$ on $\pD$. We conclude by compactness that there exists $\delta_0 > 0$ such that $I(x) > 2\delta_0$ for all $x \in \pD$.

\vspace{.3cm}

\textbf{Step 2.} Set $$\delta' := \frac12 \min_{z, z' \in \Delta, z \leftrightarrow z'} |(z-z') \cdot n_z| |(z-z') \cdot n_{z'}| > 0, $$
since $\Delta$ is finite.  Let $x_0,y_0 \in \pD$. Choose $z_1$ such that $I(x_0) = |(z_1 - x_0) \cdot n_{z_1}||(z_1 - x_0) \cdot n_{x_0}|$, $z_{N+1}$ such that $I(y_0) = |(z_{N+1} - y_0) \cdot n_{y_0}| |(z_{N+1} - y_0) \cdot n_{z_{N+1}}|$.
 By Proposition \ref{PropEvans}, there exists $z_2, \dots, z_N$ such that $z_i \leftrightarrow z_{i+1}$ for all $i \in \{1, \dots, N\}$. Since $z_0 = x_0$, $z_{N+2} = y_0$, $z_i \leftrightarrow z_{i+1}$ for all $i \in \{0, \dots, N+1\}$ and for all $i \in \{1, \dots, N\}$,  we have,
$$ |(z_i-z_{i+1}) \cdot n_{z_i}| |(z_i-z_{i+1}) \cdot n_{z_{i+1}}| \geq 2\delta', $$
while, using Step 1,
$$ \Big(|(z_1 - z_0) \cdot n_{z_0}| |(z_1 - z_0) \cdot n_{z_1}| \Big) \wedge \Big(|(z_{N+1} - z_{N+2}) \cdot n_{z_{N+1}}| |(z_{N+1} - z_{N+2}) \wedge n_{z_{N+2}}| \Big) \geq 2\delta_0. $$
We set $\delta = \delta_0 \wedge \delta'$ to conclude the proof of (\ref{Eq1CorollDelta}).

\vspace{.3cm}

\textbf{Step 3.} Consider the function $H$ defined on $(\pD)^2$ by
$$ H(x,z) = \big((z-x) \cdot n_x \big)\big((x-z) \cdot n_z \big). $$
Since $D$ is $C^1$, $H$ is continuous on $(\pD)^2$ and also uniformly continuous by compactness and Heine's theorem. Hence there exists $\eta_0$ such that, $$\Big[(x,z) \in (\pD)^2, (x',z') \in (\pD)^2, \|(x,z) - (x',z')\| \leq \eta_0 \Big] \implies \Big[ \big|H(x,z) - H(x',z') \big| \leq \frac{\delta}{2} \Big].$$
 On the other hand, for all $(x,y) \in (\pD)^2$ with $x \leftrightarrow y$, there exists $\epsilon_{x,y} > 0$ such that we have $\BpD(x, \epsilon_{x,y}) \leftrightarrow \BpD(y,\epsilon_{x,y})$, see Lemma \ref{LemmaBallCommunication}. Setting $\eta_1 = \min_{z,z' \in \Delta, z \leftrightarrow z'} \epsilon_{z,z'} > 0$, we deduce that for all $z,z' \in \Delta$ with $z \leftrightarrow z'$, $\BpD(z,\eta_1) \leftrightarrow \BpD(z',\eta_1)$. We claim that setting $\eta = \eta_1 \wedge \eta_0$ concludes the proof of (\ref{Eq2CorollDelta}). Indeed, for $z_1' \in \BpD(z_1,\eta)$, recalling that $z'_0 = x_0$ and (\ref{Eq1CorollDelta}),
$$ H(z_1',z_0') = H(z_1,z_0') - (H(z_1,z_0') - H(z_1',z_0')) \geq 2\delta - \frac{\delta}{2} \geq \frac{3\delta}{2}, $$ 
and the same argument applies replacing $z_1'$ by $z_{N+1}' \in \BpD(z_{N+1},\eta)$ and $z_0'$ by $z'_{N+2} = y_0$. Finally, for $i \in \{1, \dots, N\}$, $z'_i \in B_{\pD}(z_i,\eta)$, $z'_{i+1} \in \BpD(z_{i+1},\eta)$, we have $z'_i \leftrightarrow z'_{i+1}$ and 
\begin{align*}
 H(z_i',z_{i+1}') 
 &= H(z_i, z_{i+1}) - (H(z_i, z_{i+1}) - H(z_i,z_{i+1}')) -  (H(z_i,z_{i+1}') - H(z_i',z_{i+1}')) \\
 &\geq 2\delta - \frac{\delta}{2} - \frac{\delta}{2} \geq \delta. 
 \end{align*}

\vspace{-.5cm}

\end{proof}

Recall the notations $\zeta$ and $q$ from (\ref{defsigma}) and (\ref{defq}).
\begin{lemma}
\label{LemmaDensityExtension}
Let $x \in \pD$. For $V \sim c_0 M(v) |v \cdot n_x| \mathbf{1}_{\{v \cdot n_x > 0\}}$, the joint law of $(\zeta(x,V), q(x,V))$ admits a density $\mu_x$ on $\RR_+ \times  \pD$ given by
$$ \mu_x(\tau,z) = c_0 M\Big(\frac{z-x}{\tau}\Big) \frac1{\tau^{n+2}} |(z-x) \cdot n_x|  |(z-x) \cdot n_z| \mathbf{1}_{\{z \leftrightarrow x\}}. $$
\end{lemma}

\begin{proof}
The computation is the same as the one of Lemma \ref{LemmaDensityCoupling}.
% Note that, for $x,z \in \pD$, by definition of $z \leftrightarrow x$, $|(z-x)\cdot n_x| |(z-x) \cdot n_z| > 0$. In particular, if $(z_n)_{n \geq 1}$ is such that $z_n \leftrightarrow x$ for all $n \geq 1$, $z_n \to z$ as $n \to \infty$ for some $z \in \pD$ with $(z-x)$ tangent to $\pD$ at $x$, since $|(z_n -x) \cdot n_x| \to 0$, $\mu_x(\tau,z_n) \to 0$.
% to find that the joint density $\mu_x$ of $(\zeta(x,V), q(x,V))$ is given, for all $z \in \pD$ with $z \leftrightarrow x$, for all $\tau > 0$, by
%$$ \mu_x(\tau,z) = c_0 M\Big(\frac{z-x}{\tau}\Big) \frac{1}{\tau^{n+2}} |(z-x) \cdot n_x| | (z-x) \cdot n_z|. $$
\end{proof}

\begin{proof}[Proof of Corollary \ref{corollaryEvans}]
We will show that there exist $n_0 \geq 1$, $\nu_0 > 0$ and $\delta_0 > 0$ such that for all $(x_0,v_0) \in \partial_+ G \cup (D \times \RR^n)$, for all $A \subset \pD$,
\begin{align*}
P_A&:= \PP \Big(X_{T_{n_0}} \in A, \underset{i \in \{1, \dots, n_0\}}{\min} \|X_{T_i} - X_{T_{i-1}}\| \geq \delta_0, O_0 \Big| (X_0,V_{0-}) = (x_0, v_0) \Big) 
\geq \nu_0 \mathcal{H}(A).
\end{align*} 
This will imply both statements.
We set, for all $x \in \pD$, the marginal law
$$ \nu_x(z) = \int_0^{\infty} \mu_x(\tau,z) d\tau = \mathbf{1}_{\{z \leftrightarrow x\}} c_0 |(z-x) \cdot n_x| |(x-z) \cdot n_z| \int_0^{\infty} M \Big( \frac{z-x}{\tau} \Big) \frac{1}{\tau^{n+2}} d\tau. $$
Let $x = q(x_0,v_0) \in \pD$, so that $x = x_0$ if $(x_0,v_0) \in \partial_+ G$. 
Let $\Delta, N$ given by Proposition \ref{PropEvans}.
For all $y \in \pD$, by Corollary \ref{CorollGivingDelta}, there exist $z_1(y), \dots, z_{N+1}(y) \in \Delta$ such that, setting $z'_0 = x$, $z'_{N+2} = y$ and taking, for all $i \in \{1, \dots, N+1\}$, $z'_i \in \BpD(z_i(y),\eta)$, we have, for all $j \in \{0, \dots, N+1\}$,  $z'_j \leftrightarrow z'_{j+1}$ and
\begin{align}
\label{IneqDeltaProofRegion}
 |(z'_{j+1} -z'_j) \cdot n_{z'_j}||(z'_j - z'_{j+1}) \cdot n_{z'_{j+1}}| \geq \delta,
 \end{align} 
where $\delta > 0$ and $\eta > 0$ are given by Corollary \ref{CorollGivingDelta}. This inequality implies $\|z'_{j+1} - z'_j\| \geq \sqrt{\delta}$, and in particular we have $d(D) \geq \sqrt{\delta}$. Let $A \subset \partial D$.  We introduce the event $$O_1 = \Big\{\|V_{T_i}\| \ne \|V_{T_i-}\| \text{ for all } i \in \{0, \dots, N + 1\}, \|X_{T_i} - X_{T_{i-1}}\| \geq \sqrt{\delta} \text{ for all } i \in \{1, \dots, N+2\} \Big\}, $$ and we have, with the choice $n_0 = N+2$, $\delta_0 = \sqrt{\delta}$,
$$ P_A = \PP \Big(\{X_{T_{N+2}} \in A\} \cap O_1|X_0 = x_0, V_{0-} = v_0\Big). $$
Since on the event $O_1$, all reflections are diffuse, and recalling the definition of $\alpha_0$, see Hypothesis \ref{HypoM}, and that $X_{T_0} = x$,
\begin{align*}
P_A &\geq \alpha_0^{N+2} \int_{y \in A} \int_{z'_1 \in \BpD(z_1(y),\eta)} \nu_x(z'_1) \int_{z'_2 \in \BpD(z_2(y),\eta)} \nu_{z'_1}(z'_2) \\
& \qquad\times \dots \times \int_{z'_{N+1} \in \BpD(z_{N+1}(y),\eta)} \nu_{z'_{N}}(z'_{N+1})  \nu_{z_{N+1}'}(y) dz'_{N+1} \dots  dz'_1 dy.
\end{align*}
For $\tau \in ( \frac{d(D)}{\delta_1}, \frac{d(D)}{\delta_1} + 1)$ with $\delta_1$ given by Hypothesis \ref{HypoM}, for all $y \in A$, $z'_{N+1} \in \BpD(z_{N+1}(y),\eta)$,
$$ \mu_{z'_{N+1}}(\tau,y) = c_0 M \Big(\frac{z'_{N+1} - y}{\tau} \Big) \frac{1}{\tau^{n+2}} |(z'_{N+1} - y) \cdot n_y| |(y - z'_{N+1}) \cdot n_{z'_{N+1}}| \geq \kappa_1, $$
with, recalling (\ref{IneqDeltaProofRegion}) and that $\sqrt{\delta} \leq \|z'_{N+1}-y\| \leq d(D)$, $$\kappa_1 = c_0 \Big( \inf_{\|v\| \in (\frac{\delta_1 \sqrt{\delta}}{d(D) + \delta_1}, \delta_1)} M(v) \Big)  \Big(\frac{\delta_1}{d(D)}\Big)^{n+2}  \delta > 0,$$
so that the infimum above is positive using Hypothesis \ref{HypoM}.
We thus have
$$ \nu_{z'_{N+1}}(y) \geq \int_{\frac{d(D)}{\delta_1}}^{\frac{d(D)}{\delta_1} + 1} \mu_{z'_{N+1}}(\tau,y) d\tau \geq \kappa_1.$$
Working similarly for the other terms, we conclude that

\begin{align*}
P_A &\geq \alpha_0^{N+2} \kappa_1^{N+2} \int_{y \in A} \int_{z'_1 \in \BpD(z_1(y),\eta)}  \int_{z'_2 \in \BpD(z_2(y),\eta)}  \\
& \qquad \times \dots \times \int_{z'_{N+1} \in \BpD(z_{N+1}(y),\eta)} dz'_{N+1} \dots dz'_2 dz'_1 dy \\
&\geq \alpha_0^{N+2} \kappa_1^{N+2} \epsilon^{N+1} \mathcal{H}(A),
\end{align*}
where $\epsilon  = \inf_ {x \in \pD} \mathcal{H}(\BpD(x,\eta)) > 0$. This completes the proof. 

\end{proof}

\subsection{Coupling of $(R, \Theta, \tilde{R}, \tilde{\Theta})$}

In this subsection, we exhibit a coupling in a certain appropriate regime, to derive a result similar to Proposition \ref{CouplingProp} in the general setting. We let $d_0, \kappa_0 > 0$ and $F, \mathcal{R} \subset \pD$ be the positive constants and compact regions of the boundary given by Lemma \ref{LemmaDefiningCriticalRegion}. Recall Notation \ref{Notationyxi} for the maps $\xi, \tilde{\xi}, y, \tilde{y}$. We also recall that $\mathcal{A} = (-\frac{\pi}{2}, \frac{\pi}{2}) \times [0,\pi)^{n-2}$ and the notation $\Upsilon$ introduced in Lemma \ref{NotatHM}. 

\begin{prop}
\label{CouplingPropGeneral}
There exists a constant $c > 0$ such that for all $x_0 \in F$, $\tilde{x}_0 \in D$, 
$\tilde{v}_0 \in \RR^n$ with $\|\tilde{v}_0\| \geq 1$ and $q(\tilde{x}_0,\tilde{v}_0) \in F$, there exists $\Lambda_{x_0, \tilde{x}_0, \tilde{v}_0} \in \mathcal{P}(((0,\infty) \times \mathcal{A} )^2)$ such that if $(R, \Theta, \tilde{R}, \tilde{\Theta})$ has law $\Lambda_{x_0, \tilde{x}_0, \tilde{v}_0}$, $(R, \Theta) \sim \Upsilon$, $(\tilde{R}, \tilde{\Theta}) \sim \Upsilon$ and for $E_{x_0, \tilde{x}_0, \tilde{v}_0}$ defined by
\begin{align*}
E_{x_0, \tilde{x}_0, \tilde{v}_0} := \Big\{(r, \theta, \tilde{r}, \tilde{\theta}) \in (\RR_+ \times \mathcal{A})^2:  y(x_0,\theta) = \tilde{y}(\tilde{x}_0, \tilde{v}_0, \tilde{\theta}), \xi(x_0,r,\theta) = \tilde{\xi}(\tilde{x}_0, \tilde{v}_0, \tilde{r}, \tilde{\theta})\Big\},
\end{align*} 
we have
\begin{align}
\label{IneqCouplingLemma2} 
\PP \big((R,\Theta, \tilde{R}, \tilde{\Theta}) \in E_{x_0, \tilde{x}_0, \tilde{v}_0} \big) \geq c.
\end{align}
\end{prop}

\begin{proof}
\textbf{Step 1.} We prove that there exists $c > 0$ such that 
\begin{align}
\label{EqKeyPropCouplingGeneral}
\inf_{(x,\tilde{x}, \tilde{t}) \in F \times F \times [0,d(D))} \int_{\{z \in \pD, z \leftrightarrow x, z \leftrightarrow y\}} \int_{\tilde{t}}^{\infty} \Big[ \mu_x(\tau,z) \wedge \mu_{\tilde{x}}(\tau - \tilde{t}, z) \Big] d\tau dz \geq c.
\end{align}
Note that by compactness of $\mathcal{R} \times F$, using continuity properties and that $\mathcal{R} \leftrightarrow F$, we have 
\begin{align}
\label{EqC'ProofCouplingExtension}
 c' := \inf_{z \in \mathcal{R}, y \in F} |(z-y) \cdot n_y| \wedge |(z-y) \cdot n_z| > 0. 
 \end{align}
For $(x,\tilde{x}, \tilde{t}) \in F \times F \times [0,d(D))$, we set
$$ J:= \int_{z \in \mathcal{R}} \int_{\tilde{t}}^{\infty} \Big[ \mu_x(\tau,z) \wedge \mu_{\tilde{x}}(\tau - \tilde{t}, z)  \Big] d\tau dz, $$
and it suffices to verify that $J$ is lower bounded away from 0. 
Recall the definition of $\bar{M}$ and $\delta_1$ from Hypothesis \ref{HypoM}. Using Lemma \ref{LemmaDensityExtension} and (\ref{EqC'ProofCouplingExtension}), we easily find
$$ J \geq c' c_0 \int_{z \in \mathcal{R}} \int_{\tilde{t}}^{\infty} \Big(\frac{1}{\tau}\Big)^{n+2} \min \Big( M \Big( \frac{z-x}{\tau} \Big), M \Big( \frac{z-\tilde{x}}{\tau - \tilde{t}} \Big) \Big) d\tau dz. $$
Note that, for $\tau \geq d(D)(1 + \frac{1}{\delta_1})$, for all $z \in \mathcal{R}$, $\frac{\|z-x\|}{\tau} \vee \frac{\|z-\tilde{x}\|}{\tau - \tilde{t}} \leq \delta_1$ using that $\tilde{t} \leq d(D)$, whence
$$M \Big(\frac{z-x}{\tau} \Big) \wedge M \Big(\frac{z-\tilde{x}}{\tau - \tilde{t}} \Big) \geq \inf_{\|v\| \leq \delta_1} \bar{M}(v) =: \kappa_1 > 0. $$
We deduce that
$$ J \geq \kappa_1 c' c_0 \int_{z \in \mathcal{R}} \int_{d(D)(1 + \frac1{\delta_1})}^{+\infty}  \Big(\frac{1}{\tau}\Big)^{n+2} d\tau dz \geq \kappa \mathcal{H}(\mathcal{R}) > 0,$$
with $\kappa$ a positive constant not depending on $x, \tilde{x}, \tilde{t}$. 
This concludes the proof of (\ref{EqKeyPropCouplingGeneral}).

\vspace{.5cm}

\textbf{Step 2.} We conclude as in the proof of Proposition \ref{CouplingProp}, see Step 2.
\end{proof}

\subsection{Construction of the coupling}

In comparison with the convex case, we change the definition of the law $\Gamma$ on $([0,1] \times \RR_+ \times \mathcal{A})^2$ by setting, for $(x,v,\tilde{x},\tilde{v}) \in (\bar{D} \times \RR^n)^2$ with $x \in \pD$ or $\tilde{x} \in \pD$,
\begin{align}
\label{EqLawGlobalCouplingExtension}
\Gamma_{x,v,\tilde{x}, \tilde{v}}(du, &dr, d\theta, d\tilde{u}, d\tilde{r}, d\tilde{\theta}) = \mathbf{1}_{\{x = \tilde{x}\}} \Big( \mathcal{Q} (du, dr,d\theta) \delta_u (d\tilde{u}) \delta_{r}(d\tilde{r}) \delta_{\theta}(d\tilde{\theta})  \Big) \\ 
& \quad + \mathbf{1}_{\{x \in F, q(\tilde{x}, \tilde{v}) \in F, \tilde{x} \in D,  \|\tilde{v}\|  \geq 1, \|v\| \geq 1\}} (\mathcal{U} \otimes \Lambda_{x,\tilde{x}, \tilde{v}}) (du,  dr, d\theta, d\tilde{r},  d\tilde{\theta}) \delta_u(d\tilde{u}) , \nonumber \\
& \quad + \mathbf{1}_{\{x \ne \tilde{x}\}} \mathbf{1}_{\{x \in F, q(\tilde{x}, \tilde{v}) \in F, \tilde{x} \in D,  \|\tilde{v}\|  \geq 1, \|v\| \geq 1\}^c} (\mathcal{Q} \otimes \mathcal{Q}) (du, dr, d\theta, d\tilde{u},d\tilde{r}, d\tilde{\theta}), \nonumber
\end{align}
with $\Lambda_{x, \tilde{x}, \tilde{v}}$ given by Proposition \ref{CouplingPropGeneral}. We construct a coupling process $(X_s,V_s, \xtilde{s},\vtilde{s}, Z_s)_{s \geq 0}$ with the same definition as the one in the convex case, see Definition \ref{DefiCouplingProcess}, except that we consider $\Gamma$ defined by (\ref{EqLawGlobalCouplingExtension}) rather than by (\ref{EqLawGlobalCoupling}).
The statements of Lemmas \ref{LemmaMarginalCode} and \ref{Lemmataudefinitive} still hold. Indeed, the difference only relies on the law $\Gamma$. 

Lemma \ref{LemmaControlCalT} and \ref{LemmaIndepProcess} also hold with this new context, since those results do not rely on the convexity of the domain.

\subsection{Proof of Theorem \ref{MainTheorem} in the general setting}

We prove first a result on independent processes similar to Lemma \ref{LemmaIndepProcess}, and conclude the proof of Theorem \ref{MainTheorem} in the general framework of $C^2$ bounded domains.
Let $d_0, \kappa_0 > 0$ and $F, \mathcal{R} \subset \pD$ given by Lemma \ref{LemmaDefiningCriticalRegion}. In this subsection, we denote by $\kappa, L$ two positive constants depending only on $(D,r,C_0,n_0,\nu_0, \kappa_0,d_0)$ with $C_0$ given by (\ref{EqDefC0}) and $(n_0,\nu_0)$ given by Corollary \ref{corollaryEvans}. The values of $\kappa$ and $L$ are allowed to vary from line to line.

\begin{lemma}
\label{LemmaIndepProcessGlobal}
There exists $\kappa > 0$ such that if $(x,v), (\tilde{x},\tilde{v}) \in (D \times \RR^n) \cup \partial_+ G$ and $(X_t,V_t)_{t \geq 0}$, $(\xtilde{t},\vtilde{t})_{t \geq 0}$ are two independent free-transport processes with initial conditions $X_0 = x$, $V_{0-} = v$, $\xtilde{0} = \tilde{x}$, $\vtilde{0-} = \tilde{v}$, setting
$$ S = \inf \{t > 0, X_t \in F, \xtilde{t} \in D, q(\xtilde{t},\vtilde{t-}) \in F,  \|V_{t-}\| \wedge \|\vtilde{t-}\| \geq 1\}, $$
we have $$ \EE[ r(S) ] \leq \kappa \Big(1 + r\Big(\frac{d(D)}{\|v\|}\Big) + r\Big(\frac{d(D)}{\|\tilde{v}\|} \Big) \Big). $$
\end{lemma}

\begin{proof}
We introduce the sequence $(T_k)_{k \geq 0}$ defined by $T_0 = \zeta(x, v)$ and, for all $k \geq 0$, $T_{k+1} = \inf \{t > T_k$, $X_t \in \pD\}$, and the sequence  $(\tilde{T}_k)_{k \geq 0}$ defined by $\tilde{T}_0 = \zeta(\tilde{x},\tilde{v}),$ and for $k \geq 0$, $\tilde{T}_{k+1} = \inf \{t > \tilde{T}_k$, $\xtilde{t} \in \pD \}$. We also introduce the filtration $\mathcal{F}_t = \sigma((X_s,V_s,\xtilde{s},\vtilde{s})_{0 \leq s \leq t} )$.  We set $S_1 = \inf \{t \geq T_{n_0}, X_t \in \pD, \xtilde{t} \in D, \|V_{t-}\| \ne \|v\|, \|\vtilde{t-}\| \ne \|\tilde{v}\|, \|V_{t-}\| \wedge \|\vtilde{t-}\| \geq 1\}$.

\vspace{.5cm}

\textbf{Step 1.} We prove that $$\EE[r(S_1)] \leq \kappa \Big( 1 + r\Big(\frac{d(D)}{\|v\|} \Big) + r \Big( \frac{d(D)}{\|\tilde{v}\|} \Big) \Big).$$ 
By the strong Markov property, using Lemma \ref{LemmaIndepProcess}, which is, as already mentioned, still valid in the non-convex case,
\begin{align*}
 \EE[r(S_1 - T_{n_0})|\mathcal{F}_{T_{n_0}}] &\leq \kappa \Big(1 + r\Big(\frac{d(D)}{\|V_{T_{n_0}}\|}\Big) + r\Big(\frac{d(D)}{\|\vtilde{T_{n_0}}\|}\Big) \Big).
 \end{align*}
We then have, using Remark \ref{RmkPolynomialR},
\begin{align*}
\EE[r(S_1)] &\leq C \Big( \EE[r(T_{n_0})] + \EE \Big[ \EE[r(S_1 - T_{n_0})|\mathcal{F}_{T_{n_0}}] \Big] \Big) \\
&\leq \kappa \Big(\sum_{i=0}^{n_0-1} r(T_{i+1}-T_i) + r (T_0) + 1 + r \Big(\frac{d(D)}{\|V_{T_{n_0}}\|} \Big) + r \Big(\frac{d(D)}{\|\vtilde{T_{n_0}}\|} \Big) \Big) \\
&\leq \kappa \Big(1 + \sum_{i=0}^{n_0} \Big[r \Big(\frac{d(D)}{\|V_{T_i}\|} \Big) \Big] +  r \Big(\frac{d(D)}{\|v\|}\Big) + r \Big(\frac{d(D)}{\|\vtilde{T_{n_0}}\|} \Big) \Big) \\
&\leq \kappa \Big(1 + r \Big(\frac{d(D)}{\|v\|}\Big) + r \Big(\frac{d(D)}{\|\tilde{v}\|}\Big) \Big),
\end{align*}
since, as usual, for all $i \in \{0, \dots, n_0\}$, we have
$\|V_{T_i}\| = \|V_{T_i}\|\mathbf{1}_{\{\|V_{T_i}\| \ne \|v\|\}} + \|v\| \mathbf{1}_{\{\|V_{T_i}\| = \|v\|\}}$ with $\mathcal{L}(\|V_{T_i}\| | \|V_{T_i}\| \ne \|v\|) = h_R$.

\vspace{.5cm}

\textbf{Step 2.} In this step, we prove that there exists $c > 0$ such that, for all initial conditions $(x,v) \in \partial_+ G$, $(\tilde{x},\tilde{v}) \in D \times \RR^n$ with $\|v\| \wedge \|\tilde{v}\| \geq 1$,
\begin{align}
\label{EqTmpProbaS1}
 \PP(X_{S_1} \in F, q(\xtilde{S_1},\vtilde{S_1}) \in F) \geq c. 
 \end{align}
Set 
\begin{align*}
O_0 &:= \Big\{ \|V_{T_0}\| \ne \|V_{T_0-}\|, \dots, \|V_{T_{n_0-1}}\| \ne \|V_{T_{n_0-1}-}\| \Big\}, \\
\tilde{O}_0 &:= \Big\{ \|\vtilde{T_0}\| \ne \|\vtilde{T_0-}\|, \dots, \|\vtilde{T_{n_0-1}}\| \ne \|\vtilde{T_{n_0-1}-}\| \Big\},
\end{align*}
 and note that one has $\PP(O_0 \cap \tilde{O}_0) \geq \alpha_0^{2n_0}$. 
We also set 
\begin{align*}
O_1 &:= \Big\{X_{T_{n_0}} \in F, \|V_{T_{n_0}-}\| \geq 1,  \|X_{T_{i}} - X_{T_{i-1}}\|  \geq \delta_0  \text{ for all } i \in \{1, \dots, n_0\} \Big\}, \\
\tilde{O}_1 &:= \Big\{ \xtilde{\tilde{T}_{n_0}} \in F, \|\vtilde{\tilde{T}_{n_0}-}\| \geq 1, \|\xtilde{\tilde{T}_i} - \xtilde{\tilde{T}_{i-1}}\| \geq \delta_0 \text{ for all } i \in \{1, \dots, n_0\} \Big\}.
\end{align*}
We have, using that $\|V_{T_{n_0}-}\|$ is independent of the sequence $(X_{T_k})_{0 \leq k \leq n_0}$  and has law $h_R$ conditionally on $O_0$,
\begin{align*}
\PP(O_1 | O_0) &= \PP\Big(X_{T_{n_0}} \in F, \underset{i \in \{1, \dots, n_0\}}{\min}\|X_{T_i} - X_{T_{i-1}}\| \geq \delta_0 \Big| O_0 \Big) \PP(\|V_{T_{n_0}-}\| \geq 1 | O_0) \\
&\geq \nu_0 \kappa_0 \int_1^{\infty} h_R(r)dr, 
\end{align*} 
using Proposition \ref{corollaryEvans} and $\mathcal{H}(F) \geq \kappa_0$. Setting $c_0 = \nu_0 \kappa_0 \int_1^{\infty} h_R(r) dr > 0$, we obtain similarly that $\PP(\tilde{O}_1 | \tilde{O}_0) \geq c_0$. 

Moreover, we have 
$$ O_0 \cap O_1 \cap \tilde{O}_0 \cap \tilde{O}_1 \cap \Big\{ T_{n_0} \in (\tilde{T}_{n_0 - 1}, \tilde{T}_{n_0}) \Big\} \subset \Big\{S_1 = T_{n_0}, X_{S_1} \in F, q(\xtilde{S_1},\vtilde{S_1}) =  \xtilde{\tilde{T}_{n_0}}\in F \Big\}. $$

To prove (\ref{EqTmpProbaS1}), it thus suffices to show that there exists some $\kappa > 0$ such that
\begin{align}
\label{ProbaTn0}
\PP \Big(T_{n_0} \in (\tilde{T}_{n_0-1}, \tilde{T}_{n_0}) \Big| O_0 \cap \tilde{O}_0 \cap O_1 \cap \tilde{O}_1 \Big) \geq \kappa.
\end{align}
Since all the random variables $R_i = \|V_{T_i}\|$, $i \in \{0,\dots,n_0-1\}$, and $\tilde{R}_i = \|\vtilde{\tilde{T}_i}\|$, $i \in \{0, \dots, n_0-1\}$ are i.i.d. and $h_R$ distributed on $O_0 \cap O_1 \cap \tilde{O}_0 \cap \tilde{O}_1$, and since
\begin{align*}
\tilde{T}_{n_0 - 1} &= \frac{\|\xtilde{\tilde{T}_0} - \tilde{x}\|}{\|\tilde{v}\|} + \sum_{i = 0}^{n_0-2} \frac{\|\xtilde{\tilde{T}_{i+1}} - \xtilde{\tilde{T}_i}\|}{\tilde{R}_i}, \qquad T_{n_0} = \sum_{i = 0}^{n_0 - 1} \frac{\|X_{T_{i+1}} - X_{T_i}\|}{R_i}, \\
\tilde{T}_{n_0} &= \frac{\|\xtilde{\tilde{T}_0} - \tilde{x}\|}{\|\tilde{v}\|} + \sum_{i = 0}^{n_0-1} \frac{\|\xtilde{\tilde{T}_{i+1}} - \xtilde{\tilde{T}_i}\|}{\tilde{R}_i},
\end{align*}
we only need to prove that, for some $c_1' > 0$,
\begin{align}
\label{IneqInfProba}
 \underset{\begin{array}{l} \qquad \qquad \tilde{a} \in (0,d(D)) \\  a_0,\tilde{a}_0, \dots, a_{n_0-1},\tilde{a}_{n_0-1} \in (\delta_0,d(D)) \end{array} } {\inf}  \PP \Big( \frac{\tilde{a}}{\|\tilde{v}\|} + \sum_{i = 0}^{n_0-2} \frac{\tilde{a}_i}{\tilde{R}_i} \leq \sum_{i=0}^{n_0-1} \frac{a_i}{R_i} \leq \frac{\tilde{a}}{\|\tilde{v}\|} +  \sum_{i=0}^{n_0-1} \frac{\tilde{a}_i}{\tilde{R}_i}\Big) \geq c_1',
 \end{align}
with $(R_i)_{i = 0, \dots, n_0-1}$, $(\tilde{R}_i)_{i =0,\dots,n_0-1}$ independent and i.i.d. of law $h_R$.
By Hypothesis \ref{HypoM}, for all $0 \leq \epsilon_0 < \epsilon_1 \leq \delta_1$, $\int_{\epsilon_0}^{\epsilon_1}h_R(r)dr > 0$.

We claim that there exists $0 < \tilde{\theta}_1 < \tilde{\theta}_2 < \delta_1$, $0 < \theta_1 < \theta_2 < \delta_1$, $0 < \tilde{\theta}_3 < \delta_1$, such that
$$ d(D) \big( 1 + \frac{n_0 - 1}{\tilde{\theta}_1} \big) \leq \frac{n_0 \delta_0}{\theta_2} \leq \frac{n_0d(D)}{\theta_1} \leq \frac{\delta_0(n_0-1)}{\tilde{\theta}_2} + \frac{\delta_0}{\tilde{\theta_3}}. $$
Indeed, taking $\tilde{\theta}_1 = \frac{\delta_1}{2} \wedge \frac12$, $\theta_2 = \tilde{\theta}_1 \frac{\delta_0}{d(D)}$, we have $\tilde{\theta}_1, \theta_2 \in (0,\delta_1)$ (because $\delta_0 < d(D)$) and
$$ d(D) \big(1 + \frac{n_0 - 1}{\tilde{\theta}_1}\big) \leq d(D) \frac{n_0}{\tilde{\theta}_1} = \frac{n_0 \delta_0}{\theta_2}. $$
We set $\theta_1 = \frac{\theta_2}{2} \in (0,\theta_2),$ $\tilde{\theta}_2 = \frac{\tilde{\theta}_1 + \delta_1}{2} \in (\tilde{\theta}_1,\delta_1)$, and, choosing $\tilde{\theta}_3$ sufficiently small, we have $\tilde{\theta}_3 < \delta_1$ and 
$$ \frac{\delta_0(n_0-1)}{\tilde{\theta}_2} + \frac{\delta_0}{\tilde{\theta}_3} \geq \frac{n_0d(D)}{\theta_1}. $$ 

We have, for all $\tilde{a} \in (0,d(D))$, for all $a_i, \tilde{a}_i \in (\delta_0,d(D))$ with $i \in \{0,\dots,n_0-1\}$, recalling that $\|\tilde{v}\| \geq 1$,
\begin{align*}
& \PP \Big( \frac{\tilde{a}}{\|\tilde{v}\|} + \sum_{i = 0}^{n_0-2} \frac{\tilde{a}_i}{\tilde{R}_i} \leq \sum_{i=0}^{n_0-1} \frac{a_i}{R_i} \leq \frac{\tilde{a}}{\|\tilde{v}\|} + \sum_{i=0}^{n_0-1} \frac{\tilde{a}_i}{\tilde{R}_i}\Big)\\
&\qquad \geq \PP\Big(\text{ for all } i \in \{0, \dots, n_0-2\}, R_i \in (\theta_1, \theta_2), \tilde{R}_i \in (\tilde{\theta}_1, \tilde{\theta}_2), R_{n_0-1} \in (\theta_1,\theta_2), \tilde{R}_{n_0-1} \in (0,\tilde{\theta}_3)\Big) \\
&\qquad \geq \Big(\int_{\theta_1}^{\theta_2} h_R(r)dr \Big)^{n_0} \Big( \int_{\tilde{\theta}_1}^{\tilde{\theta}_2} h_R(r) dr \Big)^{n_0-1} \Big(\int_0^{\tilde{\theta}_3} h_R(r)dr \Big) > 0.
\end{align*}
This completes the proof of (\ref{IneqInfProba}) and thus the proof of (\ref{EqTmpProbaS1}). 

\vspace{.5cm}

\textbf{Step 3.} We set, for any stopping time $\tau$, $T^{\tau}_0 = \inf\{ t \geq \tau, X_t \in \pD\}$ and for all $k \geq 0$, $T^{\tau}_{k+1} = \inf \{ t > T^{\tau}_k, X_t \in \pD\}$. Note that $T_k = T^0_k$ for all $k \geq 0$. We introduce the sequence $(S_i)_{i \geq 0}$ defined by $S_0 = 0$, $S_1$ defined as in Step 1 and for all $k \geq 1$, $$S_{k+1} = \inf\{t \geq T^{S_k}_{n_0}, X_t \in \pD, \xtilde{t} \in D, \|V_{t-}\| \ne \|V_{S_k-}\|, \|\vtilde{t-}\| \ne \|\vtilde{S_k-}\|, \|V_{t-}\| \wedge \|\vtilde{t-}\| \geq 1\}.$$

We set, for all $k \geq 1$, 
\begin{align*}
B_k = \{X_{S_k} \in F, q(\xtilde{S_k},\vtilde{S_k-}) \in F\}.
\end{align*}

We plan to apply Lemma \ref{LemmaSumGeometric}. 
\begin{enumerate}[i)]
\item We set, for all $k \geq 0$, $\CG_k = \mathcal{F}_{S_{k+1}-}$, and for all $k \geq 1$, $\tau_k = S_{k+1} - S_1$ which is $\CG_k$-measurable and $E_k = B_{k+1} \in \CG_k$. We set $G = \inf\{k \geq 1, E_k \text{ is realized}\}$. 
\item We have, for all $k \geq 1$,
$$ \PP(E_k|\CG_{k-1}) = \PP(B_{k+1} | \mathcal{F}_{S_{k}-}) \geq c $$
by Step 2, using the strong Markov property and that $\|V_{S_k-}\| \wedge \|\vtilde{S_k-}\| \geq 1$, $X_{S_k} \in \pD$, $\xtilde{S_k} \in D$. Hence (\ref{HypoProbaEk}) holds. 
\item Using the strong Markov property and Step 1, we have, for all $k \geq 0$,
\begin{align*}
\EE[r(S_{k+1}-S_{k})|\mathcal{F}_{S_{k}-}] \leq \kappa \Big(1 + r \Big(\frac{d(D)}{\|V_{S_k-}\|} \Big) + r \Big(\frac{d(D)}{\|\vtilde{S_{k}-}\|} \Big) \Big) =: K_k.
\end{align*} 
For $k\geq 1$,
\begin{align*}
\EE[r(\tau_{k+1}-\tau_k)| \mathcal{G}_{k-1}] = \EE \Big[ K_{k+1} \Big| \mathcal{F}_{S_{k}-}\Big]
&\leq \kappa \EE \Big[1 + r\Big( \frac{d(D)}{\|V_{S_{k+1}-}\|} \Big) + r\Big( \frac{d(D)}{\|\vtilde{S_{k+1}-}\|} \Big) \Big| \mathcal{F}_{S_{k}-} \Big] \\
&\leq \kappa (1 + 2C_0).
\end{align*} 
We used that for all $k \geq 1$, $\mathcal{L} (\|V_{S_{k+1}-}\| | \mathcal{F}_{S_{k}-} ) = \mathcal{L} (\|\vtilde{S_{k+1}-}\| | \mathcal{F}_{S_{k}-} )  = h_R$ by definition of $(S_k)_{k \geq 0}$.
Note that $\tau_1 = S_2 - S_1$. We have, since $\|V_{S_1-}\| \wedge \|\vtilde{S_1-}\| \geq 1$,
$$ \EE[r(\tau_1)|\CG_0] = \EE[r(S_2 - S_1)|\mathcal{F}_{S_1-}] \leq \kappa \Big( 1 + r \Big(\frac{d(D)}{\|V_{S_1-}\|} \Big) + r \Big(\frac{d(D)}{\|\vtilde{S_1-}\|} \Big) \Big) \leq \kappa, $$
whence (\ref{IneqHypoLemmaSumGeom}).

\end{enumerate}

Setting $J= G+1$, we conclude by Lemma \ref{LemmaSumGeometric} that $\EE[r(S_J - S_1)|\mathcal{F}_{S_1-}] = \EE[r(\tau_G)|\mathcal{G}_0] \leq \kappa$, whence, by Step 1,
 $$\EE[r(S_{J})] \leq \kappa \Big(1 + r \Big(\frac{d(D)}{\|v\|} \Big) + r \Big(\frac{d(D)}{\|\tilde{v}\|} \Big) \Big).$$
Observe that, by definition of $J$, almost surely, $X_{S_J} \in F$, $\xtilde{S_J} \in D$, $q(\xtilde{S_J},\vtilde{S_J-}) \in F$ and $\|V_{S_J-}\| \wedge \|\vtilde{S_J-}\| \geq 1$, whence $S \leq S_J$.  
\end{proof}

We introduce some notations corresponding to Notation \ref{NotatConvexProof} in the general case.
\begin{notat}
\label{NotatGeneralProof}
 Let $(X_s,V_s, \xtilde{s},\vtilde{s}, Z_s)_{s \geq 0}$ a coupling process, see Definition \ref{DefiCouplingProcess} with $\Gamma$ given by (\ref{EqLawGlobalCouplingExtension}). We use the same sequences $(S_i,Q_i,\tilde{Q}_i)_{i \geq 1} $ as in the definition, as well as $(\tilde{Q}'_i)_{i \geq 1}$, and we recall that, for all $i \geq 1$,
% we write $(Q_i, \tilde{Q}'_i)$ for the couple of $\mathcal{Q}$-distributed random variables such that
$$ V_{S_i} = w(X_{S_i},V_{S_i-},Q_i) \mathbf{1}_{\{X_{S_i} \in \pD\}} +V_{S_i-}  \mathbf{1}_{\{X_{S_i} \not \in \pD\}}, $$
$$ \vtilde{S_i} = w(\xtilde{S_i},\vtilde{S_i-},\tilde{Q}'_i) \mathbf{1}_{\{\xtilde{S_i} \in \pD\}} + \vtilde{S_i-} \mathbf{1}_{\{\xtilde{S_i} \not \in \pD\}} .$$

\noindent a) We set $T_0 = 0$, $\tilde{T}_0 = 0$ and for $k \geq 0$,
$$T_{k+1} = \inf\{t > \tilde{T}_{k}, X_t \in \pD\}, \qquad \tilde{T}_{k+1} = \inf\{t > T_{k+1}, \xtilde{t} \in \pD\}. $$
For all $k \geq 1$, we have $Z_{T_k-} = \emptyset$ and $X_{T_k} \in \pD$ so $Z_{T_k} \ne \emptyset$ if $\xtilde{T_k} \not \in \pD$. We always have $Z_{\tilde{T}_k} = \emptyset$. 
For all $k \geq 1$, we write $(\underline{Q}_k, \underline{\tilde{Q}}_k) = (\underline{U}_k,\underline{R}_k,\underline{\Theta}_k, \underline{\tilde{U}}_k, \underline{\tilde{R}}_k, \underline{\tilde{\Theta}}_k)$ for the random vector such that
$$V_{T_k} = w(X_{T_k}, V_{T_k-} ,\underline{Q}_k), \quad \text{ and } \quad \vtilde{\tilde{T}_k} = w(\xtilde{\tilde{T}_k}, \vtilde{\tilde{T}_k-}, \underline{\tilde{Q}}_k).$$
Note that $(\underline{Q}_k, \underline{\tilde{Q}}_k)_{k \geq 1}$ is a subsequence of $(Q_i, \tilde{Q}_i')_{i \geq 1}$.

%For all $k \geq 1$, $T_k$ (resp. $\tilde{T}_k$) is the value $T$ (resp $\tilde{T}$) at the beginning of the $k$-th iteration of the Repeat loop (line \ref{line:dowhilelopp}). For all $k \geq 1$, we write $(\underline{U}_k,\underline{R}_k,\underline{\Theta}_k, \underline{\tilde{U}}_k, \underline{\tilde{R}}_k, \underline{\tilde{\Theta}}_k)$ for the random vector simulated at the beginning of the $k$-th iteration of the Repeat loop, see line \ref{line:simul} of the pseudo-code, so that for all $k \geq 1$,

\noindent b) For all $t \geq 0$, we set $$\mathcal{F}_t = \sigma \Big( (X_s,V_s,\xtilde{s},\vtilde{s},Z_s)_{0 \leq s \leq t}, (Q_i \mathbf{1}_{\{S_i \leq t\}})_{i \geq 1}, (\tilde{Q}_i \mathbf{1}_{\{S_i \leq t\}})_{i \geq 1} \Big).$$

\noindent c) We set 
$ \sigma_1 = \inf\{ t > 0, X_t = \xtilde{t} \in \pD, Z_{t-} = \emptyset, \|V_{t-} \| \ne \|V_0\|, \|\vtilde{t-}\| \ne \|\vtilde{0}\|\}. $

\noindent d) We set $\nu_0 = 0$ and for all $k \geq 0$,
$$ \nu_{k+1} = \inf \{n \geq \nu_k + 1, X_{T_n} \in F, \xtilde{T_n} \in D, q(\xtilde{T_n},\vtilde{T_n-}) \in F, \|V_{T_n-}\| \geq 1, \|\vtilde{T_n-}\| \geq 1 \}.$$
\end{notat} 
The only difference with Notation \ref{NotatConvexProof} is that Definition \ref{DefiCouplingProcess} uses (\ref{EqLawGlobalCouplingExtension}) rather than (\ref{EqLawGlobalCoupling}), and that the sequence $(\nu_k)_{k \geq 1}$ has been slightly changed. 
We next update Lemma \ref{LemmaTnu}.
\begin{lemma}
\label{LemmaTnuExtension}
%i) (Pseudo-Markov property) For all $k \geq 0$, conditionally on the $\mathcal{F}_{T_k-}$-measurable vector $(X_{T_k-},\xtilde{T_k-},V_{T_k-},\vtilde{T_k-})$, the process $(X_{T_k +t}, V_{T_k+t}, \xtilde{T_k+t}, \vtilde{T_k+t})_{t \in [0,T_{k+1}-T_k)}$ is independent of $\mathcal{F}_{T_k-}$. 
There exist three constants $\kappa, L, c > 0$ such that the following holds.
\begin{enumerate}[i)]
\item For all $m \geq 1$, $$\mathbf{1}_{\{T_{\nu_m} < \sigma_1\}} \EE[r(T_{\nu_{m + 1}} \wedge \sigma_1 - T_{\nu_m})| \mathcal{F}_{T_{\nu_m}-}] \leq L. $$
\item  $\EE[r(T_{\nu_1} \wedge \sigma_1)] \leq  \kappa \Big(1 + \EE \Big[r\Big(\frac{d(D)}{\|V_{0}\|}\Big) + r\Big(\frac{d(D)}{\|\vtilde{0}\|}\Big) \Big] \Big). $ 

\item For all $m \geq 1$, setting 
$$A_m = \{\underline{U}_{\nu_m} \leq \alpha_0,  X_{T_{\nu_{m} +1}}= \xtilde{T_{\nu_{m} + 1}}, T_{\nu_{m}+1} = T_{\nu_m} + \zeta(X_{T_{\nu_m}}, V_{T_{\nu_m}})\}, $$
we have $$\PP \Big(A_m \Big| \mathcal{F}_{T_{\nu_m}-} \Big) \geq c, $$
and  $A_m \subset \{\sigma_1 \leq T_{\nu_m + 1}\}$ outside a $\PP$-null set.
\end{enumerate}
\end{lemma}

\begin{proof}
The proof is the same as the one of Lemma \ref{LemmaTnu}, using Lemma \ref{LemmaIndepProcessGlobal}, Proposition \ref{CouplingPropGeneral}, Notation \ref{NotatGeneralProof}, Equation (\ref{EqLawGlobalCouplingExtension}) instead of Lemma \ref{LemmaIndepProcess}, Propositon \ref{CouplingProp}, Notation \ref{NotatConvexProof}, Equation (\ref{EqLawGlobalCoupling}), and that Lemma \ref{LemmaControlCalT} still holds when using Definition \ref{DefiCouplingProcess} with (\ref{EqLawGlobalCouplingExtension}) instead of (\ref{EqLawGlobalCoupling}).
\end{proof}

\begin{proof}[Proof of Theorem 2 in the general setting.]
The proof is the same as the one in the convex case, using Lemma \ref{LemmaTnuExtension} instead of Lemma \ref{LemmaTnu}, Notation \ref{NotatGeneralProof} instead of Notation \ref{NotatConvexProof} and that Lemmas \ref{LemmaMarginalCode} and \ref{Lemmataudefinitive} hold when using Definition \ref{DefiCouplingProcess} with (\ref{EqLawGlobalCouplingExtension}) instead of (\ref{EqLawGlobalCoupling}).
\end{proof}

\bibliographystyle{alpha}
\bibliography{biblio}
\end{document}